\documentclass[11pt,a4paper]{article}

\usepackage{amsmath,amsfonts,amssymb,amsthm}
\usepackage{a4wide}
\usepackage{enumerate}
\usepackage{filecontents} 

\usepackage{tikz} 
\usetikzlibrary{calc}
\usepackage{hyperref}
\hypersetup{colorlinks=true,citecolor=blue,linkcolor=blue,urlcolor=blue}

\let\epsilon=\varepsilon

\newcommand{\LSM}{\ensuremath{\mathcal{LSM}}}

\newcommand{\ba}{\mathbf{a}}

\newcommand{\bu}{\mathbf{u}}
\newcommand{\bv}{\mathbf{v}}
\newcommand{\bw}{\mathbf{w}}
\newcommand{\bx}{\mathbf{x}}
\newcommand{\by}{\mathbf{y}}
\newcommand{\bz}{\mathbf{z}}

\newcommand{\calA}{\mathcal{A}}
\newcommand{\calB}{\mathcal{B}}
\newcommand{\calC}{\mathcal{C}}
\newcommand{\calE}{\mathcal{E}}
\newcommand{\calF}{\mathcal{F}}
\newcommand{\calG}{\mathcal{G}}

\newcommand{\calL}{\mathcal{L}}
\newcommand{\calM}{\mathcal{M}}
\newcommand{\calS}{\mathcal{S}}

\newcommand{\bzero}{\mathbf{0}}
\newcommand{\bone}{\mathbf{1}}

\newcommand{\Bools}{\{0,1\}}
\newcommand{\PosReals}{\mathbb{R}_{> 0}}
\newcommand{\NonNegReals}{\mathbb{R}_{\geq 0}}
\newcommand{\allReals}{\mathbb{R}}
\newcommand{\nats}{\mathbb{N}}

\newcommand{\fclone}[1]{\langle{#1}\rangle}  
\newcommand{\ofclone}[1]{\langle{#1}\rangle_{\omega}}
\newcommand{\ppso}{$\text{pps}_{\omega}$}

\newcommand{\calP}{\ensuremath{\mathcal{P}}}
\newcommand{\calPN}{\ensuremath{\mathcal{PN}}}
\newcommand{\SD}{\ensuremath{\mathcal{SD}}}
\newcommand{\SDP}{\ensuremath{\mathcal{SDP}}} 
\newcommand{\Parev}{\ensuremath{\mathcal{P\!AR}_{\mathrm{ev}}}}
\newcommand{\Mon}{\ensuremath{\mathcal{MON\!}}}

\newcommand{\wt}[1]{|{#1}|}
\newcommand{\fhat}{\widehat{f}}
\newcommand{\fghat}{\widehat{fg}}
\newcommand{\ghat}{\widehat{g}}
\newcommand{\hhat}{\widehat{h}}

\newcommand{\FerroHyperIsing}{\mathsf{FerroHyperIsing}}
\newcommand{\FerroIsing}{\mathsf{FerroIsing}}
\newcommand{\AntiFerroIsing}{\mathsf{AntiFerroIsing}}

\newcommand{\hIsing}[2]{I_{#1}^{#2}}
\newcommand{\hIsinghat}[2]{\widehat{I}_{#1}^{#2}}
\newcommand{\paritybasis}[1][k]{{\mathsf{Par'}_{#1}}} 
\newcommand{\parityuptokbasis}{{\mathsf{Par'^{ev}}_{\leq k}}}
\newcommand{\parityk}[1][\lambda]{\mathsf{Par}_k^{#1}}
\newcommand{\parity}[2]{\mathsf{Par}_{#1}^{#2}}
\newcommand{\parityhatk}[1][\lambda]{\smash{\widehat{\mathsf{Par}}}_k^{#1}}
\newcommand{\parityhat}[2]{\smash{\widehat{\mathsf{Par}}}_{#1}^{#2}}
\newcommand{\ForceOdd}[1]{{\oplus}_{#1}^{\mathrm{odd}}}

\newcommand{\Hferro}{\mathcal{H}_{\mathrm{ferro}}}
\newcommand{\Iferro}{\mathcal{I}_{\mathrm{ferro}}}
\newcommand{\Ianti}{\mathcal{I}_{\mathrm{anti}}}

\newcommand{\EQ}{\mathrm{EQ}}

\newcommand{\XOR}{\mathrm{XOR}}
\newcommand{\mtwo}[4]{\left(\begin{smallmatrix}{#1} & {#2}\\ {#3} & {#4}\end{smallmatrix}\right)}

\newcommand{\Zivny}{{\v{Z}}ivn{\'y}}
\newcommand{\pref}[1]{(\ref{#1})}

\newtheorem{theorem}{Theorem}
\newtheorem{lemma}[theorem]{Lemma}
\newtheorem{corollary}[theorem]{Corollary}
\theoremstyle{definition}
\newtheorem{definition}[theorem]{Definition}

\title{Functional Clones and Expressibility of Partition Functions\footnote{Part of this work 
was done while the authors were visiting the Simons Institute for the Theory of Computing.}}

\author{Andrei Bulatov\thanks{Simon Fraser University, Burnaby, Canada.  Supported by an NSERC Discovery Grant.} \and
Leslie Ann Goldberg\thanks{University of Oxford, UK.
The research leading to these results has received funding from 
the European Research Council under the European Union's Seventh Framework Programme (FP7/2007--2013) ERC grant agreement no.\ 334828 (Goldberg and Richerby; this grant also supported visits by Bulatov and Jerrum) and Horizon 2020 research and innovation programme (grant
agreement no.~714532, \Zivny). The paper 
reflects only the authors' views and not the views of the ERC or the European Commission. The European Union is not liable for any use that may be made of the information contained therein.}
 \and Mark Jerrum\thanks{Queen Mary, University of London, UK
 Supported by EPSRC grant EP/N004221/1, 
an NSERC Discovery Grant and the SFU Distinguished Lecturers program.}
 \and David Richerby\footnotemark[3]
  \and Stanislav \Zivny\footnotemark[3]\ \thanks{Supported by a Royal Society University Research Fellowship.
}}

\date{28th April, 2017}

\begin{document}
\maketitle{}

\begin{abstract} 
We study functional clones, which are sets of non-negative pseudo-Boolean functions (functions $\{0,1\}^k\to\NonNegReals$) closed under (essentially) multiplication, summation and limits. Functional clones naturally form a lattice under set inclusion and are closely related to counting Constraint Satisfaction Problems (CSPs). We identify a sublattice of interesting functional clones and investigate the relationships and properties of the functional clones in this sublattice.
\end{abstract}

\section{Introduction}

There is a considerable literature on the topic of relational clones,
also called co-clones.  These are sets of relations on a finite domain~$D$
that are closed under certain operations, the most interesting being 
conjunction of two relations and existential quantification over a variable.
(Other closure operations, such as introduction of ``fictitious arguments'',
are technically but not conceptually important.)  In this paper we focus on
the Boolean domain, and presently we will assume that $D=\{0,1\}$.  
It is well known that in the Boolean case, the 
set of relational clones is countably infinite and forms a lattice under set 
inclusion.  The lattice has been explicitly described by Post~\cite{Post}.

It seems natural to widen this study to other algebraic structures.  
Functional clones were introduced formally by  Bulatov, Dyer, Goldberg,
Jerrum and McQuillan~\cite{LSM}, with the motivation of studying the
computational complexity of counting constraint satisfaction problems.   
A functional clone is a set of multivariate 
functions from a finite domain~$D$ to a semiring~$R$ that is closed under
multiplication, summing over a variable and (optionally) taking a limit of
a sequence of functions.  (Other operations are needed for technical completeness.
Formal definitions are given in the following subsection.) 
In this paper, we focus attention on the case $D=\{0,1\}$ and $S=(\NonNegReals,\times,+)$. 
We reconsider functional clones as objects 
of interest in their own right, though the results we prove may yield insights
in other areas. 

There are at least three motivations for the current investigation.

The first, as indicated above, is intrinsic interest.  
Post's lattice of relational clones has a fascinating structure.
There is a Galois connection between sets of relations on~$D$ and 
sets of operations on~$D$ which establishes a beautiful duality between 
relational clones and clones of operations.  Remarkably, the closure operator 
defined by the Galois connection exactly agrees with the one described
earlier in terms of conjunction of relations and existential quantification
over variables~\cite{Geiger}.  

The situation with functional clones 
is not quite so clean.  There is apparently no Galois connection 
between sets of functions on~$D$ and sets of (somehow appropriately generalised)
operations on~$D$ that captures the closure under multiplication 
and summation described above.  Moreover, the lattice of 
functional clones has the cardinality of the continuum (or even larger, 
depending on precise definitions) and there seems to be no hope of providing a
complete description of it.    
Still, it is interesting to map out some of the main features of the lattice,
to identify maximal functional clones, to identify sublattices of functional 
clones satisfying additional properties, to find alternative characterisations 
of certain functional clones in terms of generating sets or
Fourier coefficients, etc.  As a contribution in this direction we identify 
(Figure~\ref{fig:lattice}) a sublattice of what seem to us 
to be interesting functional clones.

The second motivation, hinted at earlier,
is the desire to understand the computational complexity of certain counting problems.
A classical (decision) Constraint Satisfaction Problem (CSP) is a generalised satisfiability 
problem.  Instead of restricting clauses to being disjunctions of literals, as in 
standard satisfiability problems, we allow arbitrary relations between variables 
chosen from a specified set or ``language'' of relations~$\Gamma$.  We are interested
in how the computational complexity of a CSP varies as a function of~$\Gamma$.   
Clearly, extending the language~$\Gamma$ may increase the complexity of the
corresponding CSP{}.  It transpires that the complexity of a CSP depends not
on the fine structure of $\Gamma$, but only on the  
relational clone generated by~$\Gamma$.  This observation makes feasible the 
detailed exploration of the complexity of classical CSPs.

A counting Constraint Satisfaction Problem (\#CSP) asks for the number of 
satisfying assignments to a CSP{}.  In their weighted form, counting CSPs
are general enough to express many partition functions occurring in statistical physics.  
Just as with classical decision CSPs, the complexity of a counting CSP is determined by 
the functional clone generated by the constraint language, which now consists of 
functions taking, say, non-negative real values.  Functional clones were introduced 
in~\cite{LSM} precisely as a tool for studying the complexity of \#CSPs.
Referring to Figure~\ref{fig:lattice}, the equality at the bottom of the lattice
expresses the equivalence between (on the left) the partition function 
of the ferromagnetic Ising model and (on the right) the so-called high-temperature 
expansion in terms of even subgraphs.  Counting CSPs at this level of the lattice 
can be approximated in polynomial time by an algorithm that 
exploits this equivalence~\cite{JS1993:Ising}.  Moving up the lattice, perhaps the 
most intriguing functional clone from the complexity point of view is $\calM$ 
which includes the counting CSPs that would become feasible to approximate if 
we were to discover a polynomial time approximation algorithm for counting matchings in
a general (non-bipartite) graph.

A third motivation for our study is provided by the connection between functional 
clones and topics in statistical physics and machine learning. 
Many models in statistical 
physics are ``spin models'' defined by a graph
or more generally a hypergraph on $n$ vertices.  
To each vertex is associated a variable taking on values 
from a set of ``spins'' which, in our case, is finite.   A configuration of the system 
is an assignment of spins to the $n$~variables.  The edges of the graph or hypergraph
specify local interactions between spins.  These local interactions define a probability 
distribution on the set of all configurations.  Take for example the Ising model,
which is characterised by having just two spins.
An instance of the Ising model is specified by an undirected graph;
in other words, there are just pairwise interactions between spins.
(Refer to Section~\ref{sec:study} for details.)
One question we may ask is:  which $k$-way interactions may be induced 
in such a model?  More precisely, 
what are the possible marginal distributions that may
be observed on some $k$-subset of the vertex variables?  
This question is (modulo the normalising factor
for the probability distribution in question) 
precisely a question about functional clones.  The 
possible marginal distributions are the $k$-ary functions in the clone generated by the 
local pairwise interactions.

In the case of the antiferromagnetic Ising model, where the pairwise interactions 
favour unlike spins, the answer is given by Theorem~\ref{thm:SD-Ising}:  
the possible marginal distributions
are precisely those that are ``self-dual'', i.e.,
invariant under exchange of 0 and~1.  (It is 
clear that invariance under exchange of 0 and~1 is necessary;  
the point is that it is sufficient.)  This result 
has an implication for the expressive power 
of Boltzmann machines in machine learning~\cite{AHS}.  
Specifically, if the bias parameters of the units 
are all zero, then the distributions realisable at the visible units are precisely those that
are self-dual.  Note that this is an expressibility result, in the spirit of Le Roux 
and Bengio~\cite{LeRouxBengio},
and says nothing about the feasibility of learning the distributions in question from examples.

The analogous question in the ferromagnetic case is seemingly harder.  The three-variable 
marginals of a ferromagnetic Ising model can be described:  
they are the (normalised) functions of arity~3 in the functional clone associated with 
the Ising model, and are given in Theorem~\ref{thm:collapse1}.  
Already at arity~4 the elements of the clone become 
hard to describe.   Indeed, it is consistent with our current knowledge that 
membership in this clone is undecidable, even for functions of some fixed arity greater than three. 

Finally, there is a connection between functional clones and the idea of ``universal models''
in statistical physics proposed by De las Cuevas and Cubitt~\cite{DlCC}.  In a sense, 
functional clones formalise De las Cuevas and Cubitt's notion of ``closure''.  A spin model is
``universal'' in their sense if (very roughly) the functional clone generated by the 
model is the one at the top of the clone lattice, namely~$\calB$, that contains all functions.
They identify the planar antiferromagnetic Ising model with external fields
as an example of a universal model.

As we already noted,  the antiferromagnetic Ising model generates the clone~$\SD$
of self-dual functions.
Adding an external field takes us outside of~$\SD$.  Now, according to 
Lemma~\ref{lem:SD-maximal}, the clone $\SD$ is ``maximal'', from which we deduce that
the antiferromagnetic Ising model with fields generates $\calB$, i.e., is universal
in our sense.  Note, however, that our framework does not incorporate the 
notion of planarity, and in any case our closures do not exactly correspond to
those of De las Cuevas and Cubitt.  However, the clone lattice gives a more nuanced 
account of the expressive power of various spin models than simple universality.
For more on the expressive power of spin systems and their computational complexity,
see Goldberg and Jerrum~\cite{PNAS} and 
Chen, Dyer, Goldberg, Jerrum, Lu, McQuillan and Richerby~\cite{ApproxCSP}.
\subsection{Functional Clones}
\label{sec:fclones}

For every non-negative integer~$k$, 
let  $\calB_k$ be the set of all arity-$k$ non-negative pseudo-Boolean functions
(i.e., the set of all
functions $\{0,1\}^k\to \NonNegReals$). 
Let $\calB$ be the set of all 
non-negative pseudo-Boolean functions
(of all arities), given by
 $\calB = \calB_0 \cup \calB_1 \cup \calB_2 \cup \cdots$.
Given a  function $f\in \calB_k$ and a permutation~$\pi$ of~$\{1,\ldots,k\}$, we write
$f^\pi$ for the function that maps $(x_1, \dots, x_k)\in\{0,1\}^k$ to
$f(x_{\pi(1)}, \dots, x_{\pi(k)})$.   
Functional clones are 
subsets of $\calB$ that are closed under certain operations. We start by defining the operations.
Consider a set $\calF \subseteq \calB$. 
\begin{itemize}
\item $\calF$ is \emph{closed under  the
introduction of fictitious arguments}
if, for every  $k\geq 0$ and every
$k$-ary function $f\in \calF$,  
the $(k+1)$-ary function $g$ defined
by $g(x_1,\ldots,x_{k+1})=f(x_1,\ldots,x_k)$ is also in $\calF$.
\item  $\calF$ is 
\emph{closed under permuting arguments} if,
for every
$k\geq 1$, every  $k$-ary function $f\in \calF$ and every
permutation $\pi$ of $\{1,\ldots,k\}$,
the   function $f^{\pi}$  
is also in $\calF$.
\item   $\calF$ is \emph{closed under product} if,
for every $k\geq 0$, every $k$-ary function $f\in \calF$ and every
$k$-ary function $g\in \calF$,
the function $h$ defined by $h(x_1,\ldots,x_k) = f(x_1,\ldots,x_k)\, g(x_1,\ldots,x_k)$
is also in $\calF$.
\item $\calF$ is \emph{closed under summation} if,
for every $k\geq 1$ and every $k$-ary function $f \in \calF$, the
$(k-1)$-ary function $g$ defined by 
$g(x_1,\ldots,x_{k-1}) = \sum_{x_k\in \{0,1\}} f(x_1,\ldots,x_{k})  $
is also in~$\calF$.
\end{itemize}

Functional clones are
defined in 
 \cite[Section 2]{LSM}. 
The definition that we give here is equivalent to the one
in \cite{LSM}, but is  more suited to the setting of this paper.
 Let~$\EQ$ be the binary equality function, which is the function in $\calB_2$
defined by 
$\EQ(0,0)=\EQ(1,1)=1$, and $\EQ(0,1)=\EQ(1,0)=0$.
Suppose that $\calF\subseteq \calB$ is  a  set of functions.
The \emph{functional clone} $\fclone{\calF}$ is defined
to be the closure of $\calF \cup \{\EQ\}$ under 
the introduction of fictitious arguments, permuting arguments, product and summation.

Bulatov et al.~\cite[Proof of Lemma 2.1]{LSM} show\footnote{\label{footone} Technically, the proof of Lemma~2.1 of~\cite{LSM} just shows that  the closure of $\mathcal{A}(\calF)$ 
under product and summation is the same as the closure of $\prod(\calF)$ under summation.
That is, to produce the closure of $\calA(\calF)$ under product and summation
it suffices to first close $\calA(\calF)$ under product and then close 
the resulting set under summation.
However, it is easy to show 
that $\prod(\calF)$ is closed under 
the introduction of fictitious arguments and
permuting arguments,
and so is the closure of $\prod(\calF)$ under summation, so
without loss of generality, the three closures can be done in order: first close $\calF \cup \{\EQ\}$ under 
the introduction of fictitious arguments and permuting arguments, then close under product, then close under summation.} that the
set $\fclone{\calF}$ is  unchanged if the order of closure is restricted in the following way. 
Let $\calA(\calF)$ be the closure of $\calF \cup \{\EQ\}$ under 
the introduction of fictitious arguments and
permuting arguments. Let $\prod(\calF)$ be the closure of $\cal{A}(\calF)$ under product.
Then $\fclone{\calF}$ is the closure of $\prod(\calF)$ under summation.
In the paper, we will use the fact that the order of closure can be restricted in this way.
In particular, the definition of $\calA(\calF)$ will be used.

The reason for defining functional clones is that they are closely connected to
counting Constraint Satisfaction Problems (CSPs).
Every function in $\fclone{\calF}$ can be represented by a 
\emph{pps-formula} (``primitive product summation formula''), which is a summation of
a product of atomic formulas representing functions in $\calA(\calF)$.\footnote{There is one difference 
between pps-formulas as defined here, and pps-formulas as defined in~\cite{LSM}, but it is not important.
Consider an arity-$k$ function $f$.
Clearly, the arity-$(k-1)$ function defined by $g(x_1,\ldots,x_{k-1}) = f(x_1,\ldots,x_{k-1},x_{k-1})$
is in $\fclone{\{f\}}$ since 
$g(x_1,\ldots,x_{k-1}) = \sum_{x_k \in \{0,1\}} f(x_1,\ldots,x_k) \EQ(x_{k-1},x_k)$.
The function $g$ is not in the set $\cal{A}(\calF)$.
Nevertheless, Bulatov et al.~\cite{LSM} view the formula $\phi_g$ that represents~$g$
as an ``atomic formula'' since they allow repeated arguments.
For us, the formula~$\phi_g$ is not atomic, but this makes no difference, since $\EQ$ is in all functional clones, so our functional clones are exactly the same as those of~\cite{LSM}.} 
The pps-formula can be viewed as the input to a counting 
CSP whose output is the value of the function.
For example, consider the function $\XOR\in \calB_2$
defined by
$\XOR(0,0)=\XOR(1,1)=0$ and $\XOR(0,1)=\XOR(1,0)=1$.
Let $h$ be the function in $\calB_3$ defined by
$h(1,1,0) = h(0,0,1)=1$ and $h(x_1,x_2,x_3) =0$ for any $(x_1,x_2,x_3) \not\in \{(1,1,0),(0,0,1)\}$.
Let $\bx$ denote the tuple $(x_1,x_2,x_3,x_4)$.
It is easy to see that $h$ is in $\fclone{\{\XOR\}}$ since the functions
$f_{i,j}( \bx) = \XOR(x_i,x_j)$ are in $\mathcal{A}(\calF)$ for any distinct $i$ and $j$ in $\{1,2,3,4\}$
and the function 
 $g( \bx) = f_{1,4}( \bx) f_{2,4}( \bx)f_{1,3}( \bx)$ is in $\prod(\calF)$.
Finally, 
$h(x_1,x_2,x_3) = \sum_{x_4 \in \{0,1\}} g( \bx)$.
Now, for distinct $i$ and $j$ in $\{1,2,3,4\}$,
let $\phi_{i,j}(v_1,v_2,v_3,v_4)$ be an atomic formula representing the function $f_{i,j}$.
The function $g$ can be represented by the formula
$$\phi_g(v_1,v_2,v_3,v_4) =  \phi_{1,4}(v_1,v_2,v_3,v_4)\, \phi_{2,4}(v_1,v_2,v_3,v_4)\, \phi_{1,3}(v_1,v_2,v_3,v_4).$$
This formula can be viewed as a CSP with variables $\{v_1,v_2,v_3,v_4\}$
and three $\XOR$ constraints.
Finally, the function $h$ can be represented by the pps-formula
$\phi_h(v_1,v_2,v_3) = \sum_{v_4} \phi_g(v_1,v_2,v_3,v_4)$.

In order to study approximate counting CSPs it is necessary to go beyond
functional clones by also allowing closure under limits.
Given functions $f$ and $f'$ in $\calB_k$,
we write $\|f-f'\|_\infty$ for
the $L$-infinity distance between~$f$ and~$f'$, which is given by
$\|f-f'\|_\infty=\max_{\bx\in\Bools^k}|f(\bx)-f'(\bx)| $.
We say  that a $k$-ary function $f$ is a \emph{limit} of
a set $\calF\subseteq \calB$
if there is some finite $S_f\subseteq\calF$ such that, for every $\epsilon>0$, there
is  a $k$-ary  function $f_\epsilon \in \fclone{S_f}$ such that 
$\|f-f_\epsilon\|_\infty < \epsilon$.
We say that $\calF$ is  \emph{closed under limits}
if, for every function $f$ that is a limit of $\calF$,
$f\in \calF$.
The  \emph{$\omega$}-clone $\ofclone{\calF}$ is defined to be the 
closure of $\calF \cup \{\EQ\}$ under the introduction of fictitious arguments, permuting arguments, product, summation, and limits.
In~\cite{LSM},  the set $\ofclone{\calF}$ is referred to as the 
``\ppso{}-definable functional clone generated by~$\calF$''.
Bulatov et al.~\cite[Lemma 2.2]{LSM} show that  this set is unchanged if the order of
closure is restricted so   $\ofclone{\calF}$ is the closure of $\fclone{\calF}$ under limits.\footnote{Technically, the proof of Lemma 2.2 of~\cite{LSM} just shows that the
closure of $\calA(\calF)$ under product, summation and limits is the same as the closure of
$\fclone{\calF}$ under limits. However, it is easy to see that the closure of
$\fclone{\calF}$ under limits is closed under
the introduction of fictitious arguments and permuting arguments.} 

The following lemma is straightforward, given 
that $\fclone{\calF}$ and $\ofclone{\calF}$ are defined by
closing  a set (the set $\calF \cup \{\EQ\}$) using various operations.
Nevertheless, we state the lemma here for future use. The lemma combines
Lemmas 2.1 and 2.2 of~\cite{LSM}. (In that paper, the lemma was non-trivial, since
the order of the closure operators was restricted.)

\begin{lemma}
\label{lem:transitive}
\label{lem:clone-cup}
Suppose $\calF\subseteq \calB$.
If $g\in\fclone\calF$ then $\fclone{\calF \cup \{g\}}=\fclone{\calF}$.  
If $g$ is a limit of $\fclone{\calF}$ and $h$ is a limit of $\fclone{\calF \cup \{g\}}$
then $h$ is a limit of $\fclone{\calF}$. Equivalently,
if $g\in\ofclone\calF$ then $\ofclone{\calF \cup \{g\}}=\ofclone{\calF}$.
\end{lemma}

 \subsection{Lattices}
 
A \emph{lattice} is a set $L$ equipped with two
commutative, associative binary operations $\vee$ (\emph{join}) and
$\wedge$ (\emph{meet}) with the absorption property:
$a\vee(a\wedge b)=a$ and $a\wedge(a\vee b)=a$ for all $a,b\in L$.   
The lattice operations $\vee$
and $\wedge$ induce a partial order on $L$ as follows: for $a,b\in L$,
$a\le b$ if and only if $b=a\vee b$ (or, equivalently, $a=a\wedge b$).
It is easy to see that, for any $a,b\in L$, the elements $a\vee b$ and
$a\wedge b$ are the least upper bound and greatest lower bound of $a$
and~$b$, with respect to the order $\le$.  In other words, for any $c$
such that $a\le c$ and $b\le c$ it holds that $a\vee b\le c$, and for
any $d$ such that $d\le a$ and $d\le b$ it holds that
$d\le a\wedge b$. Conversely, if a set~$L$ has a partial order~$\le$
such that any pair of elements has a least upper bound and a greatest
lower bound, then it can be converted into a lattice by defining the
operations of join and meet as the least upper bound and the greatest
lower bound respectively.  
A subset $L'\subseteq L$ is called a \emph{sublattice} if for
all $a,b\in L'$, $a\vee b$ and $a\wedge b$ belong to $L'$. Note that
$\vee$ and $\wedge$ here are the operations of $L$.

\subsection{Lattices of functional clones}\label{sec:lfp}

Let $\calL_f$ and $\calL_\omega$ denote the set of all functional clones and
all $\omega$-clones, respectively, ordered with respect to
set inclusion. Then, for any two functional clones (or $\omega$-clones)
$\calF$ and~$\calG$, the least upper bound and the greatest lower
bound are given by $\fclone{\calF\cup\calG}$ (resp.,
$\ofclone{\calF\cup\calG}$) and $\calF\cap\calG$ (in both
cases). Therefore $\calL_f$ and $\calL_\omega$ can be viewed as
lattices with operations of join and meet
\begin{alignat*}{3}
    \calF\vee_f\calG &= \fclone{\calF\cup\calG}\,,\qquad
        &\calF\wedge_f\calG &= \calF\cap\calG 
        \qquad &&\text{for }\calL_f\,,\\
    \calF\vee_\omega\calG &= \ofclone{\calF\cup\calG}\,,\qquad 
        &\calF\wedge_\omega\calG &= \calF\cap\calG
        \qquad &&\text{for } \calL_\omega\,.
\end{alignat*}
Since we are mostly concerned with $\omega$-clones, we will omit the
subscripts of $\vee_\omega$ and~$\wedge_\omega$.

As we will show in Theorem~\ref{thm:cardinality}, the lattices $\calL_f$ and $\calL_\omega$ are quite
large, having cardinality~$\beth_2=2^{2^{\aleph_0}}$.  Therefore we will focus on the
most interesting and important $\omega$-clones.

\begin{definition}
    An $\omega$-clone $\calF$ is \emph{maximal} in an
    $\omega$-clone~$\calG$ if $\calF\subseteq\calG$ and there is no
    $\omega$-clone~$\calC$ such that $\calF \subset \calC \subset
    \calG$.
\end{definition}

It is easily seen that $\calF$ is maximal in~$\calG$ if and only if,
for any function $g\in\calG\setminus\calF$, $\ofclone{\calF \cup
  \{g\}} = \calG$.

\section{Notation and the clones that we study}
\label{sec:study}

We denote tuples in $\Bools^k$ by boldface letters.  
We use the notation $\wt{\bx}$ to denote 
the Hamming weight of~$\bx$.  
The symbols $\bzero$ and~$\bone$ are used to denote the
all-zeroes and all-ones tuple of arity appropriate to the context.
$\overline{\bx}$ is the bitwise complement of~$\bx$. We define
$[k]=\{1,\ldots,k\}$.

Recall the  function $f^{\pi}$ from Section~\ref{sec:fclones}.
We say that an arity-$k$ function~$f$ is
  is \emph{symmetric} if, for all permutations $\pi$ of~$[k]$,
$f=f^{\pi}$.
We often write symmetric
$k$-ary functions as $f=[f_0, \dots, f_k]$, where $f_i$~is the value
of~$f$ on arguments of Hamming weight~$i$.
Using this notation, the function~$\EQ$ can be written as $\EQ=[1,0,1]$.
We make use of the following unary functions:  $\delta_0 = [1,0]$ and
$\delta_1 = [0,1]$.

\begin{definition}
\label{def:ft}
The \emph{Fourier transform} of a function $f\colon \Bools^k\to\NonNegReals$
is the function $\fhat\colon\Bools^k\to\allReals$ defined by
\begin{equation*}
    \fhat(\bx) = \frac1{2^k}\sum_{\bw\in\Bools^k}
    (-1)^{\wt{\bw\land\bx}} f(\bw)\,.
\end{equation*}
\end{definition}
Note that, although we only consider functions whose range is the
nonnegative reals, the Fourier transform of such a function may have
negative numbers in its range.
Readers who are familiar with the holant framework~\cite{CLX2008:Fib, Val2008:Holant}
will recognise that, if we represent $k$-ary functions as column
vectors of length~$2^k\!$, the Fourier transform is equivalently defined
as $\fhat = H^{\otimes k}f$ where $H =
\tfrac12\mtwo{1}{1}{1}{-1} = \tfrac12H^{-1}\!$.  We will use this fact in the proof of
Theorem~\ref{thm:Ferro-E}.

\begin{definition}
\label{def:HyperIsing}
For a real number $\lambda\geq0$ and integer $k\geq 0$, the \emph{$k$-ary
  hypergraph Ising function} is given by
\begin{equation*}
    \hIsing{k}{\lambda}(\bx)
        = \begin{cases}
              \ 1       & \text{if } \bx\in\{\bzero, \bone\} \\
              \ \lambda & \text{otherwise.}
          \end{cases}
\end{equation*}
\end{definition}
The case $\lambda\leq 1$ is known as \emph{ferromagnetic} and
$\lambda\geq 1$ is \emph{antiferromagnetic}.

\begin{definition}\label{def:MC}
    An arity-$k$ \emph{match-circuit} is given by an undirected weighted
    graph~$G$ with vertex set $\{u_1,\dots,u_k\} \cup
    \{v_1,\dots,v_n\}$ for some $n\geq k$.
    Vertices $u_1,\dots,u_k$ have degree~$1$ and  
are called ``external vertices''.    
The edges adjacent
    to them (called ``terminals'') are labelled $y_1,\dots,y_k$.
Vertices $v_1,\ldots,v_n$ are called ``internal vertices''.    
    Each terminal edge has weight~$1$ and each non-terminal
    edge~$e$ is equipped with a positive weight $w_e$.
    Configurations assign spins $0$ and~$1$ to edges.  A configuration
    is a \emph{perfect matching} if every 
 internal vertex has    
    exactly one spin-$1$ edge adjacent to it.  The match-circuit
    implements the function~$f$, where $f(y_1,\ldots,y_k)$ is the sum,
    over perfect matchings, of the product of the weights of edges
    with spin~$1$, where the empty product has weight~$1$.
\end{definition}

Note that, if $f$
is implemented by a match-circuit then so are all functions $c\cdot f$ where $c$
is a positive real number: just add an isolated edge of weight~$c$ to the
match-circuit implementing~$f$. Also, some authors require the underlying graphs of match-circuits
to be planar, and some authors allow the edge weights to be negative.

\begin{definition} 
An arity-$k$ \emph{even-circuit} is given by an undirected weighted
graph~$G$ with vertex set $\{u_1,\dots,u_k\} \cup
\{v_1,\dots,v_n\}$ for some $n\geq k$.  Vertices $u_1, \dots,
u_k$ have degree~$1$ and are called ``external vertices''. The edges adjacent to them (called
``terminals'') are labelled $y_1,\dots,y_k$.  
Vertices $v_1,\ldots,v_n$ are called ``internal vertices''. 
Each terminal edge
has weight~$1$ and each non-terminal edge~$e$ is
equipped with a weight $w_e\in(0,1]$. Configurations assign spins
$0$ and~$1$ to edges.  A configuration is an even subgraph if
every internal vertex has an even number of
spin-$1$ edge adjacent to it.  The even-circuit implements the
function~$f$, where $f(y_1,\ldots,y_k)$ is the sum, over even
subgraphs, of the product of the weights of the edge with
spin~$1$, where the empty product has weight~$1$.
\end{definition}

Note that, for even-circuits, we require all weights to be in~$(0,1]$
whereas, for match-circuits, we only require that weights be positive.
In fact, match-circuits implement the same class of functions when restricted
to weights in $(0,1]$ as they do with arbitrary positive weights, but we use the less restricted definition
for convenience.

For convenience when discussing match-circuits and even-circuits, we
associate an assignment~$\sigma$ of spins to the edges of a graph~$G$
with the spanning subgraph $H = (V(G), \{e\in E(G)\mid \sigma(e)=1\})$.

\begin{definition}
    Given a weighted graph~$H$, we write $w(H) = \prod_{e\in E(H)}
    w_e$ for the weight of~$H$.
\end{definition}

\begin{definition}
\label{defn:clonelist}
    We define the following subsets of~$\calB$.
    \begin{itemize}
    \item $\SD$: all self-dual functions~$f$, i.e., functions such
        that  $f(\bx) = f(\overline{\bx})$ for all~$\bx$.
    \item $\calP$: all functions~$f$ such that $\fhat(\bx)\geq 0$ for
        all~$\bx$.
    \item $\calPN$: all functions~$f$ such that $\fhat(\bx)\geq 0$
        when $\wt{\bx}$~is even and $\fhat(\bx)\leq 0$ when
        $\wt{\bx}$~is odd.
    \item $\SDP = \SD\cap\calP\cap\calPN$.
    \item $\calE$: all functions~$c\cdot f$, where $c$ is a non-negative real
    number and $\fhat$ is implemented by an even-circuit.
    \item $\calM$: all functions~$f$ such that $\fhat$~is implemented by
        a match-circuit.
    \item $\FerroIsing =  \{\hIsing{2}{\lambda} \mid 0\leq \lambda\leq 1\}$.
The functions in the set~$\FerroIsing$ model edge interactions in the ferromagnetic Ising model.
See Cipra~\cite{Cipra} for an introduction to the Ising model.  
    
    \item $\AntiFerroIsing =  \{\hIsing{2}{\lambda} \mid \lambda\geq 1\}$. The functions in the set~$\AntiFerroIsing$ model edge interactions in the anti-ferromagnetic Ising model.
        
    \item $\FerroHyperIsing = \{\hIsing{k}{\lambda} \mid k\geq 2, \lambda\leq 1\}$. The functions in the set~$\FerroHyperIsing$ model ``many-body interactions'' in a generalisation of the 
    ferromagnetic Ising model which applies to hypergraphs --- see \cite{Grimmett} and \cite[Section 2]{FerroPotts}.  
    \end{itemize}
\end{definition}

We emphasise that the sets we have defined are subsets of~$\calB$,
the class of non-negative pseudo-Boolean functions.  There are, for example,
functions outside~$\calB$ whose Fourier transforms are
in~$\calB$, such as the  symmetric, ternary function
$f=[7,-1,-1,7]$, which has Fourier transform $\fhat=[1,0,2,0]$. Even though
$\fhat$~is nonnegative, $f$~is not in~$\calP$ because it is not
in~$\calB$.  Likewise, $f\notin \calM$, even though $\fhat$~is
implemented by a match-circuit (as shown in the proof of
Theorem~\ref{thm:collapse1}).

Instead of~$\calM$, it may seem more natural to consider the
set~$\calM'$ of functions~$f$ that are implemented by match-circuits.
However, $\calM'$~is not a functional clone: for example, it is not
closed under the introduction of fictitious arguments.  By a parity
argument, any function~$f$ that is implemented by a match-circuit must
have $f(\bx)=0$ for all~$\bx$ with even Hamming weight, or $f(\bx)=0$
for all~$\bx$ with odd Hamming weight.  However, any function that is
not everywhere zero and has a fictitious argument must be non-zero for
inputs with both odd and even Hamming weights, so cannot be
implemented by a match-circuit.

As we have remarked, the Fourier transform corresponds in the holant
framework to a holographic transformation by (the appropriate tensor
power of) the Hadamard matrix $\tfrac12\mtwo{1}{1}{1}{-1}$.  This
corresponds, in a certain sense, to transforming the computation from
using basis vectors
$\left(\begin{smallmatrix}1\\0\end{smallmatrix}\right)$ and
$\left(\begin{smallmatrix}0\\1\end{smallmatrix}\right)$ to using
$\left(\begin{smallmatrix}1\\1\end{smallmatrix}\right)$ and
$\left(\begin{smallmatrix}1\\-1\end{smallmatrix}\right)$.  It has been
shown that the latter is the unique basis in which the equality
function can be expressed using matchgates~\cite{planarCSP} and, thus,
our use of the Fourier transform here is essential.  Cai, Lu and
Xia~\cite{planarCSP} have used the Fourier transform as a holographic
transformation from counting CSPs to counting weighted perfect
matchings, as the key tool to obtain polynomial-time algorithms for a
wide range of weighted planar counting CSPs.

Note that, in the definition of~$\calE$,
we allow scaling
by a constant.  We do this to allow the implementation of functions
that have $f(\bzero)<1$.  This would be impossible without scaling,
since the empty graph is an even subgraph of every even-circuit. It
has weight~$1$ and the weight of the empty graph is one of the terms
of the sum defining~$f(\bzero)$.  In contrast, match-circuits can
already implement functions with~$f(\bzero)<1$ without the need for
scaling, and adding scaling to the definition of~$\calM$ would not, in
fact, change the class of implementable functions.

To avoid issues with scaling of Ising and hypergraph Ising functions, we
work with the following clones rather than with $\ofclone{\FerroIsing}$, etc.

\begin{definition}\label{def:Addunaries}
    \begin{align*}
        \Iferro &= \ofclone{\FerroIsing \cup \calB_0} \\
        \Ianti  &= \ofclone{\AntiFerroIsing \cup \calB_0} \\
        \Hferro &= \ofclone{\FerroHyperIsing \cup \calB_0}\,.
    \end{align*}
\end{definition}

\section{Main theorems}

Let $\calL'=\{\calB,\SD,\calP,\calPN,\SDP,\ofclone{\ofclone{\calM}\cup\Hferro},\ofclone{\calM},\Hferro,\ofclone{\calM}\cap \Hferro,\Iferro\}$.

\begin{figure}[t]
\begin{center}
\begin{tikzpicture}[scale=2.25]

    \node (B)   at (1,5.5)  {$\calB$};

    \node (SD)  at (0.25,4.75) {$\SD=\Ianti$};
    \node (P)   at (1,4.75) {$\calP$};
    \node (PN)  at (1.75,4.75) {$\calPN$};

    \node (SDP) at (1,4)    {$\SDP$};
    \node (MHj) at (1,3)    {$\ofclone{\ofclone{\calM} \cup \Hferro}$};

    \node (M)   at (0,2.5)  {$\ofclone{\calM}$};
    \node (H)   at (2,2.5)  {$\Hferro$};

    \node (MH)  at (1,2)    {$\ofclone{\calM}\cap \Hferro$};

    \node (E)   at (1,1)    {$\Iferro = \ofclone{\calE}$};

    \draw   (B)--(SD); \draw (B)--(P); \draw (B)--(PN);
    \draw  (SD)--(SDP); \draw (P)--(SDP); \draw (PN)--(SDP);
    \draw (MHj)--(M); \draw (MHj)--(H);
    \draw   (M)--(MH); \draw (H)--(MH);
    \draw [dashed] (MH)--(E);  
    \draw [dashed] (MHj)--(SDP);  
\end{tikzpicture}
\end{center}
\caption{The lattice $\calL'$.}
\label{fig:lattice}
\end{figure}
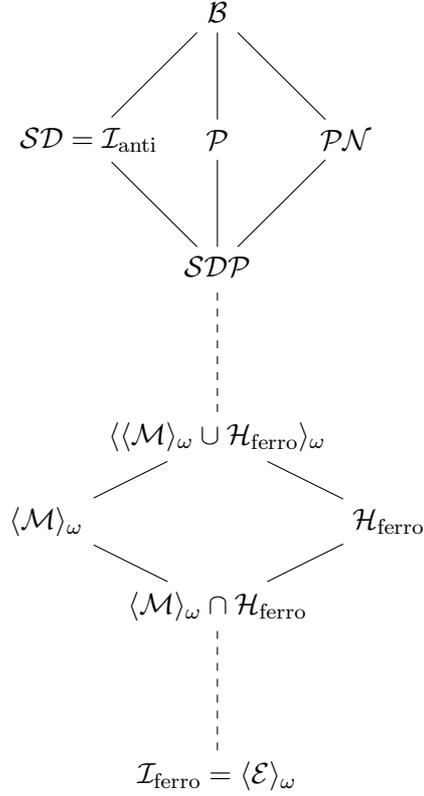

 \newcommand{\statethmmain}{The lattice $\calL'$ shown in
 Figure~\ref{fig:lattice} is a sublattice of $\calL_\omega$.  That is, all
 elements of $\calL'$ are distinct $\omega$-clones, with the possible exceptions of 
 $\SDP$ and $\ofclone{\ofclone{\calM}\cup\Hferro}$, 
 and 
 $\ofclone{\calM}\cap \Hferro$ and $\Iferro$, 
 which might be equal. (This is indicated by the
 dotted lines in Figure~\ref{fig:lattice}.)
    Furthermore, the meets and joins of elements of $\calL'$ are as depicted in Figure~\ref{fig:lattice} and
\begin{enumerate}[(i)]
\item 
$\SD=\Ianti$;
\item $\Iferro = \ofclone{\calE}$;
\item
$\SD$, $\calP$ and~$\calPN$ are maximal in $\calB$; 
\item
$\SDP$ is maximal in $\SD$.
\end{enumerate} }
\begin{theorem}
\label{thm:main}
    \statethmmain
\end{theorem}

Theorem~\ref{thm:main} is proven in Section~\ref{sec:mainthm}.

 \begin{theorem}
\label{thm:stretch}
  For any $\lambda>1$, $\ofclone{\hIsing{2}{\lambda} \cup \calB_0} =
    \Ianti$. For any $\lambda\in(0,1)$,
    $\ofclone{\hIsing{2}{\lambda} \cup \calB_0} = \Iferro$.\end{theorem}
\begin{proof}
    The two parts are Corollaries \ref{cor:Ianti-fin}
    and~\ref{cor:Iferro-fin}, respectively, from Section~\ref{sec:Ising}.
\end{proof}

\begin{theorem}
    $\calE$ and $\calM$ are functional clones, i.e., $\fclone{\calE} =
    \calE$ and $\fclone{\calM} = \calM$.
\end{theorem}
\begin{proof}
    $\fclone{\calE} = \calE$ is Theorem~\ref{thm:E-clone} and
    $\fclone{\calM} = \calM$ is Theorem~\ref{thm:M-clone}.
\end{proof}

\newcommand{\statethmcardinality}{$|\calL_f|=|\calL_\omega|=\beth_2$.}
\begin{theorem}\label{thm:cardinality}
\statethmcardinality
\end{theorem}

Theorem~\ref{thm:cardinality} is proven in
Section~\ref{sec:cardinality}. 

Theorem~\ref{thm:monotone-clone}, proved in Section~\ref{sec:monotone}, 
shows that the set of monotone functions is an $\omega$-clone and gives
examples of $\omega$-clones that generalise  this clone. 

\subsection{Ternary functions}

Given $n\geq 0$ and a set of functions~$\calF\subseteq\calB$, we write
$[\calF]_n = \calF \cap \calB_n$.  
Note that
$[\calB]_n = \calB_n$. Although $[\calF]_n$~is
a set of $n$-ary functions, it essentially includes all functions of
smaller arity.  In particular, if $\calF$~is a clone then it is closed under
the introduction of fictitious arguments, as discussed in
Section~\ref{sec:fclones}, and this allows functions of
smaller arity to be ``padded'' to arity~$n$.  However, $[\calF]_n$ is
not, itself, a clone.

We now focus on the ternary parts of the clones from $\calL'$, in which case
certain distinctions, which were present in $\calL'$, disappear.

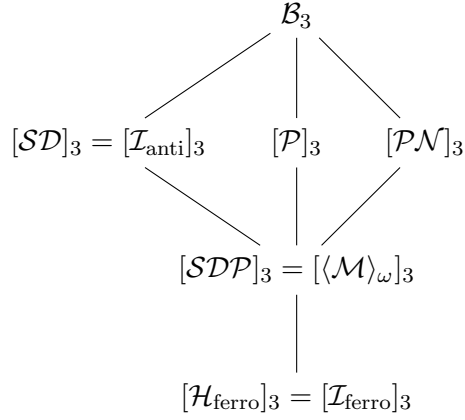
\begin{figure}[t]
\begin{center}
\begin{tikzpicture}[scale=2.25]

    \node (B)   at (1,4.5)  {$\calB_3$};

    \node (SD)  at (-0.1,3.75) {$[\SD]_3=[\Ianti]_3$};
    \node (P)   at (1,3.75) {$[\calP]_3$};
    \node (PN)  at (1.75,3.75) {$[\calPN]_3$};

    \node (SDP) at (1,3)    {$[\SDP]_3 = [\ofclone{\calM}]_3$};

    \node (H)   at (1,2.25)  {$[\Hferro]_3 = [\Iferro]_3$};

    \draw (B)--(SD); \draw (B)--(P); \draw (B)--(PN);
    \draw (SD)--(SDP); \draw (P)--(SDP); \draw (PN)--(SDP);
    \draw (SDP)--(H);
\end{tikzpicture}
\end{center}
\caption{Ternary parts of the clones in $\calL'$.   }
\label{fig:lattice3}
\end{figure}

Let $\calS_3=\{\calB_3, [\SD]_3, [\calP]_3, [\calPN]_3, [\SDP]_3,$ $[\ofclone{M}]_3, [\Hferro]_3, [\Iferro]_3\}$.  

\newcommand{\statethmTernary}{$[\SDP]_3 = [\ofclone{M}]_3$, $[\Hferro]_3 =
[\Iferro]_3$, and any other two elements of $\calS_3$ are distinct.}
\begin{theorem}
\label{thm:ternary}
\statethmTernary
\end{theorem}

Theorem~\ref{thm:ternary} is proved in Section~\ref{sec:ternary} and illustrated
in Figure~\ref{fig:lattice3}, where solid lines indicate \emph{strict} set
inclusions. The (non-strict) inclusions indicated in Figure~\ref{fig:lattice3}
follow trivially from Theorem~\ref{thm:main}. The point of
Theorem~\ref{thm:ternary}, in addition to the two collapses, is that all
inclusions are strict. We note that, however, unlike in
Figure~\ref{fig:lattice}, Figure~\ref{fig:lattice3} does not indicate any
lattice order of $\calS_3$ with respect to $\wedge$ and $\vee$.
  
\section{Finite generation}
\label{sec:finite}

When we defined $\omega$-clones in Section~\ref{sec:fclones}, we 
defined the limit of a set $\calF \subseteq \calB$ 
to be a function~$f$ which is approximated by a sequence of functions~$f_\epsilon$
that are all in the functional clone of some finite subset $S_f$ of $\calF$.
The finiteness restriction was present in the definitions of~\cite{LSM}
and it is retained in this paper because it strengthens our results.
Nevertheless, it causes slight technical problems, and to avoid these problems,
we start the paper by defining a finite subset ${\calB}'_0$ of $\calB_0$
and showing that $\calB_0 \subseteq \ofclone{{\calB}'_0}$.
In the following definition, ``$e$'' is the base of the natural logarithm.
The actual definition of ${\calB}'_0$ is not very constrained, in the sense that
we could have made other choices, but it is important to include an irrational number,
and to include a number that it is smaller than~$1$ and one that is larger than~$1$.
We  use a set of size four to simplify the argument.

\begin{definition}\label{def:B0hat}
${\cal B}'_0 = \{ 1/e,1/2, 2,e\}$
\end{definition}

\begin{lemma}
\label{lem:finitegen}
$\calB_0 \subseteq  \ofclone{{\calB}'_0}$.
\end{lemma}
\begin{proof}
We will show that every nullary function in~$\calB_0$ is a limit of the closure of~$\calB'_0$ under product.
Let $\alpha = \ln 2$.
For any   integers~$a$ and $b$, the quantity
$e^{a + b \alpha}$ 
(viewed as a nullary function) is in $\fclone{{\calB}'_0}$.
So it suffices to show that, for every  real number~$z$
(where $e^z$ is viewed as a nullary function in $\calB_0$) and any $\epsilon>0$,
there are   integers~$a$ and~$b$ such that
$| e^{a + b \alpha} - e^z| < \epsilon$.
Given the universal quantification on $\epsilon$, we can work instead with additive approximation ---
it suffices to show that for every real number~$z$ and every $\delta>0$,
there are  integers~$a$ and~$b$ such that
$|a+b \alpha - z| < \delta$. (To see this, suppose that we are given some $z$ and $\epsilon$. 
Let $\epsilon' = \min(\epsilon, 2 e^z)$
and let  
$\delta = \epsilon' e^{-z}/2$. Then since $\delta \leq 1$, we have
$e^{\delta}-1 \leq 2 \delta$
so $e^{z+\delta} - e^z = e^z(e^\delta - 1) \leq 2 \delta e^z=\epsilon'\leq \epsilon$.
Similarly, $e^z - e^{z-\delta} = e^z (1-e^{-\delta}) \leq 2 \delta e^z \leq \epsilon$.)

Now consider a real number $\delta>0$. 
By Dirichlet's approximation theorem,   there are integers~$p$ and~$q$ 
such that $1 \leq q$ and $|p - q \alpha  |<\delta$.
Since  $\alpha$   is positive and $q\geq 1$,  it is clear 
that $p$ is also positive if    $\delta < \alpha$. 
Also,  since $\alpha$ is irrational, $p-q\alpha$ is non-zero.

Consider any real number~$z$. 
Let $n$ be the integer such that 
$n \times |p - q \alpha| \leq z < (n+1) \times |p-q \alpha|$.
Then $|z - n\times |p - q \alpha| | < |p - q \alpha| < \delta$. 
If $p>q \alpha$ then
$a=p n$ and $b=-q n$ suffices.
Otherwise, 
$a=-p n$ and $b = qn$ suffices.\end{proof}

The proof of Lemma~\ref{lem:finitegen} 
is useful for one more technical finite generation result, so we 
state that here. For this, we need to define a class of parity functions.
\begin{definition}
\label{defn:parity}\label{def:parity}
    For each $k\in\nats$ and $\lambda\in\NonNegReals$, we define the $k$-ary
    function
    \begin{equation*}
        \parityk(\bx) = \begin{cases}
                        \ 1       &\text{if $\wt{\bx}$ is even} \\
                        \ \lambda &\text{otherwise.}
                        \end{cases}
    \end{equation*}
\end{definition}

By analogy to $\calB_0$, we also define a finite version.

\begin{definition}
\label{def:parityfinite}
$\paritybasis = \{\parityk[1/e],\parityk[1/2],\parityk[2],\parityk[e]\}$.
\end{definition}

\begin{lemma}
\label{lem:finitegenparity}
For any even positive integer~$k$
and any $\lambda \in \NonNegReals$, $\parityk \in   \ofclone{\paritybasis}$.
\end{lemma}
 
\begin{proof}
Consider the $k$-ary  function consisting of 
the product of
$a$ copies of $\parityk[e]$ and $b$ copies of $\parityk[2]$.
If the input has even parity, then the output is~$1$.
Otherwise, the output is $e^{a+b \alpha}$.
Combinations of other functions in $\paritybasis$ are similar.
So the proof is essentially the same as the proof of Lemma~\ref{lem:finitegen}.
\end{proof}

\section{The Ising model}
\label{sec:Ising}

Recall the definition of $\hIsing{2}{\lambda}$ from Definition~\ref{def:HyperIsing}, the definition of
$\FerroIsing$ and $\AntiFerroIsing$ from Definition~\ref{defn:clonelist} and the definition of $\Ianti$ from
Definition~\ref{def:Addunaries}.
The following lemma is well known. We include it (with its standard proof) for completeness.

\begin{lemma}\label{lem:I-ferroI}
    $\FerroIsing \subseteq \fclone{\AntiFerroIsing,\calB_0}$.
\end{lemma}
\begin{proof}
    We must show that $\hIsing{2}{\lambda}\in\fclone{\AntiFerroIsing,\calB_0}$
    for all $\lambda\in[0,1]$.  For $\lambda=0$, $\hIsing{2}{\lambda}
    = \EQ$, which is in every functional clone by definition.  For
    $\lambda=1$, $\hIsing{2}{\lambda} \in \AntiFerroIsing$ by
    definition.  Any other function in $\FerroIsing$ is of the form
    $\hIsing{2}{\lambda}$ for some $\lambda\in (0,1)$.  Let
    $\lambda' = \lambda/(1 - \sqrt{1-\lambda^2})$.  Note that
    $\lambda'$ is decreasing as $\lambda$ increases, and that
    $\lambda'>1$ so $\hIsing{2}{\lambda'}\in\AntiFerroIsing$.  Then
    note that $\hIsing{2}{\lambda}(x,y) = \tfrac1{1+\lambda'^2}\sum_w
    \hIsing{2}{\lambda'}(x,w) \hIsing{2}{\lambda'}(w,y)$
    since the weight is~$1$ if $x=y$ and 
    $(2 \lambda')/(1+ \lambda'^2) = \lambda$, otherwise.
\end{proof}

The construction in the proof of the following lemma is based on one
from the proof of \cite[Lemma 3.3]{FerroIsing}.  There are more
efficient constructions, for example \cite[Lemma 3.26]{planartutte}
but we don't need them here.

 \begin{lemma} \label{lem:stretch} Consider $\hIsing{2}{\lambda}$ and
$\hIsing{2}{\lambda'}$ in $\AntiFerroIsing$ with
$\lambda>1$.
Then $\hIsing{2}{\lambda'}\in \ofclone{\{\hIsing{2}{\lambda}\} \cup  {\calB'_0}}$.
\end{lemma}
\begin{proof} By the definition of $\AntiFerroIsing$, $\lambda' \geq 1$.
If $\lambda'=1$ then $\hIsing{2}{\lambda'}$ is the arity-$2$ constant function (with output~$1$).
This can be obtained from the constant~$1$ by introducing two fictitious arguments, so it is in
$\ofclone{\{\hIsing{2}{\lambda}\} \cup \calB_0}$. 

So suppose $\lambda'>1$.
Let $y= 1/\lambda$ and let $f$ be the symmetric arity-$2$ function  
$[y^{1/2},y^{-1/2},y^{1/2}]$, using the symmetric function notation   from Section~\ref{sec:study}. 
For every positive integer~$t$, let $F_{1,t}(x_1,x_2) = f(x_1,x_2)^t$.
For every integer $\ell>1$, 
let  $X_\ell$ be the tuple of variables in $\{x_{i,j} \mid  1\leq i \leq t, 1 \leq j \leq \ell-1 \}$
and let
$$ F_{\ell,t}(x_1,x_2) = \sum_{X_\ell} \prod_{i=1}^t 
\left( f(x_1,x_{i,1}) 
\left(\prod_{j=1}^{\ell-2} f(x_{i,j},x_{i,j+1})\right) f(x_{i,\ell-1},x_2)
 \right).$$
 Note that the quantity~$y^{1/2} $ can be viewed as a nullary function,
 so by Lemma~\ref{lem:finitegen}, 
 $y^{1/2}$ is a limit of   $  \fclone{ \calB'_0}$.
 Since 
 $f = {y}^{1/2} \hIsing{2}{\lambda}$,
 Lemma~\ref{lem:transitive} shows that
 $f$ is a limit of      $\fclone{ \{\hIsing{2}{\lambda} \} \cup  \calB'_0}$.
 Finally, since $F_{\ell,t}$ is 
 formed by summing products of functions in $\calA({\{f\}})$,
 Lemma~\ref{lem:transitive} shows that
 $F_{\ell,t}$
 is also a limit of 
 $\fclone{ \{\hIsing{2}{\lambda} \} \cup  \calB'_0}$.
  
  We wish to show that   $\hIsing{2}{\lambda'}$
  is a limit of 
   $\fclone{ \{\hIsing{2}{\lambda} \} \cup  \calB'_0}$. 
   To do this, we will show that, for every $0<\epsilon<1$,
there are positive integers $t$ and~$\ell$ and a  non-negative constant~$c$ (viewed as a 
limit of  $\fclone{ {\calB'_0}}$)
such that
$$\max_{(x_1,x_2)\in\{0,1\}^2}|\hIsing{2}{\lambda'}(x_1,x_2) - c F_{{\ell,t}}(x_1,x_2)| <\epsilon\,.$$

To see this, consider the following mutual recurrences.
\begin{align*}
m_\ell &= \begin{cases}
\ y^{1/2}, & \text{if }\ell=1, \\
\ y^{1/2} m_{\ell-1} + y^{-1/2} b_{\ell-1}, & \text{if }\ell>1.
\end{cases}\\
b_\ell &= \begin{cases}
\ y^{-1/2},& \text{if }\ell=1, \\
\ y^{-1/2} m_{\ell-1} + y^{1/2} b_{\ell-1}, & \text{if }\ell>1.
\end{cases}
\end{align*}
First, consider $t=1$. 
Renaming the variables $\{x_{1,1},\ldots,x_{1,\ell-1}\}$ to $\{x_3,\ldots,x_{\ell+1}\}$,
the definition of $F_{\ell,t}$ (for $\ell>1$) can be written as
 $$ F_{\ell,1}(x_1,x_2) = \sum_{(x_3,\ldots,x_{\ell+1})}  
  f(x_1,x_{3}) 
\left(\prod_{j=3}^{\ell} f(x_{j},x_{j+1})\right) f(x_{\ell+1},x_2).$$  From the recurrences, it is easy to see that
$F_{{\ell,1}}(0,0) = F_{{\ell,1}}(1,1) = m_\ell$ (``$m$'' stands for ``monochromatic'')
and $F_{{\ell,1}}(0,1) = F_{{\ell,1}}(1,0) = b_\ell$ (``$b$'' stands for ``bichromatic'').
Thus, for general~$t$,
$F_{{\ell,t}}(0,0) = F_{{\ell,t}}(1,1) = m_\ell^t$  
and $F_{{\ell,t}}(0,1) = F_{{\ell,t}}(1,0) = b_\ell^t$.

Now the solution to the recurrences is 
\begin{align*}
m_\ell &= y^{-\ell/2} ((y+1)^\ell + (y-1)^\ell)/2\\
b_\ell &= y^{-\ell/2} ((y+1)^\ell - (y-1)^\ell)/2\,.
\end{align*}
Thus,  since $0<y<1$, for odd~$\ell$ we have
$$\frac{b_\ell}{m_\ell} = 1+\frac{2}{{\left(\frac{1+y}{1-y}\right)}^\ell-1}\,.$$

So finally, given $0<\epsilon<1$, let $\ell$ be the smallest  odd integer so that  
$${ \left(\frac{1+y}{1-y}\right)}^\ell > 1+\frac{2 \lambda'}{\epsilon}\,.$$
Let $t$ be the  smallest integer so that
$$\left(1+\frac{2}{{\left(\frac{1+y}{1-y}\right)}^\ell-1}\right)^t > \lambda'\,.$$
Let $c=m_\ell^{-t}$.
Then $c  F_{{\ell,t}}(0,0) = c F_{{\ell,t}}(1,1) =  1$.
Also  $ c F_{{\ell,t}}(0,1) =  c F_{{\ell,t}}(1,0) = (b_\ell/m_\ell)^t$ so
\begin{align*}
\lambda' < c F_{{\ell,t}}(0,1) =  c F_{{\ell,t}}(1,0) 
&= {\left(
1+\frac{2}{{\left(\frac{1+y}{1-y}\right)}^\ell-1}
\right)}^t \\
&= 
 {\left(
1+\frac{2}{{\left(\frac{1+y}{1-y}\right)}^\ell-1}
\right)}^{t-1}
 {\left(
1+\frac{2}{{\left(\frac{1+y}{1-y}\right)}^\ell-1}
\right)} \\
&\leq
  \lambda'
 {\left(
1+\frac{\epsilon}{\lambda'}
\right)} < \lambda' + \epsilon\,,
\end{align*}
as required.
\end{proof}

\begin{corollary}
\label{cor:Ianti-fin}
    For any $\lambda > 1$, $\ofclone{\hIsing{2}{\lambda} \cup \calB'_0}
    = \Ianti$.
\end{corollary}
\begin{proof}
Recall from definition~\ref{def:Addunaries} 
that
$\Ianti  = \ofclone{\AntiFerroIsing \cup \calB_0}$
and from Definition~\ref{defn:clonelist}
that for any $\lambda >1$, $\hIsing{2}{\lambda}\in \AntiFerroIsing$. 
This shows 
$\ofclone{\hIsing{2}{\lambda} \cup \calB'_0}    \subseteq \Ianti$.
To see that 
$\Ianti \subseteq \ofclone{\hIsing{2}{\lambda} \cup \calB'_0}$
we only need to show that for any $\hIsing{2}{\lambda} \in \Ianti$,
$\hIsing{2}{\lambda} \in \ofclone{\hIsing{2}{\lambda} \cup \calB'_0} $, and this is Lemma~\ref{lem:stretch}.
\end{proof}
 
\begin{lemma} \label{lem:stretchferro}
Consider $\hIsing{2}{\lambda}$ and
$\hIsing{2}{\lambda'}$ in $\FerroIsing$ with
$0<\lambda<1$.
Then $\hIsing{2}{\lambda'}\in \ofclone{\{\hIsing{2}{\lambda}\} \cup \calB'_0}$.
\end{lemma}
\begin{proof} 
As in the proof of Lemma~\ref{lem:stretch}, the proof is straightforward if
$\lambda'\in\{0,1\}$,
so assume $0<\lambda'<1$.  Define $y$,
$f$   and~$F_{{\ell,t}}$ as in the proof of Lemma~\ref{lem:stretch}.
Note that $y>1$, so
$$\frac{m_\ell}{b_\ell} = 1 +\frac{2}
{{ \left( \frac{y+1}{y-1} \right) }^\ell-1}\,.$$
Given $0<\epsilon<1$, let $\ell$ be the smallest positive  integer so that  
$${ \left(\frac{y+1}{y-1}\right)}^\ell > 1+\frac{2}{\epsilon} .$$
Let $t$ be the   largest integer so that
$$\left(1+\frac{2}{{\left(\frac{y+1}{y-1}\right)}^\ell-1}\right)^{t-1} \leq \frac{1}{\lambda'}\,.$$
Let $c=\lambda' b_\ell^{-t}$.
Then $c  F_{{\ell,t}}(0,1) = c F_{{\ell,t}}(1,0) =  \lambda'$.
Also $c F_{{\ell,t}}(0,0) =  c F_{{\ell,t}}(1,1) 
= \lambda' m_\ell^t/b_\ell^t$
and $m_\ell^t/b_\ell^t> 1/\lambda'$, 
so
\begin{align*}  1 < c F_{{\ell,t}}(0,0) =  c F_{{\ell,t}}(1,1) 
&= \lambda' {\left(
1+\frac{2}{{\left(\frac{y+1}{y-1}\right)}^\ell-1}
\right)}^t \\
&= \lambda'
 {\left(
1+\frac{2}{{\left(\frac{y+1}{y-1}\right)}^\ell-1}
\right)}^{t-1}
 {\left(
1+\frac{2}{{\left(\frac{y+1}{y-1}\right)}^\ell-1}
\right)} \\
&<
  { 
1+ \epsilon  
 }\,,
\end{align*}
as required.
\end{proof}

The   proof of the following corollary is straightforward and is essentially identical to the proof
of Corollary~\ref{cor:Ianti-fin}.
 
\begin{corollary}
\label{cor:Iferro-fin}
    For any $\lambda\in(0,1)$,
    $\ofclone{\hIsing{2}{\lambda}\cup\calB'_0} = \Iferro$.
\end{corollary}

\section{$\boldsymbol{\omega}$-clones defined by Fourier coefficients}
\label{sec:Fourier}

\subsection{Properties of Fourier coefficients}

The proofs of the following three lemmas are routine calculations and we
defer them to Appendix~\ref{app:Fourier}.   

\newcommand{\statelemfops}{  Let $f$ and $g$ be functions in $\calB_k$.
    \begin{enumerate}[(i)]
    \item \label{op-perm}
        For any permutation~$\pi$ of~$[k]$, $\widehat{f^{\pi}}(\bx) =
        \fhat(\pi(\bx))$.
    \item \label{op-fict}
        If $h(\bx z) = f(\bx)$, then $\hhat(\bx0) = \fhat(\bx)$ and
        $\hhat(\bx1) = 0$.
    \item \label{op-sum}
        If $h(\bx) = f(\bx0)+f(\bx1)$, then $\hhat(\bx) =
        2\fhat(\bx0)$.
    \item \label{op-comp}
        If $h(\bx) = f(\overline{\bx})$, then $\hhat(\bx) =
        (-1)^{\wt{\bx}}\fhat(\bx)$.
    \item \label{op-lim}
        If $\|g-f\|_\infty < \epsilon$, then $\|\ghat-\fhat\|_\infty
        < \epsilon$.
        \item \label{op-nullary}
        If $k=0$ then $\fhat = f$.
    \end{enumerate}
}
\begin{lemma}
\label{lem:fops}
\statelemfops
\end{lemma}

It is also well-known (see, e.g.,~\cite{deWolf08:brief,ODonnell14:book}) that, if $f,g\in\calB_k$, and $h$~is
defined by $h(\bx) = f(\bx)\,g(\bx)$, then $\hhat$~is given by the
convolution
\begin{equation}\label{eq:convolution}
    \hhat(\bx)\ \; = \!\!\sum_{\bw\in\Bools^k}\!\! \fhat(\bw)\,\ghat(\bw\oplus \bx)\,.
\end{equation}

We will later need to know the Fourier coefficients of hypergraph
Ising functions and of  the parity functions defined in Definition~\ref{defn:parity}.

\newcommand{\statelemfthIsing}{For any $k$ and~$\lambda$,
    \begin{equation*}
        \hIsinghat{k}{\lambda}(\bx) = \begin{cases}
            \ \lambda + (1-\lambda)/2^{k-1} & \text{if }\bx = \bzero \\
            \ (1-\lambda)/2^{k-1} &\text{if $\wt{\bx}$ is even
              and positive} \\
            \ 0 & \text{if $\wt{\bx}$ is odd.}
        \end{cases}
    \end{equation*}}
\begin{lemma}
\label{lem:ft-hIsing}
  \statelemfthIsing
\end{lemma}

\newcommand{\statelemftparity}{For any $k$ and~$\lambda$,
    $\parityhatk(\bzero) = \tfrac12(1+\lambda)$,
    $\parityhatk(\bone) = \tfrac12(1-\lambda)$ and
    $\parityhatk(\bx) = 0$ for any $\bx\notin\{\bzero,
    \bone\}$.
}
\begin{lemma}
\label{lem:ft-parity}
    \statelemftparity{}
\end{lemma}

\subsection{$\boldsymbol{\calP}$ and $\boldsymbol{\calPN}$}

Recall from Definition~\ref{defn:clonelist} that $\calP$ is the class
of functions~$f$ such that $\fhat(\bx)\geq 0$ for all~$\bx$, and that
$\calPN$ is the class of functions~$f$ such that $\fhat(\bx)\geq 0$ if
$\wt{\bx}$ is even and $\fhat(\bx)\leq 0$ if $\wt{\bx}$ is odd.
We first show that $\calP$ and $\calPN$ are  $\omega$-clones, and that they contain $\calB_0$.

\begin{theorem}
\label{thm:P-clone}
    $\ofclone{\calP} = \calP{}$ and $\calB_0 \subseteq \calP$.
\end{theorem}
\begin{proof}
By the definition of $\omega$-clones,  
the fact that $\calP$ is an $\omega$-clone follows 
from the fact that it contains $\EQ$ and that it is closed under the various operations.
    \begin{itemize}
    \item It is easily verified (for example, apply Lemma~\ref{lem:ft-hIsing} with $\lambda=0$) that $\widehat{\EQ} = \tfrac12\EQ$,
        which is a non-negative function.  Therefore, $\EQ\in\calP$.
    \item For closure under permuting arguments, suppose that
        $f\in\calP$ and let $h = f^{\pi}$ for some permutation~$\pi$.
        By Lemma~\ref{lem:fops}\pref{op-perm}, $\fhat$ and~$\hhat$
        have the same range, so $\hhat$~is a nonnegative function, so
        $h\in\calP$.
    \item For closure under introducing fictitious arguments, let
        $f\in\calP$ and define $h(\bx y) = f(\bx)$.  Then $h\in\calP$
        because, by Lemma~\ref{lem:fops}\pref{op-fict}, every Fourier
        coefficient of~$h$ is either zero or a Fourier coefficient
        of~$f$.
    \item For closure under summation, let $f\in\calP$ and define
        $h(\bx)=f(\bx0)+f(\bx1)$.  By Lemma~\ref{lem:fops}\pref{op-sum},
        $\hhat(\bx)=2\fhat(\bx0)\geq 0$ for any~$\bx$, so $h\in\calP$.
    \item For closure under products, let $f,g\in\calP$.
        $\fghat(\bx) = \sum_{\bw}\fhat(\bw)\,\ghat(\bw\oplus \bx) \geq
        0$, since every term of the sum is nonnegative, so $fg\in\calP$.
    \item For closure under limits, let $f$ be a function and suppose
        that, for every $\epsilon>0$, there is some $f_\epsilon\in\calP$
        with $\|f_\epsilon - f\|_\infty < \epsilon$ (this is a weaker
        condition than requiring all such $f_\epsilon$ to be in
        $\fclone{\calG}$ for some finite $\calG\subseteq\calP$).  Then, by
        Lemma~\ref{lem:fops}\pref{op-lim},
        $\|\fhat_\epsilon - \fhat\|_\infty < \epsilon$.  In
        particular, $\fhat_\epsilon(\bx)\geq 0$ for all~$\bx$ so, for
        all~$\bx$ and all $\epsilon>0$, $\fhat(\bx) > -\epsilon$.
        Therefore, $\fhat(\bx)\geq 0$ and $f\in\calP$. 
    \end{itemize}
We now show that $\calB_0 \subseteq \calP$.
Consider any $c\in \NonNegReals$
and let $f_c$ be the nullary function in~$\calB_0$ with range $\{c\}$. 
Let $g_c$ be the unary function defined by $g_c(0)=g_c(1)=c/2$.
$\widehat{g}_c(0) = c$ and $\widehat{g}_c(1) = 0$,
so $g_c \in \calP$.
But $f_c(x_1) = \sum_{x_1} g_c(x_1)$ and $\omega$-clones are closed under summation, so 
$f_c \in \calP$.  \end{proof}

\begin{definition}Consider a function $f\in \calB_k$.
We define the \emph{complement} $\overline{f}$ of~$f$
by $\overline{f}(\bx) = f(\overline{\bx})$.
\end{definition}

\begin{theorem}
\label{thm:PN-clone}
    $\ofclone{\calPN} = \calPN$
    and $\calB_0 \subseteq \calPN$.    
\end{theorem}
\begin{proof}
    By Lemma~\ref{lem:fops}\pref{op-comp}, $\calPN = \{f\mid
    \overline{f}\in\calP\}$.   We first show that $\calPN$ is an $\omega$-clone.
    \begin{itemize}
    \item Since $\overline{\EQ} = \EQ$ 
    and $\calP$ contains $\EQ$, 
    $\calPN$ also contains $\EQ$.

    \item For closure under permuting arguments, let $f$ be a $k$-ary
        function in~$\calPN$ and let $\pi$ be a permutation of~$[k]$.
        By Lemma~\ref{lem:fops}\pref{op-perm},
        $\widehat{f^{\pi}}(\bx) = \fhat(\pi(\bx))$ and, since
        $\wt{\bx}=\wt{\pi(\bx)}$, we have $f^{\pi}\in\calPN$.

    \item For closure under introducing fictitious arguments, let
        $f\in\calPN$ and define $h(\bx y) = f(\bx)$.  Then
        $\overline{h}(\bx y) = h(\overline{\bx y}) = f(\overline{\bx})
        = \overline{f}(\bx)\in \calP$, so $h\in\calPN$.

    \item For closure under summation, let $f\in\calPN$, so
        $\overline{f}\in\calP$.  Define $h(\bx) = f(\bx0) + f(\bx1)$.
        Then
        $\overline{h}(\bx) = h(\overline{\bx}) = f(\overline{\bx}0) +
        f(\overline{\bx}1) = \overline{f}(\bx1) +
        \overline{f}(\bx0)\in \calP$, so $h\in\calPN$.

    \item For closure under products, suppose $f,g\in\calPN$ and let
        $h(\bx)=f(\bx)\,g(\bx)$.  Then
        $\overline{h}(\bx) = h(\overline{\bx}) = f(\overline{\bx})\,
        g(\overline{\bx}) = \overline{f}(\bx)\,\overline{g}(\bx)\in
        \calP$, so $h\in\calPN$.

    \item For closure under limits, let $f$ be a function and suppose
        that, for all $\epsilon>0$, there is some $f_\epsilon\in\calPN$
        such that $\|f_\epsilon - f\|_\infty < \epsilon$.  We must
        show that $f\in\calPN$. 

        By Lemma~\ref{lem:fops}\pref{op-lim},
        $\|\fhat_\epsilon - \fhat\|_\infty < \epsilon$.  In
        particular, $\fhat_\epsilon(\bx)\geq 0$ for all
        even-weight~$\bx$, and $\fhat_\epsilon(\bx)\leq 0$ for all
        odd-weight~$\bx$.  Therefore, for all even-weight~$\bx$, and
        all $\epsilon>0$, $\fhat(\bx) > -\epsilon$, so
        $\fhat(\bx)\geq 0$.  Similarly, $\fhat(\bx)\leq 0$ for all
        odd-weight~$\bx$, so $f\in\calPN$.
    \end{itemize}
    The proof that $\calB_0 \subseteq \calPN$ is the same as the proof that $\calB_0 \subseteq \calP$
    (see the proof of Theorem~\ref{thm:P-clone}). \end{proof}

 We now
investigate  the position of~$\calP$ and~$\calPN$ in the lattice~$\calL'$ from Theorem~\ref{thm:main}.  To do this, we
use two technical lemmas, which we will also use in
Section~\ref{sec:SD-max}.

\begin{lemma}
\label{lem:neg-fc}
    Let $f\in\calB_n$.  If $\fhat(\ba)<0$ for some $\ba\in\Bools^n\!$,
    then there is a function $g\in\fclone{\{f\}}$ such that
    $\ghat(\bone)<0$.
\end{lemma}
\begin{proof}
Since $\fhat(\ba)\neq 0$, $f$~cannot be the constant zero
function.  Therefore, $f(\bx)>0$
for some $\bx\in\Bools^n\!$, which means that $\fhat(\bzero)>0$,
so $\ba\neq\bzero$.  
Since functional clones are closed under permuting arguments,
and (by    Lemma~\ref{lem:fops}\pref{op-perm}), permuting arguments just permutes
Fourier coefficients,  we may
    assume that, for some $k\in[n]$, $a_1 = \dots = a_k = 1$ and
    $a_{k+1} = \dots = a_n = 0$.  Let
    \begin{equation*}
        g(x_1, \dots, x_k)\ \ = \!\!\!\sum_{x_{k+1}, \dots, x_n}\!\!\! f(x_1, \dots, x_n)\,.
    \end{equation*}
    By Lemma~\ref{lem:fops}\pref{op-sum},
    $\ghat(\bone) = 2^{n-k}\,\fhat(\ba)<0$.
\end{proof} 
\begin{definition}\label{def:permissive}
A function $f\colon\Bools^k\to\NonNegReals$ is \emph{permissive} if its range is~$\PosReals$. 
\end{definition}
\begin{lemma}
\label{lem:tech}
    Let $f\in\calB_n$ with $\fhat(\bone)<0$.  Then, for every $k>0$,
    there is a $k$-ary permissive function
    $h\in\fclone{\{f, \parity{k+n}{1/2}\}}$ such that
    $\hhat(\bone)<0$.
\end{lemma}
\begin{proof}
    Let $\bx = (x_1, \dots, x_k)$ and $\by = (y_1,\dots, y_n)$.  Define the
    functions
    \begin{align*}
        f'(\bx,\by) &= f(\by) \\
        g(\bx,\by)  &= \parity{k+n}{1/2}(\bx,\by)\,f'(\bx,\by) \\
        h(\bx)      &= \!\!\sum_{\by\in\Bools^n}\!\! g(\bx,\by)\,.
    \end{align*}
    Thus, $h\in\fclone{\{f,\parity{k+n}{1/2}\}}$.  Further, since
    $\parity{k+n}{1/2}$    
    is permissive and $f$~is not the constant zero function
    (because $\fhat(\bone)\neq 0$), $h$~is also permissive.

    For the claim that $\hhat(\bone)<0$, we have the
    following.  The first equality is by
    Lemma~\ref{lem:fops}\pref{op-sum} and the second
    by Equation~\eqref{eq:convolution} from the beginning of Section~\ref{sec:Fourier}.  The third equality is because the
    first $k$~arguments of~$f'$ are fictitious so, by
    Lemma~\ref{lem:fops}\pref{op-fict},
    $\widehat{f'}(\overline{\bv},\bw)$ is $\fhat(\bw)$ when
    $\bv=\bone$ and is zero, otherwise.  The final equality is
    because, by Lemma~\ref{lem:ft-parity},
    $\parityhat{k+n}{1/2}(\bone,\bw)$ is~$\tfrac14$ when $\bw=\bone$ and is
    zero, otherwise.
    \begin{align*}
        \hhat(\bone)
            &= 2^n\, \ghat(\bone, \bzero) \\
            &= 2^n \!\!\!\!\sum_{\bv\bw\in\Bools^{k+n}}\!\!\!\!
                       \parityhat{k+n}{1/2}(\bv,\bw)\,
                           \widehat{f'}(\overline{\bv},\bw) \\
            &= 2^n \sum_{\bw\in\Bools^n}\parityhat{k+n}{1/2}(\bone,\bw)\,
                           \fhat(\bw) \\
            &= 2^{n-2}\,\fhat(\bone) < 0\,.\qedhere
    \end{align*}
\end{proof}

\begin{lemma}
\label{lem:P-max}
    $\calP$ is maximal in~$\calB$.
\end{lemma}
\begin{proof}
    Consider any $n$-ary $f\in\calB\setminus\calP$.  By definition,
    $\fhat(\ba)<0$ for some~$\ba\in\Bools^n$ and, by
    Lemma~\ref{lem:neg-fc}, we may assume that $\fhat(\bone)<0$.  By
    Lemma~\ref{lem:tech}, there is a permissive unary function
    $h\in\fclone{\{f,\parity{n+1}{1/2}\}}$ such that
    $\hhat(1)<0$.  By Lemma~\ref{lem:ft-parity},
    $\parity{n+1}{1/2}\in\calP$ for every~$n$, so we
    have $h\in\fclone{\calP\cup\{f\}}$.

    Since $h$~is permissive and $\hhat(1) = \tfrac12(h(0)-h(1))$, we
    have $h(1)>h(0)>0$.  We may further assume that $h(0)<\tfrac12$
    and $h(1)=1$: if this is not the case, replace~$h$ with the function
    $h'(x)=(h(x)/h(1))^j$ for any sufficiently large integer~$j$.
The function~$h'$ is in $\fclone{\calP \cup \{f\}}$ since $h\in \fclone{\calP \cup \{f\}}$
and nullary functions such as $1/h(1)$ are in~$\calP$ by  Theorem~\ref{thm:P-clone}.

    Now, consider the symmetric binary function
    $g = [h(0)^{-2}, h(0)^{-1}, 0]$, which is in~$\calP$ by the
    assumption on~$h$.  We have
    \begin{equation*}
        \mathrm{NAND}(x,y) = [1,1,0]
            = g(x,y)\,h(x)\,h(y)\in\fclone{\calP\cup\{f\}}\,.
    \end{equation*}
    By \cite[Corollary~13.2(ii)]{LSM}, any $\omega$-clone that contains
    $\mathrm{NAND}$, a unary function~$h$ such that $h(1)>h(0)>0$
    and the nullary function~$1/2$    
    also contains all of~$\calB_1$.  Therefore,
    $\calB_1\cup\{\mathrm{NAND}\}\subseteq \ofclone{\calP\cup\{f\}}$.
    By Lemmas 7.1 and~8.1 of~\cite{LSM}, this implies that that
    $\ofclone{\calP\cup\{f\}} = \calB$, so $\calP$~is maximal
    in~$\calB$.
\end{proof}

\begin{corollary}
\label{cor:PN-max}
    $\calPN$ is maximal in~$\calB$.
\end{corollary}
\begin{proof}
    Let $f\in\calB\setminus\calPN$.  By
    Lemma~\ref{lem:fops}\pref{op-comp}, we have
    $\overline{f}\notin \calP$.
We will now show that     
 $ \fclone{\calPN\cup \{f\}}
            = \{\overline{g}\mid g\in\fclone{\calP\cup\{\overline{f}\}}\}$.    
To see this, suppose that $g \in \fclone{\calPN\cup \{f\}}$.
Then $g$ is defined by a summation of a product of
functions in $\calA(\calPN\cup \{f\})$.
Complementing all of the functions in $\calA(\calPN\cup \{f\})$
exchanges the roles of $0$'s and $1$'s so the 
summing the product of the complements defines $\overline{g}$.
Since the complements of the functions in $\calA(\calPN\cup \{f\})$
are in $\calA(\calP \cup \{\overline{f}\})$, this shows that $\overline{g}$
is in $\fclone{\calP\cup\{\overline{f}}\}$. 
A similar argument gives  the other direction.

Closing under limits, we get
    \begin{align*}
        \ofclone{\calPN\cup \{f\}}
            &= \{\overline{g}\mid g\in\ofclone{\calP\cup\{\overline{f}\}}\}\\ 
            &= \{\overline{g}\mid g\in\calB\}\\
            &= \calB\,.\qedhere
    \end{align*}
\end{proof}

\begin{corollary}
\label{cor:PjoinPN}
    $\ofclone{\calP\cup\calPN} = \calB$
\end{corollary}
\begin{proof}
    Consider the symmetric function $f=[0,1,2]$.  We have
    $\fhat=[1,-\tfrac12,0]$, so $f\in\calPN\setminus\calP$ and the
    result is immediate from maximality of~$\calP$ in~$\calB$
    (Lemma~\ref{lem:P-max}).
\end{proof}

\section{Self-dual functions}
\label{sec:self-dual}

Recall from Definition~\ref{defn:clonelist} that $\SD$ is the class of
self-dual functions, i.e., functions for which $f(\bx) =
f(\overline{\bx})$ for all~$\bx$ of appropriate arity.

\begin{theorem}\label{thm:SDclone}
    $\ofclone{\SD} = \SD$ and $\calB_0 \subseteq \SD$.\end{theorem}
\begin{proof}
By the definition of $\omega$-clones,  
the fact that $\SD$ is an $\omega$-clone follows 
from the fact that it contains $\EQ$ and that it is closed under the various operations.
    \begin{itemize}
    \item The equality function  is clearly
        self-dual  so it is in $\SD$.

    \item For closure under permuting arguments, let $f\in\SD$ be a
        $k$-ary function and let $\pi$ be a permutation of~$[k]$.
        Then $f^{\pi}\in\SD$, since
        $f^{\pi}(\overline\bx) = f(\pi(\overline{\bx})) =
        f(\overline{\pi(\bx)}) = f(\pi(\bx)) = f^{\pi}(\bx)$.

    \item For closure under introducing fictitious arguments, let
        $f\in\SD$ and define $h(\bx y) = f(\bx)$.  Then
        $h(\overline{\bx y}) = f(\overline{\bx}) = f(\bx) = h(\bx y)$,
        so $h$~is self-dual.

    \item For closure under summation, let $f\in\SD$ and define
        $h(\bx) = f(\bx0) + f(\bx1)$. Then
        $h(\overline{\bx}) = f(\overline{\bx}0) + f(\overline{\bx}1) =
        f(\bx1) + f(\bx0) = h(\bx)$, so $h\in\SD$.

    \item For closure under products, let $f,g\in\SD$ and consider
        $h(\bx)=f(\bx)\,g(\bx)$.  We have
        $h(\overline{\bx}) = f(\overline{\bx})\,g(\overline{\bx}) =
        f(\bx)\,g(\bx) = h(\bx)$, so $h\in\SD$.

    \item For closure under limits, let $f\in\calB$ and suppose that,
        for all $\epsilon>0$, there is some $f_\epsilon\in\SD$ such that
        $\|f-f_\epsilon\|_\infty < \epsilon$.  We must show that $f\in\SD$.

        For any~$\bx$, $|f(\bx) - f_\epsilon(\bx)| < \epsilon$ and
        $|f(\overline{\bx}) - f_\epsilon(\overline{\bx})| <
        \epsilon$.  But, since $f_\epsilon$ is self-dual, this gives
        $|f(\overline{\bx}) - f_\epsilon(\bx)| < \epsilon$.
        It follows that $|f(\bx) - f(\overline{\bx})| < 2\epsilon$ for
        all $\epsilon>0$, so $f(\bx) = f(\overline{\bx})$, so
        $f\in\SD$.
    \end{itemize}
    The proof that $\calB_0 \subseteq \SD$ is the same as the proof that $\calB_0\subseteq \calP$
    in the proof of Theorem~\ref{thm:P-clone}.\qedhere
\end{proof}

It turns out that the functions in $\SD$ also have a natural
characterisation in terms of their Fourier transforms.  This allows us
to study the relationship between $\SD$ and the $\omega$-clones from
Section~\ref{sec:Fourier}.

\begin{lemma} 
\label{lem:oddSD}
A $k$-ary function~$f$ is in \SD{} if and only if $\fhat(\bx) = 0$
for all $\bx$ with odd Hamming weight.
\end{lemma}
\begin{proof} Suppose $f\in\SD$.  We have $f(\bx)=f(\overline{\bx})$
  so, by Lemma~\ref{lem:fops}\pref{op-comp}, $\fhat(\bx) =
  (-1)^{|\bx|}\fhat(\bx)$.  When $|\bx|$~is odd, this
  implies that $\fhat(\bx)=0$.

  Conversely, if $\fhat(\bx)=0$ for all~$\bx$ with $|\bx|$~odd, then
\begin{align*}
    2^k f(\bx)\ \ &= \!\sum_{\bw\in\Bools^k}\!\! (-1)^{|\bw\wedge\bx|} \fhat(\bw) \\
      &= \!\sum_{\substack{\bw\in\Bools^k,\\|\bw|\text{ even}}} \!\!(-1)^{|\bw\wedge\bx|} \fhat(\bw) \\
      &= \!\sum_{\substack{\bw\in\Bools^k,\\|\bw|\text{ even}}} \!\!(-1)^{|\bw\wedge\overline{\bx}|} \fhat(\bw)
      = 2^k f(\overline{\bx})\,,
\end{align*}
so $f\in\SD$.
\end{proof}

\begin{theorem}
\label{thm:SDP}
    $\calP \cap \SD = \calPN \cap \SD = \calP\cap\calPN$.
\end{theorem}
\begin{proof}
    It is immediate from the definitions and Lemma~\ref{lem:oddSD}
    that each of these is the class of functions~$f$ such that
    $\fhat(\bx)\geq 0$ if $\wt{\bx}$~is even and $\fhat(\bx)=0$
    if $\wt{\bx}$ is odd.
\end{proof}
 
Recall from Section~\ref{sec:lfp} that the intersection of two $\omega$-clones is an $\omega$-clone. 
In the light of Theorem~\ref{thm:SDP}, we
make the following definition.
\begin{definition}
Let $\SDP$ be the $\omega$-clone $\SD \cap \calP \cap \calPN$.
\end{definition} 
 Theorem~\ref{thm:SDP} makes
it clear that $\SDP = \calP \cap \SD = \calPN \cap \SD = \calP\cap\calPN$.

\begin{lemma}
\label{lem:SDPN-incomp}
    $\SD$, $\calP$ and $\calPN$ are pairwise incomparable
    under  subset inclusion. 
\end{lemma}
\begin{proof}
    Consider the functions $f = [0,1,0]$ (the binary disequality
    function), $g = [2,1,0]$ and $h = [0,1,2]$.  We have $\fhat =
   [\tfrac12, 0, -\tfrac12]$, 
   $\ghat = [1,\tfrac12,0]$ and $\hhat =
    [1,-\tfrac12,0]$.   
    Lemma~\ref{lem:oddSD} and the definitions of $\calP$ and $\calPN$ imply that, 
    among the three
    $\omega$-clones in the statement, $f$~is only in~$\SD$, $g$~is
    only in~$\calP$ and $h$~is only in~$\calPN$.
\end{proof}

For the relationship between $\SDP$ and $\Hferro$, we use the concept
of \emph{log-supermodular} functions.  A function
$f\colon\Bools^k\to\NonNegReals$ is log-supermodular if
$f(\bx\lor\by)\,f(\bx\land\by)\geq f(\bx)\,f(\by)$ for all
$\bx,\by\in\Bools^k$, where $\lor$ and $\land$ are applied bitwise.

\begin{definition}
    Let $\LSM$ be the set of all log-supermodular functions.
\end{definition}

$\LSM$ is
an $\omega$-clone \cite[Lemma~4.2]{LSM}.  Note that all 
unary functions are trivially  log-supermodular. Nullary functions are also log-supermodular,
for example, using the proof that $\calB_0 \subseteq \calP$ (proof of Theorem~\ref{thm:P-clone}).

The following characterisation of permissive log-supermodular
functions of arity at least~$2$ is due essentially to
Topkis~\cite{Top1978:submod} (see also \cite[Lemma~5.1]{LSM}).  It provides a
simple way to check that a permissive function is log-supermodular.  A
$2$-pinning of a $k$-ary function~$f$ (with $k\geq 2$) is any binary function $g(x,y) = f(z_1,
\dots, z_k)$ where each $z_i\in\{0,1,x,y\}$ ($i\in[k]$), such that $x$ and~$y$
each appear exactly once in the sequence $z_1, \dots, z_k$. It is immediate from the definition that every $2$-pinning of a log-supermodular function~$f$ is also log-supermodular. The
following lemma states that, for permissive functions, this condition is also
sufficient.

\begin{lemma}[\cite{Top1978:submod}]\label{lem:topkis}
A permissive $k$-ary function is log-supermodular if, and only if, every $2$-pinning of~$f$ is log-supermodular.
\end{lemma}

\begin{theorem}
\label{thm:Hferro-SDP}
    $\Hferro\subset\SDP$.
\end{theorem}
\begin{proof}
    Inspection of Lemma~\ref{lem:ft-hIsing} shows that
    $\Hferro\subseteq \calP\cap \calPN$ and, by Theorem~\ref{thm:SDP},
    $\calP\cap\calPN = \SDP$.  It remains to show that the inclusion
    is strict.

    It is easy to check that every function in
    $\FerroHyperIsing\cup\calB_0$ is log-supermodular.  It follows that
    $\Hferro$~is a subset of the $\omega$-clone of all log-supermodular
    functions so, in particular, every function in $\Hferro$ is
    log-supermodular.

    Consider the $4$-ary function~$f = [13,4,1,4,13]$.  This
    function is not log-supermodular by Lemma~\ref{lem:topkis}, since the
    pinning $g(x,y) =
    f(x,y,0,0)$ has $g(1,1)\,g(0,0) = 13 < g(0,1)\,g(1,0) = 16$.
    Therefore, $f\notin\Hferro$.  However, $f\in\SD$ and we have
    $\fhat = [4,0,\tfrac32,0,0]$ (the odd-weight coefficients are zero
    by Lemma~\ref{lem:oddSD}) so $f\in\calP\cap\calPN\cap\SD = \SDP$.
\end{proof}

\subsection{Self-dual functions and Ising}

In this section, we prove that $\SD = \Ianti$ (Theorem~\ref{thm:SD-Ising}).  To do this, we
introduce a  functional clone, $\Parev$, of weighted, even-arity parity
functions.  Recall from Definition~\ref{defn:parity} that, for $k\in\nats$
and $\lambda\in\NonNegReals$, $\parityk(\bx)=1$ if $\wt{\bx}$~is even,
and $\parityk(\bx)=\lambda$, otherwise.
Note that, when $k$~is even, $\parityk$~is self-dual. Note also that
$\parity{2}{\lambda} = \hIsing{2}{\lambda}$.  Our new clone is
\begin{equation*}
    \Parev = \fclone{\{\parityk\mid k\text{ is even},
                                     \lambda\in\NonNegReals\}}\,.
\end{equation*}

\begin{lemma}
\label{lem:SD-Parev} 
$\SD \subseteq \ofclone{\Parev \cup\calB'_0}$.
\end{lemma}
\begin{proof}
Recall from Definition~\ref{def:parityfinite}
that 
$\paritybasis = \{\parityk[1/e],\parityk[1/2],\parityk[2],\parityk[e]\}$.
Let $\parityuptokbasis = \bigcup_{1\leq j \leq \lfloor k/2 \rfloor} \paritybasis[2j]$.
Recall from Definition~\ref{def:permissive} that
a function $F\colon\Bools^k\to\NonNegReals$ is \emph{permissive} if its range is~$\PosReals$. 
The proof splits into two parts.
First, we show that every $k$-ary permissive function in $\SD$ is 
a limit of  
$\fclone{\calB'_0 \cup  \parityuptokbasis}$.
Then we show the same for every other function in $\SD$.

{\bf Part One:} 
Consider any permissive $k$-ary  function~$F \in \SD$.
Let $f(\by) = \log F(\by)$.  (Note that $f$~is not necessarily in~$\calB$, since its range may include negative numbers; this is not a problem.) By the definition of the Fourier transform,
    \begin{equation*}
        f(\mathbf{y})
            = 2^{-k} \sum_{\bw\in \Bools^k} (-1)^{\wt{\bw\wedge\by}} \fhat(\mathbf{w}).
    \end{equation*}
      Exponentiating gives
    \begin{equation*}
        F(\by) = \prod_{\bw\in \Bools^k} \exp\big(
                     2^{-k} (-1)^{\wt{\bw \wedge \by}}\,\fhat(\bw)\big)\,.
    \end{equation*}
    Since $f$ is self-dual, we may restrict the  product to $\bw$~with even
    Hamming weight, since the odd-weight terms vanish by
    Lemma~\ref{lem:oddSD}.
Let     
    \begin{equation*}
        G_{\bw}(\by) = \exp\big(2^{-k} (-1)^{\wt{\bw\wedge\by}}\,\fhat(\bw)\big).
    \end{equation*}    
Then $F(\by) = \prod_{\bw\in \Bools^k: \text{$\wt{\bw}$ is even}} G_{\bw}(\by)$, so to
finish (using Lemma~\ref{lem:transitive}) we just have to show that, for any $\bw\in \Bools^k$ with even
Hamming weight,  $G_{\bw}$   is a limit of  
$\fclone{\calB'_0 \cup  \parityuptokbasis}$.

Consider any such $\bw$.
Let $j= \wt{\bw}/2$.
Given any $\by$, let $\bz$ be the 
arity-$2j$ tuple obtained from~$\by$ by
    deleting all positions that are $0$ in~$\bw$.  
Let
    \begin{equation*}
        G'_{\bw}(\bz) = \exp\big(2^{-k} (-1)^{\wt{\bz}}\,\fhat(\bw)\big).
    \end{equation*}    
Note that the arity-$k$ function $G_{\bw}$ is constructed from 
the arity-$2j$ function $G'_{\bw}$ by
adding fictitous arguments.    
We will show that  every arity-$2j$ function $G'_{\bw}$
is a limit of a function in  $\fclone{\calB_0' \cup \paritybasis[2j]}$.
This is all that we need since, by the closure of $\omega$-clones under
the addition of fictitious arguments, it also implies that 
$G_{\bw}$ is  a limit of  
$\fclone{\calB'_0 \cup  \parityuptokbasis}$.

 Now $G'_{\bw}(\bz)$
    is $\exp(2^{-k}\fhat(\bw))$ if $\wt{\bz}$ is even and it is
    $\exp(-2^{-k} \fhat(\bw))$ otherwise.  Therefore, $G'_{\bw}$ is 
$\lambda \times \parity{2j}{1/\lambda^2}$, where   $\lambda = \exp(2^{-k}\fhat(\bw))$.  
To finish note that $\lambda$ is a limit of   $\fclone{\calB'_0}$ (by Lemma~\ref{lem:finitegen})
and  $\parity{2j}{1/\lambda^2}$ is a limit of  $  \fclone{\paritybasis[2j]}$
(by Lemma~\ref{lem:finitegenparity}). Finally, their product is
a limit of $\fclone{\calB'_0 \cup \paritybasis[2j]}$ by Lemma~\ref{lem:transitive}.

{\bf Part Two:}
Consider any non-permissive $k$-ary function $F\in \SD$.
Let $G$ be the $(k+2)$-ary function
    \begin{equation*}
        G(\bx yz) = \begin{cases}
                    \ F(\bx) & \text{if }y\neq z\\
                    \ 1      & \text{if }y=z.
                    \end{cases}
    \end{equation*}
For every positive integer~$j$, let $H_j$ be the $k$-ary function
    \begin{equation*}
        H_j(\bx)
            = 2^{-(j+1)}\sum_{y,z\in\Bools} G(\bx yz)\,\big(\hIsing{2}{2}(yz)\big)^j.
                \end{equation*}
Note that $H_j$ is a good approximation for $F$ in the sense that
$H_j(\bx) = F(\bx) + 2^{-j}$.  Since $F$~is self-dual, so is~$H_j$; since $H_j(\bx) > F(\bx)\geq 0$, $H_j$ is permissive.
By Part One, $H_j$ is a limit of 
$\fclone{\calB'_0 \cup  \parityuptokbasis}$.

For any $\epsilon>0$, there is $j$ such that $2^{-j}<\epsilon$ and thus
$\|F-H_j\|_\infty<\epsilon$. By transitivity of limits (Lemma~\ref{lem:transitive}),
$F$ itself is a limit of $\fclone{\calB'_0 \cup  \parityuptokbasis}$.    
\end{proof}

We define the family of $k$-ary functions
\begin{equation*}
    \ForceOdd{k}(\bx) = \begin{cases}
                        \ 1 &\text{if $\wt{\bx}$ is odd} \\
                        \ 0 &\text{otherwise.}
                        \end{cases}
\end{equation*}
The function $\ForceOdd{4}$ turns out to be particularly useful.

 \begin{lemma}
\label{lem:Parev-Ianti} 
$\Parev \subseteq \ofclone{\AntiFerroIsing\cup\calB'_0\cup\{\ForceOdd{4}\}}$.
\end{lemma}
\begin{proof}
     For any $k>1$, we have
    \begin{equation*}
        \parityk(x_1, \dots, x_k)
            = \lambda \sum_{y\in\Bools} \ForceOdd{k}(x_1, \dots, x_{k-1}, y)\,
                                      \hIsing{2}{1/\lambda}(y,x_k).
    \end{equation*}
If $\lambda \leq 1$ then $\hIsing{2}{1/\lambda}$ is in 
$\AntiFerroIsing$.
Otherwise, by Lemma~\ref{lem:I-ferroI}, there is a $c\in \calB_0$
and a $\hIsing{2}{\lambda'} \in \AntiFerroIsing$ such that
 $\hIsing{2}{1/\lambda} \in \fclone{\{c, \hIsing{2}{\lambda'}\}}$.
    The nullary functions~$\lambda$ and~$c$ are
     limits of $\ofclone{\calB_0'}$ by Lemma~\ref{lem:finitegen}.

    We now show that, for any even $k>0$, $\ForceOdd{k}\in
    \fclone{\ForceOdd{4}}$.  For $k=2$, we have 
    $$\ForceOdd{2}(x,y) =
 \sum_{w,z} \ForceOdd{4}(x,y,w,z)\,\EQ(x,w)\,\EQ(x,z)\,,$$
    
     and the case $k=4$ is trivial.  For $k\geq
    6$,
    \begin{equation*}
        \ForceOdd{k}(x_1,\dots, x_k) = \sum_{y,z\in\Bools}
            \ForceOdd{4}(x_1, x_2, x_3, y)\,
            \ForceOdd{2}(y,z)\,
            \ForceOdd{k-2}(z,x_4,\dots,x_k)\,,
    \end{equation*}
    and the claim follows by induction on even~$k$.
    The lemma follows by Lemma~\ref{lem:transitive}.
 \end{proof}

\begin{lemma} 
\label{lem:FOdd-Ising}  
There exist $\lambda, \lambda'>1$ such that 
$\ForceOdd{4} \in \ofclone{\{\hIsing{2}{\lambda},\hIsing{2}{\lambda'}\}\cup \calB'_0}$.
\end{lemma}
\begin{proof}
Let $\lambda'=2$. (Any value that is greater than one would do, but we take $\lambda'=2$ for concreteness.)
Let $\lambda = \sqrt{ {\lambda'}^4 + \sqrt{{\lambda'}^8-1}}$. Note that $\lambda>1$.
Consider 
the (self-dual, symmetric)  function
$$
        f(x_1,x_2,x_3,x_4)
            = \sum_{y\in \{0,1\}} \prod_{i\in[4]}
                  \hIsing{2}{\lambda}(x_i,y)
                  \prod_{i\in[4],j\neq i\in[4]}
                      \hIsing{2}{\lambda'}(x_i,x_j)\,.
$$
Note that 
$f(0,0,0,0) = (\lambda^4+1) =
2 \lambda^2 {\lambda'}^4 =
f(0,0,1,1)=1023$.
Also,
$f(0,0,0,1) =    (\lambda+\lambda^3){\lambda'}^3\approx 1491.$
Now, define the function $g(\bx) = f(\bx)/f(0,0,0,1)$.  
Note that 
$f(\bx) \in \fclone{\{\hIsing{2}{\lambda},\hIsing{2}{\lambda'}\} }$
and $1/f(0,0,0,1) \in \ofclone{\calB'_0}$ by Lemma~\ref{lem:finitegen}. So,
by Lemma~\ref{lem:transitive}, for every positive integer~$j$,
$g^j \in  
\ofclone{\{\hIsing{2}{\lambda},\hIsing{2}{\lambda'}\}\cup \calB'_0}$.

We have
    $g(\bx)=1$ if $|\bx|$~is odd, and $g(\bx)<1$, otherwise.  This gives
    $\|g-\ForceOdd{4}\|_\infty = g(0,0,0,0) < 1$ so,
    for any $\epsilon>0$, we can choose an integer~$j$ such that
    $\|g^j-\ForceOdd{4}\|_\infty < \epsilon$.   
\end{proof}

\begin{theorem} \label{thm:SD-Ising} 
    $\SD = \Ianti$. 
\end{theorem} 
\begin{proof}
    Since every function in $\AntiFerroIsing\cup\calB_0$ is self-dual by
    definition and $\SD$
    is an $\omega$-clone by Theorem~\ref{thm:SDclone}, we have $\Ianti\subseteq\SD$.

We now show $\SD\subseteq\Ianti$. 
Lemmas~\ref{lem:SD-Parev} and~\ref{lem:Parev-Ianti} 
(together with Lemma~\ref{lem:transitive}) imply
$\SD \subseteq \ofclone{\AntiFerroIsing\cup\calB'_0\cup\{\ForceOdd{4}\}}$.
From this and Lemma~\ref{lem:FOdd-Ising} (together with Lemma~\ref{lem:transitive})
we have
$\SD \subseteq \ofclone{\AntiFerroIsing\cup\calB'_0}$.
\end{proof}

\subsection{Maximality}
\label{sec:SD-max}

\begin{lemma}\label{lem:SD-maximal}
$\SD$ is a maximal $\omega$-clone in $\calB$; i.e., for any function $f\not\in\SD$, $\ofclone{\SD\cup\{f\}}=\calB$.
\end{lemma}
\begin{proof}
Let $k$ be the arity of $f$.
First, we show that  $\ofclone{\SD \cup \{f\}}$ 
contains $\delta_0=[1,0]$ or~$\delta_1=[0,1]$.
If $f(\bone)>f(\bzero)$, we have $f(\bone)>0$ so the  nullary
function $f_1=1/f(\bone)$ is well-defined, and it is in
$\ofclone{\SD \cup \{f\}}$ since $\calB_0 \subseteq \SD$ by Theorem~\ref{thm:SDclone}.
In this case, $$\delta_1(x) = \lim_{n\to\infty} \sum_{x_2,\ldots,x_k} 
f(x_1,x_2,\ldots,x_k)^n \left(\prod_{i=1}^{k-1} \EQ(x_i,x_{i+1})\right) {\left(\frac{1}{f(\bone)}\right)}^n,$$
so $\delta_1(x) \in \ofclone{\SD \cup \{f\}}$.
If $f(\bone)<f(\bzero)$, we  similarly show $\delta_0 
\in \ofclone{\SD \cup \{f\}}$.

If $f(\bzero)=f(\bone)$ there is some $\ba\in\Bools^k$ such that $f(\ba)\ne f(\overline{\ba})$ so $k\geq 2$.
Because $\omega$-clones are closed under permuting arguments, we may assume without loss of generality that $a_1=\dots=a_\ell=0$ and $a_{\ell+1}=
\dots=a_k=1$.  Let $$g(x_1,x_k) = \sum_{x_2,\ldots,x_{k-1}} 
f(x_1, \ldots,x_k)  \left(\prod_{i=1}^{\ell-1} \EQ(x_i,x_{i+1})\right) \left(\prod_{i=\ell+1}^{k-1} \EQ(x_i,x_{i+1})\right).$$
Clearly, $g \in \fclone{f}$.
 The function~$g$ satisfies
$g(0,0)=g(1,1)$ and $g(0,1)\ne g(1,0)$. Set $h(x)=g(x,0)+g(x,1)$. We have
$h(0)=g(0,0)+g(0,1)\ne g(1,0)+g(1,1)=h(1)$ and $h\in \fclone{f}$. 
An argument similar to the one in the first paragraph of this proof shows
that $\delta_0$ or $\delta_1$ is in $\ofclone{\SD \cup \{h\}}$.
By Lemma~\ref{lem:transitive}, 
$\delta_0$ or $\delta_1$ is in $\ofclone{\SD\cup \{f\}}$.

For the rest of the proof, suppose that $\delta_0\in\ofclone {\SD \cup \{f\}}$; the case where
$\delta_1\in \ofclone{\SD\cup \{f\}}$ is very similar. Take any $g\in\calB$,
say of arity~$k$. Let   $h$ be the $(k+1)$-ary function defined as follows:
for any $\bx\in\Bools^k$, $h(\bx0)=h(\overline{\bx}1)=g(\bx)$. As is easily seen,
$h\in\SD$. It is also easy to see that $g(\bx)=\sum_y h(\bx y)\,\delta_0(y)$.
Thus, by Lemma~\ref{lem:transitive}, $g\in \ofclone{\SD \cup \{f\}}$.
\end{proof}

Next, we prove that $\SDP$ is a maximal $\omega$-clone in $\SD$.
Recall the functions $\parityk$ from Definition~\ref{defn:parity}.

\begin{theorem}\label{thm:SDP-maximal}
$\SDP$ is a maximal $\omega$-clone in $\SD$; i.e., for any $n$-ary self-dual
function $f\in\SD\setminus\SDP$, $\ofclone{\SDP\cup\{f\}}=\SD$.
\end{theorem}
\begin{proof}
Since $\calB_0\subseteq \SDP$
by Theorems~\ref{thm:SDclone} and~\ref{thm:P-clone}, it suffices by Theorem~\ref{thm:SD-Ising},
Corollary~\ref{cor:Ianti-fin} and Lemma~\ref{lem:transitive},
to show that $\hIsing{2}{\lambda}\in\ofclone{\SDP\cup\{f\}}$ for some
$\lambda>1$. Since $f\in\SD\setminus\SDP$, there is
some $\ba\in\{0,1\}^n$ such that $\fhat(\ba)<0$.
By Lemma~\ref{lem:neg-fc}, we may assume that $\ba=\bone$.
Then, by Lemma~\ref{lem:tech}, there is a permissive binary function
$h\in\fclone{\{f,\parity{n+2}{1/2}\}}$ such that $\hhat(1,1)<0$.

Because the $n$-ary function $f$~is self-dual and 
$\fhat(\bone)\neq 0$,
$n$~must be even by Lemma~\ref{lem:oddSD}.  For all even~$n$,
$\parity{n+2}{1/2}$ is self-dual.
Also, $\parity{n+2}{1/2}\in \calP$ by Lemma~\ref{lem:ft-parity},
so it is in $\SDP$ and, by Lemma~\ref{lem:transitive}, $h\in \SDP$.

Since $\hhat(1,1)<0$ and $h$~is permissive, there are constants
$c>b>0$ such that $h(0,0)=h(1,1)=b$ and $h(0,1)=h(1,0)=c$.  Therefore,
the function $(1/b)h$ is $\hIsing{2}{c/b}$, with $c/b>1$, and this function belongs
to $\ofclone{\SDP\cup\{f\}}$.
\end{proof}

\section{Match-circuits and even-circuits}
\label{sec:match-even}

\subsection{Match-circuits}

We first show that $\calM$ and $\calE$ are functional clones.   It is  still open 
whether they are $\omega$-clones.  As far as we know, there may be a function
in $\ofclone{\calM}$ that is the limit of a sequence  of functions~$f_\epsilon$ where
each $f_\epsilon$ is  implemented by 
a match-circuit with its own underlying graph. It  is not clear in this case
whether $f$ itself can be implemented by a match-circuit.
A similar comment applies to~$\calE$.

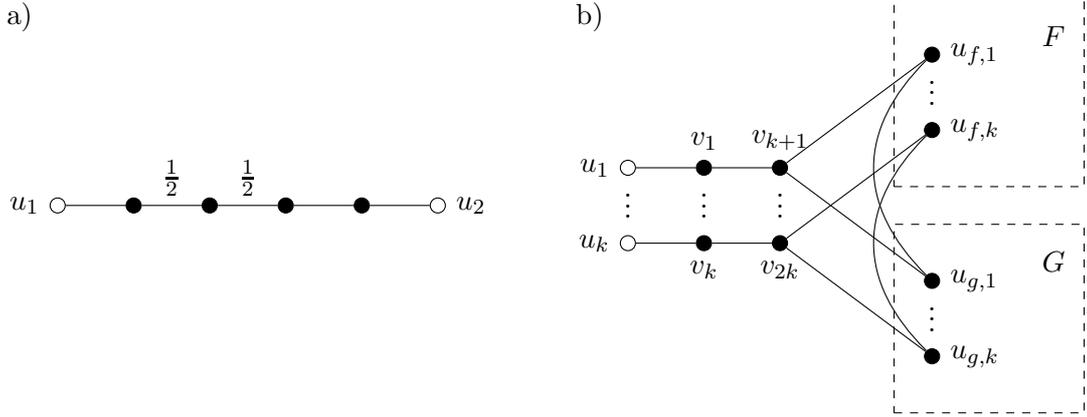
\begin{figure}
    \begin{center}
    \begin{tikzpicture}[scale=1,node distance = 1.5cm]
    \tikzstyle{vertex}=[fill=black, draw=black, circle, inner sep=2pt]
    \tikzstyle{dist}  =[fill=white, draw=black, circle, inner sep=2pt]

    \begin{scope}[shift={(-1.5,-0.5)}] 
        \node[dist]   (u1b) at (0,0) [label=180:$u_1$] {};
        \node[vertex] (v1b) at (1,0)                   {};
        \node[vertex] (v2b) at (2,0)                   {};
        \node[vertex] (v3b) at (3,0)                   {};
        \node[vertex] (v4b) at (4,0)                   {};
        \node[dist]   (u2b) at (5,0) [label=0:  $u_2$] {};

        \draw (u1b)-- (v1b) -- node [above] {$\tfrac12$} (v2b)
                   -- node [above] {$\tfrac12$} (v3b) -- (v4b) -- (u2b);
        \node at (-0.5,2.5) {a)};
    \end{scope}

    \begin{scope}[shift={(6,-0.5)}] 
        \node[dist]   (u1c) at (0, 0.5) [label=180:$u_1$] {};
        \node               at (0, 0.1)                   {$\vdots$};
        \node[dist]   (ukc) at (0,-0.5) [label=180:$u_k$] {};

        \node[vertex] (v1c) at (1, 0.5) [label=90:$v_1$] {};
        \node               at (1, 0.1)                   {$\vdots$};
        \node[vertex] (vkc) at (1,-0.5) [label=270:$v_k$] {};

        \node[vertex] (w1c) at (2, 0.5) [label=90:$v_{k+1}$] {};
        \node               at (2, 0.1)                     {$\vdots$};
        \node[vertex] (wkc) at (2,-0.5) [label=270:$v_{2k}$] {};

        \node[vertex] (f1c) at (4, 2) [label=0:$u_{f,1}$] {};
        \node               at (4, 1.6)                   {$\vdots$};
        \node[vertex] (fkc) at (4, 1) [label=0:$u_{f,k}$] {};
        \node[vertex] (g1c) at (4,-1) [label=0:$u_{g,1}$] {};
        \node               at (4,-1.4)                   {$\vdots$};
        \node[vertex] (gkc) at (4,-2) [label=0:$u_{g,k}$] {};

        \draw[dashed] (3.5, 2.75) rectangle (6, 0.25);
        \draw[dashed] (3.5,-0.25) rectangle (6,-2.75);
        \node at (5.6, 2.25) {$F$};
        \node at (5.6,-0.75) {$G$};

        \draw (u1c)--(v1c)--(w1c);  \draw (f1c)--(w1c)--(g1c);
        \draw (ukc)--(vkc)--(wkc);  \draw (fkc)--(wkc)--(gkc);

        \draw (f1c) .. controls (3,1) and (3,0) .. (g1c);
        \draw (fkc) .. controls (3,0) and (3,-1) .. (gkc);

        \node at (-0.5,2.5) {b)};

    \end{scope}
    \end{tikzpicture}
    \end{center}
    \caption{Match-circuits used in the proof of
      Theorem~\ref{thm:M-clone}. Every edge has weight~$1$ unless
      otherwise indicated.  
      a)~The equality function.  b)~The product of functions
      implemented by the match-circuits $F$ and~$G$.}
    \label{fig:M-clone}
\end{figure}

\begin{theorem}
\label{thm:M-clone}
$\fclone{\calM} = \calM$ and $\calB_0 \subseteq \calM$.
\end{theorem}
\begin{proof}
We show that $\calM$ contains   the
equality function and has all the closure properties required by
 the definition of functional clone.
\begin{itemize}
\item For the equality function, we have $\widehat{\EQ} =
\tfrac12\EQ$.  This function is implemented by the graph shown in
Figure~\ref{fig:M-clone}(a).

\item Permuting arguments corresponds directly to renaming the
terminals of the circuit, so it is clear that $\calM$~is closed under
this operation.

\item For closure under the introduction of fictitious arguments, let
$g(\bx z) = f(\bx)$ for some $k$-ary $f\in\mathcal{M}$.  By
Lemma~\ref{lem:fops}\pref{op-fict}, $\ghat(\bx0) = \fhat(\bx)$ and
$\ghat(\bx1)=0$.  The match-circuit for~$\ghat$ is the disjoint union
of the match-circuit~$F$ for~$\fhat$ and a weight-$1$ path on new
vertices $u_{k+1}$, $v$ and~$v'$ (in that order).  If $y_{k+1}$~is
assigned~$0$, then any perfect matching is the union of the edge
$(v,v')$ and a perfect matching of~$F$, so has weight $\fhat(y_1,\dots,
y_k)$; if $y_{k+1}$~is assigned~$1$, there is no perfect matching, so
the assignment has weight~$0$, as required.

\item For closure under summation, let $g(\bx)= \sum_z f(\bx z)$
for some $(k+1)$-ary function $f\in \mathcal{M}$.  By
Lemma~\ref{lem:fops}\pref{op-sum}, $\ghat(\bx) = 2\fhat(\bx0)$, so we obtain a
match-circuit for~$\ghat$ from the circuit~$F$ for~$\fhat$ by deleting
the vertex~$u_{k+1}$ (which is equivalent to forcing its adjacent edge
in~$F$ to be spin-$0$) 
and adding a new weight-$2$ edge between two new vertices
(which doubles the weight of any perfect matching).

\item For closure under products, let $h(\bx)=f(\bx)\,g(\bx)$ for $k$-ary
functions $f,g\in\mathcal{M}$.  Let $\fhat$ and~$\ghat$ be implemented
by match-circuits $F$ and~$G$, with terminal vertices $u_{f,1}, \dots,
u_{f,k}$ and $u_{g,1}, \dots, u_{g,k}$, respectively.  For each
$i\in[k]$, let $y_{f,i}$ be the unique edge adjacent to~$u_{f,i}$
in~$F$ and define~$y_{g,i}$ similarly in~$G$.  Recall that
    $\hhat(\by) = \sum_{\bw\in\Bools^k}
    \fhat(\bw)\,\ghat(\bw\oplus\by)$.

Let $\sigma$ be any assignment of spins $0$ and~$1$ to the edges of
the match-circuit~$H$ shown in Figure~\ref{fig:M-clone}(b).  We claim
that, if $\sigma$~is a perfect matching then, for all $i\in[k]$,
$\sigma(y_{f,i}) = \sigma(y_{g,i}) \oplus \sigma(y_i)$.

If $\sigma(y_i) = 0$ then we must have $\sigma(v_i,v_{k+i})=1$.  We
may have $\sigma(y_{f,i}) = 0$ or $\sigma(y_{f,i})=1$ but, in either
case, $\sigma(y_{f,i}) = \sigma(y_{g,i}) = \sigma(y_{g,i}) \oplus0$.

Otherwise, $\sigma(y_i) = 1$ and we must have $\sigma(v_i,v_{k+i})=0$.
Now there are two cases.  If $\sigma(v_{k+i},v_{f,i})=1$, then
$\sigma(y_{f,i}) = 0 = \sigma(y_{g,i})\oplus 1$; if
$\sigma(v_{k+i},v_{g,i})=1$, then $\sigma(y_{f,i}) = 1 =
\sigma(y_{g,i})\oplus 1$.  This completes the proof of the claim.

For any choice of spins $x_1, \dots, x_k$ for the edges $y_1, \dots,
y_k$ we can choose any spins $w_1, \dots, w_k$ for the edges $y_{f,1},
\dots, y_{f,k}$.  Doing so forces us to assign the spin $w_i\oplus
x_i$ to~$y_{g,i}$.  Therefore, the value computed by the match-circuit
is $\sum_{\bw}\fhat(\bw)\,\ghat(\bw\oplus\bx)$, as required. 
\end{itemize}
The fact that $\calB_0 \subseteq \calM$ comes from the definition of match-circuit.
Any positive $c\in \calB_0$ can be implemented by a match-circuit with no terminals containing one
edge with weight~$c$. The constant~$0$ is implemented by a match-circuit with no terminals
whose three edges form a $3$-cycle.
\end{proof}

\begin{theorem}
\label{thm:Iferro-MHferro}
    $\Iferro \subseteq \ofclone{M}\cap \Hferro$.
\end{theorem}
\begin{proof}
    It is trivial that $\Iferro\subseteq \Hferro$, so it remains to
    prove that $\Iferro \subseteq \ofclone{\calM}$.

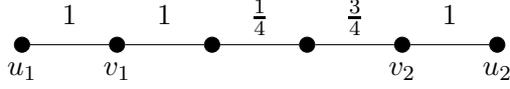
\begin{figure}
    \begin{center}
    \begin{tikzpicture}[scale=1.25,node distance=1.5cm]
    \tikzstyle{vertex}=[fill=black, draw=black, circle, inner sep=2pt]

    \node[vertex] (u1) at (0,0) [label=270:$u_1$] {};
    \node[vertex] (v1) at (1,0) [label=270:$v_1$] {};
    \node[vertex] (w1) at (2,0)                   {};
    \node[vertex] (w2) at (3,0)                   {};
    \node[vertex] (v2) at (4,0) [label=270:$v_2$] {};
    \node[vertex] (u2) at (5,0) [label=270:$u_2$] {};

    \draw (u1) -- node [above] {$\vphantom{\tfrac34}1$} (v1)
               -- node [above] {$\vphantom{\tfrac34}1$} (w1)
               -- node [above] {$\tfrac14$}             (w2)
               -- node [above] {$\tfrac34$}             (v2)
               -- node [above] {$\vphantom{\tfrac34}1$} (u2);
    \end{tikzpicture}
    \end{center}
    \caption{The match-circuit with terminals $(u_1,v_1)$ and $(u_2,v_2)$, used in the proof of Theorem~\ref{thm:Iferro-MHferro}.}
    \label{fig:IFerro-MHferro}
\end{figure}

    By Corollary~\ref{cor:Iferro-fin}, it suffices
    to show that $f = \hIsing{2}{1/2}\in\calM$.  By
    Lemma~\ref{lem:ft-hIsing}, we have $\fhat = [\tfrac34, 0,
    \tfrac14]$.  This is implemented by the match-circuit shown in Figure~\ref{fig:IFerro-MHferro}.
\end{proof}

We do not know whether the inclusion in the statement of
Theorem~\ref{thm:Iferro-MHferro} is strict.  This corresponds to the
dotted line in Figure~\ref{fig:lattice}.

\subsection{Even-circuits}

The proof that $\calE$ is a functional clone is similar to Theorem~\ref{thm:M-clone}. 
\begin{theorem}
\label{thm:E-clone}
$\fclone{\calE} = \calE$ and $\calB_0 \subseteq \calE$.
\end{theorem}
\begin{proof}
We show that $\calE$ contains   the
equality function and has all the closure properties required by
 the definition of functional clone.
\begin{itemize}

\item For the equality function, we have $\widehat{\EQ} =
\tfrac12\EQ$.   The definition of $\calE$ accounts for the multiplication by~$1/2$, so we need
only show that $\EQ$ is implemented by an even-circuit. Indeed, it is implemented by a three-edge
path between two terminals (where all edges have weight~$1$).

\item Permuting arguments corresponds directly to renaming the
terminals of the circuit, so it is clear that $\calE$~is closed under
this operation.

\item For closure under the introduction of fictitious arguments, let
$g(\bx z) = f(\bx)$ for some $k$-ary $f\in\mathcal{E}$.  By
Lemma~\ref{lem:fops}\pref{op-fict}, $\ghat(\bx0) = \fhat(\bx)$ and
$\ghat(\bx1)=0$.  The even-circuit for~$\ghat$ is the disjoint union
of the even-circuit~$F$ for~$\fhat$ and a weight-$1$ edge on new
vertices $u_{k+1}$  and~$v_{k+1}$.  
Even subgraphs with $y_{k+1}=0$ correspond to even subgraphs of~$F$.
There are no even subgraphs with $y_{k+1}=1$.

\item For closure under summation, let $g(\bx)= \sum_z f(\bx z)$
for some $(k+1)$-ary function $f\in \mathcal{E}$.  By
Lemma~\ref{lem:fops}\pref{op-sum}, $\ghat(\bx) = 2\fhat(\bx0)$, so we obtain an
even-circuit for~$\ghat/2$ from the circuit~$F$ for~$\fhat$ by deleting
the vertex~$u_{k+1}$ (which is equivalent to forcing its adjacent edge
in~$F$ to be spin-$0$).

\item For closure under products, let $h(\bx)=f(\bx)\,g(\bx)$ for $k$-ary
functions $f,g\in\mathcal{E}$.  Let $\fhat$ and~$\ghat$ be implemented
by even-circuits $F$ and~$G$, with terminal vertices $u_{f,1}, \dots,
u_{f,k}$ and $u_{g,1}, \dots, u_{g,k}$, respectively.  For each
$i\in[k]$, let $y_{f,i}$ be the unique edge adjacent to~$u_{f,i}$
in~$F$ and define~$y_{g,i}$ similarly in~$G$.  Recall that,
by~\eqref{eq:convolution},
    $\hhat(\by) = \sum_{\bw\in\Bools^k} 
    \fhat(\bw)\,\ghat(\bw\oplus\by)$.

 Let $H$ be the even-circuit that is the same as the one shown
in Figure~\ref{fig:M-clone}(b)
except that the edges $(u_{f,j}, u_{g,j})$ are deleted.
Let $\sigma$ be any assignment of spins $0$ and~$1$ to the edges of~$H$.  We claim
that, if $\sigma$~is an even subgraph then, for all $i\in[k]$,
$\sigma(y_{f,i}) = \sigma(y_{g,i}) \oplus \sigma(y_i)$.  

If $\sigma(y_i) = 1$ then we must have $\sigma(v_i,v_{k+i})=1$.  
Thus, exactly one of the edges $(v_{k+i},u_{f,i})$ and $(v_{k+i},u_{g,i})$ has
spin one. So $\sigma(y_{f,i})$ and $\sigma(y_{g,i})$ differ.

Otherwise, $\sigma(y_i) = 0$ so $\sigma(v_i,v_{i+i})=0$ 
so $\sigma(y_{f_i})$ and $\sigma(y_{g,i})$ agree. This completes the proof of the claim.

For any choice of spins $x_1, \dots, x_k$ for the edges $y_1, \dots,
y_k$ we can choose any spins $w_1, \dots, w_k$ for the edges $y_{f,1},
\dots, y_{f,k}$.  Doing so forces us to assign the spin $w_i\oplus
x_i$ to~$y_{g,i}$.  Therefore, the value computed by the  even-circuit
is $\sum_{\bw}\fhat(\bw)\,\ghat(\bw\oplus\bx)$, as required. 
\end{itemize}
The fact that $\calB_0 \subseteq \calE$ comes from the definition of  even-circuit.
The  nullary zero function $f=0$ is in $\calE$, since $f=0\cdot g$ for any
function~$g$ implemented by an even-circuit. Any non-zero  nullary function
$f=c$ is in $\calE$, since $f=c\cdot 1$, where $1$ is the constant one function
implemented by the empty graph, whose unique even subgraph is the empty graph,
which has weight 1.  
\end{proof}

The next theorem shows that the functional clone~$\calE$ is
the same as the clone generated by nullary functions and ferromagnetic Ising model interactions.
Something very close to this equivalence is seen in the ``high-temperature
expansion'' of the Ising model, first elucidated by Van der Waerden~\cite{VdWaerden}. 
In our proof, we employ the framework of holants and holographic transformations.
See Cai, Lu and Xia~\cite{planarCSP} for the wider context, 
particularly the introduction to that paper and Theorem~IV.1.

\begin{theorem}
\label{thm:Ferro-E}
    $\fclone{\FerroIsing\cup \calB_0} = \calE$.
\end{theorem}
\begin{proof}
Theorem~\ref{thm:E-clone} shows that $\calB_0 \subseteq \calE$.
It is also true that $\hIsing{2}{1}\in \calE$ since it can be constructed from~$\calB_0$ by
introducing fictitious arguments.
To see that $\FerroIsing \subseteq \calE$ 
consider any function 
$\hIsing{2}{\lambda}$ with $0 \leq \lambda < 1$.
Let $f(x_1,x_2) = (2/(1+\lambda)) \hIsing{2}{\lambda}$,
 By Lemma~\ref{lem:ft-hIsing}, $\fhat(0,0)=1$, $\fhat(0,1)=\fhat(1,0)=0$
and $\fhat(1,1) = (1-\lambda)/(1+\lambda)$.
But $\fhat$ can be implemented by an even-circuit consisting of a three-edge
path between two terminals in which the middle edge has weight $(1-\lambda)/(1+\lambda)$.
Thus, $\hIsing{2}{\lambda} \in \fclone{\calE}$.  
By Lemma~\ref{lem:transitive},     $\fclone{\FerroIsing\cup \calB_0} \subseteq \fclone{\calE}$, so by Theorem~\ref{thm:E-clone},
$\fclone{\FerroIsing\cup \calB_0} \subseteq  \calE$.

We now  show
that  $\calE \subseteq \fclone{\FerroIsing\cup \calB_0}$.
It is obvious that any nullary function is in $\fclone{\FerroIsing\cup \calB_0}$
so consider any function $g\in \calE$ with arity $k\geq 1$.
From the definition of $\calE$, 
$g = c \cdot f$ for some non-negative real number~$c$
and $\fhat$ is implemented by an even-circuit~$G$.

We  will show that $f$ is in $\fclone{\FerroIsing\cup \calB_0}$, which implies that $g$ is 
also in  
$\fclone{\FerroIsing\cup \calB_0}$.
To do this, we start by viewing the even-circuit~$G$
as an instance of a holant problem. A holant problem \cite{planarCSP}
consists of a graph in which, for all~$d$, every degree-$d$ vertex~$v$ is equipped with
a function  $f_v \in \calB_d$.
A configuration assigns spins $0$ and $1$ to the edges, and
the weight of a configuration is the product, over all vertices~$v$, of $f_v(\bx)$, where $\bx$ is
the string of spins of edges around~$v$ (in some appropriate order). The partition function is the sum of the weights of the configurations.

To represent the relevant holant problem cleanly,
we first construct $G'$ from $G$
by two-stretching the internal edges of~$G$ (turning them all into two-edge paths).
That is, if $G$ contains an edge $e=(v_i,v_j)$, 
we add a new vertex $v_{e}$ to $G'$.
We view $G'$ has a holant problem, so configurations assign spins $0$ and $1$ to the edges of $G'$.
At each vertex $v_i$ of $G'$ we 
add a function $f_{v_i}$ which is $1$ if an even number of its arguments have spin-$1$
and is $0$ otherwise.
At each new vertex $v_e$ of $G'$ we add 
a function $f_{v_e}$ which is the symmetric arity-$2$ function $[1,0,w_e]$.
This ensures that, in configurations with non-zero weight, 
the two edges adjacent to the new vertex $v_e$ get the same spin
(so non-zero configurations of  $G'$ correspond to even subgraphs of $G$). It also ensures
also that the weight $w_e$ of the edge $e$ of the even-circuit~$G$ is accounted for.
It is easy to see that the partition function
of the holant problem $G'$ is the same as the function implemented by the even-circuit~$G$,
which is $\fhat$.

Now we apply a standard trick from the holant literature.
Let $H = \tfrac12[1,1,-1]$ be the symmetric arity-$2$ Hadamard/FFT function.
Construct  a new holant instance $G''$ from $G'$ by three-stretching every edge of $G'$ 
and equipping every new vertex with the function~$H$. Since $H = \tfrac12 H^{-1}$, 
the partition function of $G''$ is $2^{-|E(G')|}$ times the partition function of~$G'\!$, so it is, up to a constant factor, $\fhat$.

Now construct a new holant problem $G'''$ from $G''$ by considering 
all of the original vertices of $G'$.
\begin{itemize}
 \item
For any arity-$d$ vertex $v_i$ (which is an original internal vertex of $G$),
replace the subgraph consisting of $v_i$ 
(with its ``arity-$d$ even parity'' function)
and all of its neighbours 
(which have $H$ functions) with an equivalent arity-$d$ vertex
equipped with the arity-$d$ equality function. 
This leaves the partition function unchanged.

\item For any vertex $v_e$ (one of the new nodes with $[1,0,w_e]$ functions added in the construction of $G'$)
 let $\lambda_e = (1-w_e)/(1+w_e)$ and
 replace $v_e$ together with its two neighbours (which have $H$ functions)
 with the equivalent degree-$2$ vertex whose function is   $\tfrac14(1+w_e) [1,\lambda_e,1]$.
Again, this does not change the partition function.
\item The only remaining 
vertices with
$H$  functions are adjacent to the external vertices of~$G$.
Replacing these functions with arity-$2$ equality, we obtain a holant problem $G'''$ whose partition
function is the Fourier transform of that of $G''$. Thus, its partition function is~$f$.
\end{itemize}
 
 We now have a holant problem $G'''$ implementing $f$, up to a constant factor.
 All  of the functions at the vertices of $G'''$ are   equality (of any arity) or 
 $\tfrac14(1+w_e) [1,\lambda_e,1]$   for some $0\leq \lambda_e < 1$ so they are all in $\fclone{\FerroIsing \cup \calB_0}$.
Moreover, the equality constraints correspond to the original internal vertices $v_i$ of $G$
and the new Ising constraints correspond to edges between internal vertices of $G$.
Thus, $G'''$ implements a sum (over the spins of the internal vertices of $G$)
of a product (over the spins of the internal edges of $G$) of  
constraints in $\fclone{\FerroIsing \cup \calB_0}$. This shows that $f$ is in the closure of 
$\fclone{\FerroIsing \cup \calB_0}$ under product and summation, so $f$ itself is in 
 $\fclone{\FerroIsing \cup \calB_0}$. 
\end{proof}

Recall that $\Iferro = \ofclone{\FerroIsing \cup \calB_0}$. 
The following corollary follows immediately from
Theorem~\ref{thm:Ferro-E} and Theorem~\ref{thm:E-clone} using the definition of an $\omega$-clone.

\begin{corollary}
\label{cor:Iferro-E}
    $\Iferro = \ofclone{\calE}$.
\end{corollary}

\subsection{Relationship of $\boldsymbol{\ofclone{\calM}}$ with other clones}

In this section, we give, in Lemma~\ref{lem:matchineq-a}, a necessary condition
for a $4$-ary function to be in~$\calM$. Moreover, we show, in
Lemma~\ref{lem:matchineq-b}, that for symmetric functions this condition is also
sufficient. We then use these results to study the relationship between
$\ofclone{\calM}$ and the clones around it in the lattice~$\calL'\!$.

\begin{lemma}\label{lem:matchineq-a}
For every $4$-ary function $f\in\calM$,
\begin{equation}
\label{eq:matchineq}
\fhat(0011)\,\fhat(1100) + \fhat(0101)\,\fhat(1010) + \fhat(0110)\,\fhat(1001) \geq \fhat(0000)\,\fhat(1111)\,.
\end{equation}
\end{lemma}
\begin{proof} 
Consider an arity-$4$ match-circuit~$G$ that implements $\fhat$ as described in Definition~\ref{def:MC}.
Let $S = \{u_1,u_2,u_3,u_4\}$ be the set of external vertices of~$G$.
For $A\subseteq S$,  let $M_A$ denote the set of perfect matchings which include terminals adjacent to~$A$
(by assigning them spin~$1$) and exclude terminals adjacent to $S\setminus A$ (by assigning them spin~$0$).
  We exhibit an injective map 
\begin{equation*}
\nu\colon M_\emptyset\times M_S\to M_{\{u_1,u_2\}}\times M_{\{u_3,u_4\}}
\cup M_{\{u_1,u_3\}}\times M_{\{u_2,u_4\}}\cup M_{\{u_1,u_4\}}\times M_{\{u_2,u_3\}}
\end{equation*}
which is weight-preserving in the sense that, for matchings $m_1,
\dots, m_4$ with $\nu(m_1, m_2) = (m_3,m_4)$, we have $w(m_1)\,w(m_2) =
w(m_3)\,w(m_4)$.  The existence of~$\nu$ implies~\eqref{eq:matchineq}.

Given $(m_1,m_2)\in M_\emptyset\times M_S$, consider $m_1\oplus m_2$
and note that this is a collection of cycles together with two paths
$\pi$ and $\pi'\!$.  Let $\pi_1$ be the path connecting  vertex~$u_1$
to one of the other external vertices; the other path  connects the remaining external vertices.
If $\pi$ joins $u_1$ to~$u_2$, then $m_3:=m_1\oplus\pi\in M_{\{u_1,u_2\}}$ and
$m_4:=m_2\oplus\pi\in M_{\{u_3,u_4\}}$, with similar claims for $\pi$ joining
$u_1$ to $u_3$ or~$u_4$.  The construction is invertible, since $m_3\oplus
m_4 = m_1\oplus m_2$, from which we can recover~$\pi$ and, hence, $m_1$
and~$m_2$.  Therefore, $\nu\colon (m_1,m_2)\mapsto(m_3,m_4)$ is an
injection as claimed.

To see that $\nu$~is weight-preserving, observe that the edges
of~$\pi$ each appear in exactly one of $m_1$ and~$m_2$ and in exactly
one of $m_3$ and~$m_4$ and that, for $i\in\{1,2\}$, $m_i\setminus\pi =
m_{i+2}\setminus\pi$.
\begin{equation*}
    w(m_1)\,w(m_2)
        = \prod_{e\in m_1\setminus\pi}  w_e
           \prod_{e\in m_2\setminus\pi} w_e
           \prod_{e\in \pi} w_e
        = \prod_{e\in m_3\setminus\pi}  w_e
           \prod_{e\in m_4\setminus\pi} w_e
           \prod_{e\in \pi} w_e
        = w(m_3)\,w(m_4)\,.\qedhere
\end{equation*}
\end{proof}

We give the converse of Lemma~\ref{lem:matchineq-a} for symmetric functions.

\begin{lemma}
\label{lem:matchineq-b}
    If $f$ is a symmetric, arity-$4$, self-dual function such that
    \begin{equation}
    \label{eq:matchineq-sym}
        3 \fhat(0011)^2 \geq \fhat(0000)\,\fhat(1111)\,,
    \end{equation}
    then $f\in\calM$.
\end{lemma}

Note that~\eqref{eq:matchineq-sym} is just~\eqref{eq:matchineq}
specialised to symmetric functions.

\begin{figure}
    \begin{center}
    \begin{tikzpicture}[scale=1,node distance = 1.5cm]
    \tikzstyle{vertex}=[fill=black, draw=black, circle, inner sep=2pt]
    \tikzstyle{dist}  =[fill=white, draw=black, circle, inner sep=2pt]

    \begin{scope}[shift={(-6,4)}] 
        \node at (-0.75,-1) {a)};

        \foreach \i in {1,2,3,4} {
            \node[dist]   (au\i) at (0.0,-\i) [label=270:$u_\i$] {};
            \node[vertex] (av\i) at (1.5,-\i) [label=270:$v_\i$] {};
            \node[vertex] (aw\i) at (3.0,-\i)                    {};
            \draw (au\i) -- (av\i) -- (aw\i);
        }
        \node at (2.25,-0.7) {$C_0$};
    \end{scope}

    \begin{scope}[shift={(1.5,0)}]
        \node at (-2.75,3) {b)};
        \node[vertex] (bx1) at (4,3)                [label= 90:$x_1$] {};
        \node[vertex] (bx2) at ($(bx1)+(210:1.75)$) [label= 90:$x_2$] {};
        \node[vertex] (bx3) at ($(bx2)+(270:1.75)$) [label=270:$x_3$] {};
        \node[vertex] (bx4) at ($(bx3)+(330:1.75)$) [label=270:$x_4$] {};
        \node[vertex] (bx5) at ($(bx4)+( 30:1.75)$) [label=330:$x_5$] {};
        \node[vertex] (bx6) at ($(bx5)+( 90:1.75)$) [label= 30:$x_6$] {};

        \node (temp1) at ($(bx2)+(-1.5,0)$) {};
        \node[vertex] (bw1) at (temp1 |- bx1) [label=270:$w_1$] {};
        \node[vertex] (bw2) at (temp1 |- bx2) [label=270:$w_2$] {};
        \node[vertex] (bw3) at (temp1 |- bx3) [label=270:$w_3$] {};
        \node[vertex] (bw4) at (temp1 |- bx4) [label=270:$w_4$] {};

        \node (temp1) at ($(bx2)+(-3.0,0)$) {};
        \node[vertex] (bv1) at (temp1 |- bx1) [label=270:$v_1$] {};
        \node[vertex] (bv2) at (temp1 |- bx2) [label=270:$v_2$] {};
        \node[vertex] (bv3) at (temp1 |- bx3) [label=270:$v_3$] {};
        \node[vertex] (bv4) at (temp1 |- bx4) [label=270:$v_4$] {};

        \node (temp2) at ($(bx2)+(-4.5,0)$) {};
        \node[dist]   (bu1) at (temp2 |- bx1) [label=270:$u_1$] {};
        \node[dist]   (bu2) at (temp2 |- bx2) [label=270:$u_2$] {};
        \node[dist]   (bu3) at (temp2 |- bx3) [label=270:$u_3$] {};
        \node[dist]   (bu4) at (temp2 |- bx4) [label=270:$u_4$] {};

        \foreach \i in {1, ..., 5} {
            \foreach \j in {\i, ..., 6} {
                \draw (bx\i) -- (bx\j);
            }
        }
        \foreach \i in {1, 2, 3, 4} {
            \draw (bu\i) -- (bv\i) -- (bw\i) -- (bx\i);
        }
 
        \node at ($(bx5)!0.5!(bx6) + (0.3,0)$) {$\lambda$};
        \foreach \i in {1, 2, 3} {
            \node at ($(bw\i)!0.5!(bx\i) + (0,0.3)$) {$\mu$};
        }
        \node at ($(bw4)!0.5!(bx4) - (0,0.3)$) {$\mu$};
    \end{scope}

    \end{tikzpicture}
    \end{center}
    \caption{Match-circuits used in the proof of
      Lemma~\ref{lem:matchineq-b}. Every edge has weight~$1$ unless
      otherwise indicated.  
      a)~The case $C_0>0$, $C_2=C_4=0$.  b)~The case $C_0,C_2>0$,
      $C_4\geq 0$.}
    \label{fig:matchineq-b}
\end{figure}
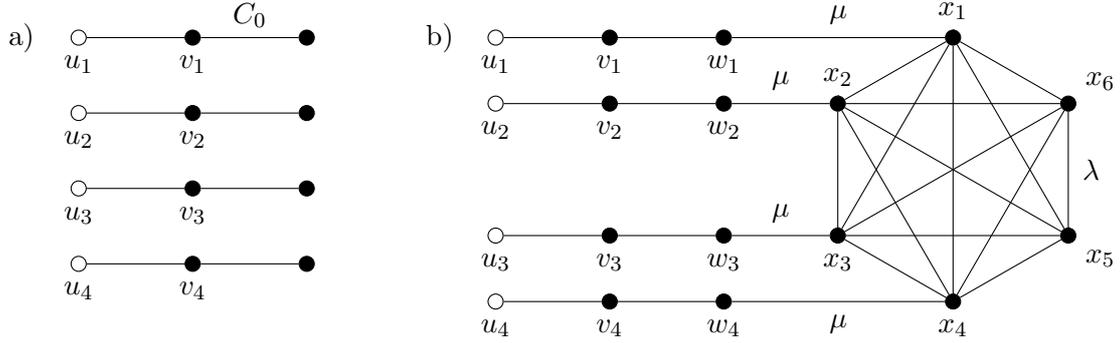

\begin{proof} 
For ease of notation, let $C_0 = \fhat(0,0,0,0)$, $C_2 =
\fhat(0,0,1,1)$ and $C_4 = \fhat(1,1,1,1)$.  Since $C_0 =
\tfrac1{16}\sum_{\bz\in\{0,1\}^4}f(\bz)$ is a sum of nonnegative
terms, if $C_0=0$, then $f$~is the constant zero function, which is
in~$\calM$ by Theorem~\ref{thm:M-clone}.  For the rest of the proof,
we assume that $C_0>0$.

If $C_2=C_4=0$, then $\fhat$~is implemented by the match-circuit shown
in Figure~\ref{fig:matchineq-b}(a).  If at most one of $C_2$ and~$C_4$
is zero, then \eqref{eq:matchineq-sym} implies that $C_2>0$.

We will construct a match-circuit $G$ for~$\fhat$ (see
Figure~\ref{fig:matchineq-b}(b)).
In addition to the terminal edges $y_i = (u_i,v_i)$ for $i\in [4]$,
 $G$ will have edges $(v_i,w_i)$ and $(w_i,x_i)$.
It will also contain a clique on the six vertices $x_1,x_2,x_3,x_4,x_5,x_6$.
The edge $(x_5,x_6)$ has weight~$\lambda$ and the
edges $(w_i,x_i)$ have weight~$\mu$. All other edges have weight~$1$.

Let $S =\{u_1,u_2,u_3,u_4\}$.
Following the proof of Lemma~\ref{lem:matchineq-a}, 
for $A\subseteq S$,  $M_A$ denotes the set of perfect matchings which include terminals adjacent to~$A$
(by assigning them spin~$1$) and exclude terminals adjacent to $S\setminus A$ (by assigning them spin~$0$).
Let $Z_A$ denote the sum of the weights of the perfect matchings in $M_A$.

The perfect matchings in $M_\emptyset$ contain all of the edges $(v_i,w_i)$
and none of the edges $(u_i,v_i)$ or $(w_i,x_i)$.
There are three perfect matchings of the clique that include the weight-$\lambda$ edge $(x_5,x_6)$ 
and twelve perfect matchings of the clique that do not include this edge,
so $Z_\emptyset=3\lambda+12$.
Similarly, the single perfect matching in $M_{S}$ contains all of the edges $(u_i,v_i)$
and $(w_i,x_i)$ (which have weight~$\mu$) and none of the edges $(v_i,w_i)$. The edge $(x_5,x_6)$
is present, so $Z_S=\lambda\mu^4$.
Finally, the perfect matchings in  
$M_{\{u_1,u_2\}}$ contain 
the two weight-$\mu$ edges $(w_1,x_1)$ and $(w_2,x_2)$ but not the
two weight-$\mu$ edges $(w_3,x_3)$ and $(w_4,x_4)$.
There are three matchings of the 4-clique containing $x_3,x_4,x_5,x_6$, one of which has weight~$\lambda$,
so $Z_{\{u_1,u_2\}} = (\lambda+2)\mu^2$ and the same is true for $Z_A$ for any other
size-two set $A\subseteq S$. 
Now let 
$$z(\lambda)=\frac
{Z_\emptyset Z_S}
{Z_{\{u_1,u_2\}} Z_{\{u_1,u_2\}}}.$$
Note that 
$z(\lambda) ={\lambda(3\lambda+12)}/{(\lambda+2)^2}$,
and that the range of~$z(\lambda)$ includes the interval $[0,3)$.
There are now three cases.

If $3C_2^2>C_0C_4>0$, we can choose~$\lambda$ so that $z(\lambda)=C_0C_4/C_2^2$.
Now choose~$\mu$ to obtain $(Z_\emptyset,Z_{\{u_1,u_2\}},Z_S)\propto(C_0,C_2,C_4)$.
In order to get the constant multiple correct, $G$ can be supplemented with an additional edge.

If $3C_2^2>C_0C_4=0$ (so $C_4=0$, since $C_0$ and~$C_2$ are positive),
we can simulate $\lambda=0$ by deleting the edge $(x_5,x_6)$.  We have
$Z_\emptyset=12$ and $Z_{\{u_1,u_2\}}=2\mu^2\!$ and, as before, we can
choose~$\mu$ so that $(Z_\emptyset,Z_{\{u_1,u_2\}})\propto(C_0,C_2)$
and add an edge to~$G$ for the required constant multiple.

Finally, if $3C_2^2=C_0C_4>0$, we must
achieve $z(\lambda)=3$.  This can be done by effectively setting $\lambda=\infty$ by removing the 
vertices $x_5$ and~$x_6$ and their incident edges.
\end{proof}

\begin{theorem}
\label{thm:M-SDP}
    $\ofclone{\calM} \subset \SDP$.
\end{theorem}
\begin{proof}
    Let $G$ be a match-circuit with terminals $y_1, \dots, y_k$, where
    $y_i = (u_i,v_i)$  for each $i\in[k]$, which implements the function
    $\fhat(\by)$.  For any assignment $\ba$ to~$\by$,
    $\fhat(\ba)$~is the total weight of the perfect matchings of the
    graph $G-\{v_i\mid a_i=1\}$.  Since a graph with an odd number of
    vertices has no perfect matchings, $\fhat(\ba)=0$ whenever $\wt{\ba}$ is odd
    or $\fhat(\ba)=0$ whenever $\wt{\ba}$ is even. 
    (This is the
    so-called ``parity condition'' of match-circuits; see,
    e.g.,~\cite{CG2014:matchgates}.) 

    We first show that $\calM\subseteq\SDP$.
    Suppose that $f\in\calM$.  
We will show that $f \in \SDP$.
Since Theorems~\ref{thm:SDclone} and~\ref{thm:P-clone}
guarantee that $\calB_0 \in \SDP$, we can assume
without loss of generality that the arity, $k$, of~$f$ is positive.
    Since $\fhat$~is implemented by a
    match-circuit, $\fhat(\bx)\geq 0$ for all~$\bx$, so
$f\in \calP$.
If $f$~is the constant arity-$k$ zero function, then
    $f\in\SDP$ trivially,
    so assume that
    $f(\ba)>0$ for at least one~$\ba\in\Bools^k\!$.  This implies that
    $\fhat(\bzero) = 2^{-k}\sum_{\bx} f(\bx) > 0$ so, by the parity
    condition, $\fhat(\bx) = 0$ whenever $\wt{\bx}$~is odd.  By
    Lemma~\ref{lem:oddSD}, $f\in\SD$, and we have established that
    $f\in\SD\cap \calP$.  But $\SDP = \SD \cap \calP$ by
    Theorem~\ref{thm:SDP}, so $f\in \SDP$.

    To show that $\calM\subset\SDP$, consider the function
    $g=\hIsing{4}{1/4} = [1,\tfrac14, \tfrac14, \tfrac14, 1]$, which
    is in~$\SD$.  By Lemma~\ref{lem:ft-hIsing},
    $\ghat(0000) = \tfrac{11}{32}$, $\ghat(\bx) = \tfrac{3}{32}$ when
    $\wt{\bx}\in\{2,4\}$ and $\ghat(\bx) = 0$, otherwise, so
    $g\in\calP\cap\calPN=\SDP$.  However, $g\notin\calM$ by
    Lemma~\ref{lem:matchineq-a}.

    It remains to ``lift'' the result to $\omega$-clones.  We have
    $\ofclone{\calM}\subseteq \ofclone{\SDP} = \SDP$, so it is enough
    to show that $g=\hIsing{4}{1/4}\notin\ofclone{\calM}$.  Suppose,
    towards a contradiction, that $g\in\ofclone{\calM}$.  By the
    definition of $\omega$-clones, for every $\epsilon>0$, there is a
    function $f\in\fclone{\calM}$ such that $\|f-g\|_\infty < \epsilon/32$.  
 Since $\calM$ is a functional clone by Theorem~\ref{thm:M-clone}, 
 $f\in \calM$.   
    Then
    Lemma~\ref{lem:fops}\pref{op-lim} implies that $\|\fhat-\ghat\|_\infty <
    \epsilon/32$, also.  Thus, for all~$\bx\in\Bools^4\!$,
    \begin{equation}
    \label{eq:fgbound}
        \ghat(\bx) - \epsilon/32 < \fhat(\bx)
            < \ghat(\bx) + \epsilon/32\,.
    \end{equation}
    Since $f\in\calM$, \eqref{eq:matchineq} must hold.  Plugging
    \eqref{eq:fgbound} and the values of~$\ghat$
    into~\eqref{eq:matchineq} gives     \begin{equation*}
        3\big(\tfrac{3}{32} + \tfrac{\epsilon}{32}\big)^2
            > \big(\tfrac{11}{32} - \tfrac{\epsilon}{32}\big)
                  \big(\tfrac{3}{32} - \tfrac{\epsilon}{32}\big)\,,
    \end{equation*} 
but this only holds
     for sufficiently large
    positive values of~$\epsilon$, contradicting the assumption that
    $g\in\ofclone{\calM}$.
\end{proof}

Lemmas~\ref{lem:matchineq-a} and~\ref{lem:matchineq-b} also allow
us to separate $\Hferro$ from~$\ofclone{\calM}$.

\begin{lemma}
\label{lem:MHferro-incomp}
    $\ofclone{\calM}$ and $\Hferro$ are incomparable under~$\subseteq$.
\end{lemma}
\begin{proof}
    Let $f = [13,4,1,4,13]$.  We saw in the proof of
    Theorem~\ref{thm:Hferro-SDP} that $f\notin \Hferro$ and that
    $\fhat = [4,0,\tfrac32,0,0]$.  However, $f\in\ofclone{\calM}$ by
    Lemma~\ref{lem:matchineq-b}.

    Now, let $g = \hIsing{4}{1/2}\in\Hferro$.  By
    Lemma~\ref{lem:ft-hIsing}, $\ghat=[\tfrac{9}{16}, 0,
    \tfrac{1}{16}, 0, \tfrac{1}{16}]$ and $g\notin \ofclone{\calM}$ by
    Lemma~\ref{lem:matchineq-a}, since $3(\tfrac{1}{16})^2 <
    \tfrac{9}{16}\cdot\tfrac{1}{16}$.
\end{proof}

\section{The lattice $\boldsymbol{\calL'}$}
\label{sec:mainthm}

In this section, we prove Theorem~\ref{thm:main}.

 {\renewcommand{\thetheorem}{\ref{thm:main}}
\begin{theorem} \statethmmain
\end{theorem}
\addtocounter{theorem}{-1}
}
\begin{proof} 
    First, we check that the vertices of~$\calL'$ are, indeed, $\omega$-clones.
    $\calB$ is trivially an $\omega$-clone. $\ofclone{\calM}$, $\Hferro$,
    $\ofclone{\ofclone{\calM}\cup\Hferro}$ and $\Iferro$, are $\omega$-clones by
    definition. Hence, so is the intersection $\ofclone{\calM}\cap \Hferro$.
    $\calP$, $\calPN$ and~$\SD$ are $\omega$-clones by Theorems
    \ref{thm:P-clone}, \ref{thm:PN-clone} and~\ref{thm:SDclone}, so their
    intersection $\SDP$ is also an $\omega$-clone.

    Next, we check the lattice structure. 
    
    We start with the strict inclusions (indicated by the solid lines in
    Figure~\ref{fig:lattice}).  
It is easy to see, using the definitions, that $\SD$, $\calP$ and $\calPN$ are strict subsets of~$\calB$.    
    $\SD$, $\calP$ and~$\calPN$ are pairwise-incomparable
    under~$\subseteq$ by Lemma~\ref{lem:SDPN-incomp}. $\SDP$ is a subset of
    $\SD$, $\calP$ and~$\calPN$ by definition; it is a strict subset because
    $\SD$, $\calP$ and~$\calPN$ are distinct. $\ofclone{\calM}$ and~$\Hferro$
    are $\subseteq$-incomparable by Lemma~\ref{lem:MHferro-incomp}.
    Consequently, $\ofclone{\calM}\cap\Hferro$ is a strict subset of both
    $\ofclone{\calM}$ and $\Hferro$ and $\ofclone{\ofclone{\calM}\cup\Hferro}$
    is a strict superset of both $\ofclone{\calM}$ and $\Hferro$. For the
    remaining two inclusions (which we do not know to be strict, as indicated by
    the dotted lines in Figure~\ref{fig:lattice}), $\Iferro\subseteq
    \ofclone{\calM}\cap\Hferro$ by Theorem~\ref{thm:Iferro-MHferro}. Also,
    $\ofclone{\calM} \subset \SDP$ by Theorem~\ref{thm:M-SDP} and $\Hferro
    \subset \SDP$ by Theorem~\ref{thm:Hferro-SDP}. Hence
    $\ofclone{\ofclone{\calM}\cup\SDP}\subseteq\SDP$ since $\SDP$ is an
    $\omega$-clone.

    Note that since the meet of any two clones is defined as their
    intersection the meets of any two $\omega$-clones from $\calL'$ are indeed
    as shown in Figure~\ref{fig:lattice} (this uses Theorem~\ref{thm:SDP}). We now show that also the joins of any
    two $\omega$-clones from $\calL'$ are as shown in Figure~\ref{fig:lattice}.
    Lemma~\ref{lem:SD-maximal} implies $\SD\vee\calP=\SD\vee\calPN=\calB$. By
    Corollary~\ref{cor:PjoinPN}, $\calP\vee\calPN=\calB$.
    $\ofclone{\calM}\vee\Hferro=\ofclone{\ofclone{\calM}\cup\Hferro}$ by
    definition.

    $\SD = \Ianti$ by Theorem~\ref{thm:SD-Ising}.  $\Iferro = \ofclone{\calE}$ by Corollary~\ref{cor:Iferro-E}.  

    Maximality of $\SD$, $\calP$ and~$\calPN$ in~$\calB$ is by
    Lemma~\ref{lem:SD-maximal}, Lemma~\ref{lem:P-max} and
    Corollary~\ref{cor:PN-max}, respectively.
    Maximality of $\SDP$ in~$\SD$ is by Theorem~\ref{thm:SDP-maximal}.
\end{proof}

\section{Clones of monotone functions}
\label{sec:monotone}

For $\bx,\by\in\{0,1\}^n$, $\bx\le\by$ denotes the fact that
$x_i\le y_i$ for all $i\in[n]$.
For any function $f(x_1,\ldots,x_n) \in \calB_n$,
let $\sim_f$ denote the equivalence relation on~$[n]$ given by
$i\sim_f j$ if and only if for every~$\bx\in\{0,1\}^n$, $f(\bx)=0$
whenever $x_i\ne x_j$.
Let $V_1,\ldots,V_\ell$ be the equivalence classes of $\sim_f$
and let $\tilde{f}$ be the function in $\calB_\ell$
defined as follows. For any $\bx \in \{0,1\}^\ell$,
construct $\by \in \{0,1\}^n$ as follows. For all $i\in [n]$, if $i\in V_j$, then set $y_i=x_j$.
Then $\tilde{f}(\bx) = f(\by)$.

\begin{lemma}\label{lem:tilde-function}
For any $f\in\calB$, 
$\tilde f\in\fclone f$  and $f\in\fclone{\tilde f}$.
\end{lemma}

\begin{proof}
The function $f$ is constructed from $\tilde{f}$ by introducing fictitious arguments and adding $\EQ$ factors.
The function $\tilde{f}$ is constructed from $f$ by summing out variables.
\end{proof}

\begin{definition}
(Definition of monotone and block-monotone functions.)

\begin{itemize}
\item
Given a function $f\in \calB_n$ and an index $i\in [n]$,
the
argument $x_i$ is said to be \emph{fictitious} in~$f$ if,  for all $\bx$ and $\bx'$ that differ only at position~$i$,
$f(\bx)=f(\bx')$.
\item
For any non-negative integer~$n$ and any
$\alpha\ge0$, the function $f(x_1,\ldots,x_n)\in\calB_n$ is said to be \emph{$\alpha$-monotone}
if, for every argument $x_i$ that is not fictitious in~$f$,  and every $\bx$ with $x_i=0$,
$\alpha f(\bx)\le f(x_1,\ldots,x_{i-1},\overline{x_i},x_{i+1},\ldots,x_n)$.
\item A function $f\in\calB$ is said to be \emph{block-$\alpha$-monotone} if $\tilde f$
is $\alpha$-monotone. 
\item  
$1$-monotone functions are called \emph{monotone} functions. 
Block-$1$-monotone functions are called \emph{block-monotone} functions.
\item  
The set of all block-monotone functions is denoted by $\Mon$, and the set of all 
block-$\alpha$-monotone functions is denoted by $\Mon_\alpha$.
\end{itemize}
\end{definition}

Note that if $\alpha\ge\beta$ then $\alpha
f(\bx)\le f(\by)$ implies $\beta f(\bx)\le f(\by)$. Thus, 
every $\alpha$-monotone function is $\beta$-monotone and $\Mon_\alpha\subseteq\Mon_\beta$.
We next show that, for any $\alpha\geq 1$, $\Mon_\alpha$ is an $\omega$-clone.
Since $\Mon_1=\Mon$, this implies that $\Mon$ is an $\omega$-clone.

\begin{theorem}\label{thm:monotone-clone} 
For any $\alpha \geq 1$, $\ofclone{\Mon_\alpha} = \Mon_\alpha$ and $\calB_0 \subseteq \Mon_\alpha$.
\end{theorem}

\begin{proof} 
Fix any $\alpha\geq 1$.
We will show that $\EQ \in \Mon_\alpha$ and that it is closed under the usual operations.
\begin{itemize}
\item  $\widetilde{\EQ}$ is the unary constant function $\widetilde{\EQ}(x)=1$.
Its only argument is fictitious. Thus, it is in $\Mon_\alpha$.

\item For closure under permuting arguments, 
suppose that $f\in \Mon_\alpha$ and that
$g$ is formed from~$f$
by permuting arguments.
Then $\tilde{g}$ is formed from $\tilde{f}$  by permuting arguments.
Since $\tilde{f}$ is $\alpha$-monotone, so is $\tilde{g}$, so
$g\in \Mon_\alpha$.

\item For closure under introducing fictitious arguments, let
$f\in\Mon_\alpha$ and define $h(\bx y) = f(\bx)$.  Then
the argument $y$ is in its own equivalence class in $\sim_h$
so $\tilde{h}(\bx y) = \tilde{f}(x)$.
Since $\tilde{f}$ is $\alpha$-monotone, and $y$ is fictitious, $\tilde{h}$ is $\alpha$-monotone.        

\item For closure under summation, let $f\in\Mon_\alpha$ and define
$h(\bx) = f(\bx0) + f(\bx1)$.  
Let $n+1$ be the arity of~$f$.
There are two possibilities. 
\begin{itemize}
\item
First, suppose that
 $x_{n+1}$ is equivalent to some other argument under $\sim_f$
(for convenience, assume that it is equivalent to $x_n$).
Then $f(x_1,\ldots,x_n,x_n) = h(x_1,\ldots,x_n)$
so $\tilde{f} = \tilde{h}$, and $\tilde{h}$ is $\alpha$-monotone because $\tilde{f}$ is.
\item
Otherwise, let $V_1,\ldots,V_{\ell+1}$ be the equivalence classes of $\sim_f$, where $V_{\ell+1}$ contains only
$x_{n+1}$.
Then $V_1,\ldots,V_\ell$ are the equivalence classes of $\sim_h$.
We claim that if the argument corresponding to $V_i$ is not fictitious in $\tilde{h}$ then it
is not fictitious in $\tilde{f}$. Then, since $\tilde{f}$ is $\alpha$-monotone,
changing the value of this argument increases the value of the function by a factor of $\alpha$, both when $x_{n+1}=0$
and when $x_{n+1}=1$. Thus, changing the value of the argument also inceases the value of $\tilde{h}$ by a factor of~$\alpha$,
and $\tilde{h}$ is $\alpha$-monotone. 
 \end{itemize}

\item For closure under products, let $f,g\in\Mon_\alpha$ and consider
$h(\bx)=f(\bx)\,g(\bx)$.   
Let $n$ be the arity of $h$, $f$ and $g$.
Suppose that $V_1,\ldots,V_\ell$ are the equivalence classes of $\sim_{f}$.
Consider an equivalence class $V_i$
and a string $\bx \in \{0,1\}^n$ which sets all arguments in $V_i$ to~$0$.
Let $\bx'$ be the string constructed from~$\bx$ by changing the value of the arguments in $V_i$ to~$1$.
Since $\tilde{f}$ is $\alpha$-monotone, there are two possibilities:
\begin{enumerate}
\item If the argument of $\tilde{f}$ corresponding to $V_i$ is fictitious then  $f(\bx')=f(\bx)$.
\item If the argument of $\tilde{f}$ corresponding to $V_i$ is not fictitious then  $f(\bx') \geq \alpha f(\bx)$.
\end{enumerate}
A similar comment applies to~$g$.
Now  let the equivalence classes of $\sim_h$ be $V'_1,\ldots, V'_{\ell'}$.
Consider some equivalence class $V'_i$ --- this is a union of $\sim_f$ classes and a union of $\sim_g$ classes.
Suppose that the argument corresponding to $V'_i$ is not fictitious in $\tilde{h}$.
We want to argue that there is at least one of the $\sim_f$ and $\sim_g$ classes corresponding to $V'_i$
that is not fictitious. To see this, suppose for contradiction that they are all fictitious.
Start with a string $\bx$ in which all of the arguments in $V'_i$ are  the same.
First consider changing, one-by-one all of the 
values of arguments in the $\sim_f$ classes corresponding to $V'_i$.
Since they are fictitious, this does not change the value of~$f$.
Similarly, changing-one-by-one all of the 
values in the $\sim_g$ classes corresponding to $V'_i$ does not change the value of~$g$.
So if $\bx'$ is the string derived from $\bx$ by  changing the spin of $V'_i$,
then $f(\bx')=f(\bx)$ and $g(\bx')=g(\bx)$. But this is a contradiction, since the argument corresponding to $V'_i$
is not fictitious in $\tilde{h}$.
Now suppose that $\bx$ takes value~$0$ on $V'_i$.
Changing the value of some $\sim_f$ or $\sim_g$ class inside $V'_i$ inceases the value of the function
by a factor of~$\alpha$. Changing each other $\sim_f$ or $\sim_g$ class inside $V'_i$
either leaves the value alone, or inceases it by another factor of~$\alpha$.
Hence, $\tilde{h}$ is $\alpha$-monotone. 

\item For closure under limits, let $f\in\calB_n$ and suppose that,
        for all integers $i>0$, there is some $g_i\in\Mon_\alpha$ such that
        $\|f-g_i\|_\infty < 2^{-i}$.  We must show that $f\in\Mon_\alpha$.

    There must be some equivalence
    relation~${\sim_g}$ on~$[n]$ such that ${\sim_{g_i}}={\sim_g}$ for
    infinitely many~$i$.  In fact, we may assume that
    ${\sim_{g_i}}={\sim_g}$ for all~$i$: if not, let $g'_1, g'_2,
    \dots$ be the subsequence of functions whose equivalence relation
    is~${\sim_g}$, note that $\|f-g'_i\|_\infty < 2^{-i}$ for all~$i$
    and use the sequence $g'_1, g'_2, \dots$ in place of $g_1, g_2,
    \dots$.

    Now, every equivalence class of~${\sim_f}$ is a union of
    equivalence classes of~${\sim_g}$. To see this suppose that
    $r\sim_g s$. For all~$i$ and all $\bx\in\Bools^n$ with
    $x_r\neq x_s$, we have $g_i(\bx)=0$. Therefore, $|f(\bx)|<2^{-i}$
    for all~$i$, so $f(\bx)=0$ and $r\sim_f s$.\footnote{Note that we
      do not necessarily have ${\sim_f}={\sim_g}$. For example, the
      $2$-monotone symmetric binary function
    $f(x,y)=[0,0,1]$ has just one equivalence class, but it is the limit
    of the $2$-monotone functions $f_i(x,y)=[0,2^{-i},1]$ for $i\geq 1$,
    and each of these functions has two equivalence classes.}

    Now, consider some argument~$x_j$ that is not fictitious in~$f$.  Let
    $\bx\in\{0,1\}^n$ be a tuple such that $x_j=0$ and $x_r=x_s$
    whenever $r\sim_f s$. Let $\by$ be the tuple with $y_s=x_s$ for all
    $s\not\sim_f j$ and $y_s=1$ for $s\sim_f j$.  We must show that
    $f(\by)\geq \alpha f(\bx)$.  This is trivial when $f(\bx)=0$ so we
    consider the case that $f(\bx) = \lambda > 0$.  Then, for all
    large enough~$i$, $g_i(\bx) > \lambda - 2^{-i} > 0$ so, by
    block-$\alpha$-monotonicity of~$g_i$,
    $g_i(\by) > \alpha(\lambda-2^{-i})$. So, for all large enough~$i$,
    $g(\by) > \alpha(\lambda-2^{-i}) - 2^{-i}$, so
    $g(\by)\geq\alpha\lambda$, as required.
    \end{itemize}
The proof that $\calB_0 \subseteq \Mon_\alpha$ is  straightforward  since a function $f\in \calB_0$
has no arguments, fictitious or otherwise.
\end{proof}

\section{Cardinality of the set of clones}
\label{sec:cardinality}

In this section we determine the cardinality of the lattices of functional and $\omega$-clones,
proving Theorem~\ref{thm:cardinality}.
  {\renewcommand{\thetheorem}{\ref{thm:cardinality}}
\begin{theorem} \statethmcardinality
\end{theorem}
\addtocounter{theorem}{-1}
}

Since $|\calB|=\beth_1$, we have $|\calL_f|,|\calL_\omega|\le\beth_2$. Therefore, we focus on
proving the inverse inequality. As every $\omega$-clone is also a functional clone, it 
suffices to prove that $|\calL_\omega|\ge\beth_2$. We construct a set of functions,
$\calF\subseteq\calB_2$ with $|\calF|=\beth_1$ that has the following property:  For any $\calG\subseteq\calF$, 
$\ofclone\calG\cap\calF=\calG$. This immediately implies that, for any 
$\calG_1,\calG_2\subseteq\calF$ 
with $\calG_1\ne\calG_2$
we have $\ofclone{\calG_1}\ne\ofclone{\calG_2}$. Therefore,
$|\calL_\omega|\ge 2^{|\calF|}=\beth_2$.

For any real $\alpha>2$, let $f_\alpha$ denote the binary function given by $f_\alpha(0,0)=1$,
$f_\alpha(0,1)=f_\alpha(1,0)=2$ and $f_\alpha(1,1)=2\alpha$. Let $\calF$ denote the set 
$\{f_\alpha\mid \alpha>3\}$. Note that $f_\alpha$ is 2-monotone for any $\alpha>2$. 
Therefore $\ofclone \calF\subseteq\Mon_2$.

 \begin{lemma}\label{lem:alpha-generation}
Let $\calG\subseteq \calF$ be a finite set and let $\beta=\min\{\alpha\mid f_\alpha\in\calG\}$.
If $f\in\fclone \calG$ is a binary  function such that $\tilde f=f$, then either 
$\gamma f\in \calG$ for some constant $\gamma$ or 
$f(0,1)\ge\min\left\{4,\frac{2+\beta}2\right\}f(0,0)$, or 
$f(1,0)\ge\min\left\{4,\frac{2+\beta}2\right\}f(0,0)$.
\end{lemma}

\begin{proof} 
Suppose~$f$ is a binary function in~$\fclone{\calG}$. As noted in the introduction to the paper, $f$ can be expressed
as 
\begin{equation}
\label{eq:expr1}
f(x,y)\;\ = \sum_{x_1,\ldots,x_k} \prod_{j=1}^t g'_j(x,y,x_1,\ldots,x_k),
\end{equation}
where $t$ and $k$ 
are non-negative integers and each function $g'_j$ is a $(k+2)$-ary function in~$\calA(\calG)$.
Recall that $\calA(\calG)$ is the  
closure of $\calG \cup \{\EQ\}$ under 
the introduction of fictitious arguments and
permuting arguments. Thus, every function $g'_j$ in~\eqref{eq:expr1}
is constructed from a binary function $h \in \calG \cup \{\EQ\}$ by introducing fictitious arguments and
permuting arguments. If $h =\EQ$ than $g'_j$ can be removed from the expression on the right-hand-side of~\eqref{eq:expr1}
without changing the function $f(x,y)$
by instead 
allowing re-use of variables.
(If $h$ forces $x_i$ and $x_j$ to be equal, then we can just remove all instances of $x_j$ and replace them with $x_i$
and we can also remove
 $x_j$ from the sum.)
Also, we can remove the fictitious arguments in the $g'_j$ functions, replacing each $g'_j$ with 
the corresponding binary function~$g_j\in \calG$.
Suppose that $f = \tilde f$ (so the variables~$x$ and~$y$ are in different $\sim_f$ classes and removing the $h=\EQ$
functions from~\eqref{eq:expr1} does not remove either~$x$ or~$y$).
Then, by these transformations, 
\eqref{eq:expr1} shows that
there are non-negative integers $s$ and $m$ so that
\begin{equation}
\label{eq:expr2} 
f(x,y)\;\  = \sum_{u_1,\ldots,u_m} \prod_{j=1}^s g_j(x_{j,1},x_{j,2}),
\end{equation}
where, for all $j\in [s]$, $g_j\in \calG$ and
$x_{j,1}$ and~$x_{j,2}$ are in $\{x,y,u_1,\ldots,u_m\}$
(though  $x_{j,1}$ and~$x_{j,2}$ may not necessarily be distinct).

 Without loss of generality we 
assume that there are $0\le p,q,r\le s$ such that functions $g_{j}$ for 
$0<j\le p$ involve both $x$ and $y$; functions $g_{j}$ for 
$p<j\le p+q$ involve $x$ but not $y$; functions $g_{j}$ for $p+q<j\le p+q+r$ 
involve $y$ but not $x$; and the remaining functions do not involve $x$ or $y$. Note that
since none of $x,y$ is fictitious, $p+q>0$ and $p+r>0$. For $x,y\in\Bools$ and 
$\bu\in\Bools^m$, let also
\begin{alignat*}{2}
    T_{xy}(x,y)\  &=\ \,\prod_{j=1}^p\  g_{j}(x_{j,1},x_{j,2})\,, \quad\quad&
    T_x(x,\bu)\;\  &= \prod_{j=p+1}^{p+q} g_{j}(x_{j,1},x_{j,2})\,,\\
    T_y(y,\bu)\  &=\!\!\! \prod_{j=p+q+1}^{p+q+r} g_{j}(x_{j,1},x_{j,2})\,,&
    T_0(\bu)\;\  &= \!\!\!\!\! \prod_{j=p+q+r+1}^s g_{j}(x_{j,1},x_{j,2})\,.
\end{alignat*}
Thus, 
\begin{equation*}
    f(x,y) \ \;= \!\! \sum_{\bu\in\Bools^m}\!\!
                        T_{xy}(x,y)\,T_x(x,\bu)\,T_y(y,\bu)\,T_0(\bu)\,.
\end{equation*}

If $p=1$ and $q=r=0$ then $\gamma f\in \calG$ for 
$1/\gamma =\sum_{\bu\in\Bools^m} T_0(\bu)$. If $p+q>1$ we have for any 
$y\in\Bools$ and $\bu\in\Bools^m$
\begin{equation*}
    4T_{xy}(0,y)\,T_x(0,\bu)\,T_y(y,\bu)\,T_0(\bu)\le
        T_{xy}(1,y)\,T_x(1,\bu)\,T_y(y,\bu)\,T_0(\bu),
\end{equation*}
because every $g_{j}\in \calG$ is 
$2$-monotone. This implies that $4f(0,y)\le f(1,y)$ for all~$y$ so, in particular, $4f(0,0)\leq f(1,0)$. Similarly, $4f(0,0)\leq f(0,1)$ if $p+r>1$. Therefore, 
the only remaining case is $p=0$ and $q=r=1$. 

Define $\alpha$ so that $g_{1}$ is the function $f_\alpha(x,u_1)$, where $f_\alpha\in \calG$. Then
{\allowdisplaybreaks
\begin{align*}
    f(1,0) &= \sum_{\bu\in\Bools^m}f_\alpha(1,u_1)\,T_y(0,\bu)\,T_0(\bu)\\
        &= \sum_{\bu'\in\Bools^{m-1}}f_\alpha(1,0)\,T_y(0,0,\bu')\,T_0(0,\bu')
               +\sum_{\bu'\in\Bools^{m-1}}f_\alpha(1,1)\,T_y(0,1,\bu')\,T_0(1,\bu')\\
        &= 2\sum_{\bu'\in\Bools^{m-1}}f_\alpha(0,0)\,T_y(0,0,\bu')\,T_0(0,\bu')
               +\alpha\sum_{\bu'\in\Bools^{m-1}}
                          f_\alpha(0,1)\,T_y(0,1,\bu')\,T_0(1,\bu')\\
        &= 2f(0,0)+(\alpha-2)\sum_{\bu'\in\Bools^{m-1}}
                                 f_\alpha(0,1)\,T_y(0,1,\bu')\,T_0(1,\bu')\\
        &\ge 2f(0,0)+(\alpha-2)\sum_{\bu'\in\Bools^{m-1}}\frac12
               \big(f_\alpha(0,0)\,T_y(0,0,\bu')\,T_0(0,\bu')
                   + f_\alpha(0,1)\,T_y(0,1,\bu')\,T_0(1,\bu')\big)\\
        &= \frac{2+\alpha}2 f(0,0)\,.
\end{align*}}
The inequality here holds because all the functions involved are monotone and therefore
\begin{equation*}
    f_\alpha(0,0)\,T_y(0,0,\bu')\,T_0(0,\bu')
        \le f_\alpha(0,1)\,T_y(0,1,\bu')\,T_0(1,\bu')\,.\end{equation*}
The result follows by the choice of~$\beta$.
\end{proof}

  We can now prove Theorem~\ref{thm:cardinality}.

\begin{proof}
As we observed at the beginning of the section, to prove that $|\calL_\omega|\ge\beth_2$, it suffices to show that for 
any $\calG_1,\calG_2\subseteq \calF$ with $\calG_1 \neq \calG_2$,
$\ofclone{\calG_1}\ne\ofclone{\calG_2}$. Suppose 
$f_\alpha\in \calG_1\setminus \calG_2$. We show that $f_\alpha\not\in \ofclone{\calG_2}$.
The function $f_\alpha$ is symmetric and binary and it satisfies
$f_\alpha = \widetilde{f_\alpha}$ so take
 any binary symmetric function $f\in\ofclone{\calG_2}$ with $\tilde f=f$. 
 We will show that $f\neq f_\alpha$.
 
By the definition of $\omega$-clone,
there is a finite
set $\calG''\subseteq \fclone{\calG_2}$ such that $f=\lim_{n\to\infty} h_n$ where 
each $h_n$ is a binary function in~$\fclone{\calG''}$. 
Each of the finitely many functions in $\calG''$ can be 
written as a finite sum of a product of (finitely many) functions in $\calA(\calG_2)$.
Let $\calG'$ be the finite set of functions in $\calG_2$ which correspond to the relevant functions in $\calA(\calG_2)$.
Then clearly each $h_n \in \fclone{\calG'}$.

Since $\tilde f = f$ there are at most finitely many $n$ such that $\widetilde{h_n} \neq h_n$ --- so
we will remove these from the sequence of functions $\{h_n\}$  
and of course it is still true that $f=\lim_{n\to\infty} h_n$.
 
Let $\beta=\min\{\mu\mid f_\mu\in \calG'\}$. By Lemma~\ref{lem:alpha-generation},
either there is a $\gamma_n$ such that $\gamma_n h_n \in \calG'$ or $h_n(0,1)\ge\min\left\{4,\frac{2+\beta}2\right\}h_n(0,0)$, or 
$h_n(1,0)\ge\min\left\{4,\frac{2+\beta}2\right\}h_n(0,0)$.

If there are infinitely many
$h_n$ such that $\gamma_n h_n\in \calG'$ for some constants $\gamma_n$ then, since 
$\calG'$ is finite, infinitely many of  the $\gamma_n h_n$ functions are equal. 
Since $f(0,0)=1$, $\lim_{n\to\infty}\gamma_n=1$ for such functions. Therefore, each of
$\gamma_n h_n$ is equal to $f$. Thus $f\in \calG'$ in this case. Clearly, $f\neq f_\alpha$ since $f_\alpha \not \in \calG_2$ so $f_\alpha \not\in \calG'$.

If there are finitely many $h_n$ with $\gamma_nh_n \in \calG'$ then
we can remove these from the sequence of functions $\{h_n\}$  
and  as before it is still true that $f=\lim_{n\to\infty} h_n$. From now on, we therefore
assume that none of the functions~$h_n$ belong   to $\calG'$.
Suppose that $h_n(1,0)\ge\min\left\{4,\frac{2+\beta}2\right\}h_n(0,0)$ for infinitely many 
$h_n$. Then $f(1,0)\ge\min\left\{4,\frac{2+\beta}2\right\}f(0,0)$, and $f\not\in \calF$
so clearly $f\neq f_\alpha$.
The case when there are infinitely many $h_n$ with 
$h_n(1,0)\ge\min\left\{4,\frac{2+\beta}2\right\}h_n(0,0)$ is similar.
\end{proof}

\section{Ternary functions}
\label{sec:ternary}

In this section, we prove Theorem~\ref{thm:ternary}, which we restate here for convenience.

Recall that $\calS_3=\{\calB_3, [\SD]_3, [\calP]_3, [\calPN]_3, [\SDP]_3,$ $[\ofclone{M}]_3, [\Hferro]_3, [\Iferro]_3\}$.  

{\renewcommand{\thetheorem}{\ref{thm:ternary}}
\begin{theorem} \statethmTernary
\end{theorem}
\addtocounter{theorem}{-1}
}

\begin{proof}
The two collapses are proved in Theorem~\ref{thm:collapse1} and Theorem~\ref{thm:collapse2}, respectively.

Trivially, $[\SD]_3$, $[\calP]_3$, and $[\calPN]_3$ are strict subsets of $[\calB]_3$.
We now show that these three sets are distinct. Recall the binary
functions $f,g,h$ from the proof of Lemma~\ref{lem:SDPN-incomp} with
$f\in\SD\setminus (\calP\cup\calPN)$,
$g\in\calP\setminus(\SD\cup\calPN)$ and $h\in\calPN\setminus(\SD\cup\calP)$.
Define $f'(x,y,z)=f(x,y)$,
$g'(x,y,z)=g(x,y)$, and $h'(x,y,z)=h(x,y)$. From the proof
Lemma~\ref{lem:SDPN-incomp}, $\fhat = [\tfrac12, 0, -\tfrac12]$, $\ghat =
[1,\tfrac12,0]$ and $\hhat = [1,-\tfrac12,0]$. By
Lemma~\ref{lem:fops}\pref{op-fict}, $\widehat{f'}(x,y,0)=\fhat(x,y)$ and
$\widehat{f'}(x,y,1)=0$, and similarly for $\widehat{g'}$ and
$\widehat{h'}$.
It is now easy to verify that
$f'\in[\SD]_3\setminus ([\calP]_3\cup[\calPN]_3)$,
$g'\in[\calP]_3\setminus([\SD]_3\cup[\calPN]_3)$ and $h'\in[\calPN]_3\setminus([\SD]_3\cup[\calP]_3)$.

Since $[\SD]_3$, $[\calP]_3$, $[\calPN]_3$ are distinct, we have
$[\SDP]_3\subset[\SD]_3, [\calP]_3, [\calPN]_3$.

Finally, we show that $[\Hferro]_3\subset[\SDP]_3$. 
The non-strict inclusion $[\Hferro]_3\subseteq [\SDP]_3$ comes from the fact that $\Hferro\subseteq \SD$
(from Theorem~\ref{thm:main}).
To get the strict inclusion, we will  exhibit a ternary function that is in $[\SDP]_3$, but
is not in $[\LSM]_3$ (so is not in $[\Hferro]_3$ since $\Hferro \subseteq \LSM$, as discussed in the proof of
Theorem~\ref{thm:Hferro-SDP}, and hence $[\Hferro]_3\subseteq[\LSM]_3$).

Consider
the $3$-ary self-dual function $f$ defined by $f(0,0,0) = 6$, $f(0,0,1) = 4$,
and $f(0,1,0) = f(1,0,0)=5$. It can be verified that $f\in\calP$ and thus
$f\in[\SDP]_3$. We use Lemma~\ref{lem:topkis} to show that $f\not\in\LSM$: for
$2$-pinning $g(x,y)=f(x,y,1)$ we have $g(0,0)\,g(1,1)=4\cdot 6=24<25=5\cdot 5=g(0,1)\,g(1,0)$.
\end{proof}

\begin{figure}
    \begin{center}
    \begin{tikzpicture}[scale=1,node distance = 1.5cm]
    \tikzstyle{vertex}=[fill=black, draw=black, circle, inner sep=2pt]
    \tikzstyle{dist}  =[fill=white, draw=black, circle, inner sep=2pt]

        \node[dist]   (u1) at (0,3.3) [label=180:$u_1$] {};
        \node[dist]   (u2) at (0,2) [label=180:$u_2$] {};
        \node[dist]   (u3) at (0,0.7) [label=180:$u_3$] {};
        \node[vertex] (v1) at (1.5,3.3)                 {};
        \node[vertex] (v2) at (1.5,2)                 {};
        \node[vertex] (v3) at (1.5,0.7)                 {};
        \node[vertex] (w1) at (4.5,3.3)                 {};
        \node[vertex] (w2) at (3,2)                   {}; 
        \node[vertex] (w3) at (4.5,0.7)                 {};

        \draw (u1)-- (v1) -- (w1);
        \draw (u2)-- (v2) -- (w2);
        \draw (u3)-- (v3) -- (w3);
        
        \draw (w1)-- node [below ] {$\ \ d/a$} (w2);
        \draw (w2)-- node [above ] {$\ \ b/a$} (w3);
        \draw (w3)-- node [right] {$c/a$} (w1);
 
        \node[vertex] (a) at (6.5,3)                 {};
        \node[vertex] (b) at (6.5,1)                 {};
        \draw (a)-- node [right] {$a$} (b);

    \end{tikzpicture}
    \end{center}
    \caption{The match-circuit used in the proof of
      Theorem~\ref{thm:collapse1}.
      Every edge has weight~$1$ unless
      otherwise indicated.  }
    \label{fig:gadget}
\end{figure}
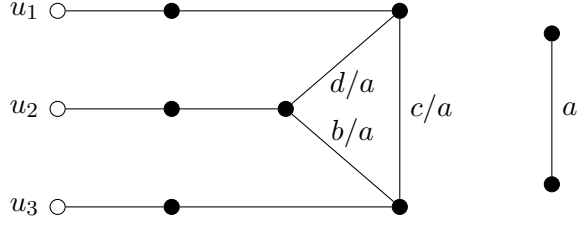

\begin{theorem}\label{thm:collapse1}
    $[\SDP]_3 = [\ofclone{\calM}]_3$.
\end{theorem}

\begin{proof} 
By Theorem~\ref{thm:main}, $\ofclone{\calM}\subseteq \SDP$ and thus
$[\ofclone{\calM}]_3\subseteq [\SDP]_3$. It remains to show the other inclusion,
$[\SDP]_3\subseteq[\ofclone{\calM}]_3$. 
Consider a $3$-ary function $f(x,y,z)\in \SDP$.  If $f$ is the constant zero
function then $f\in[\ofclone{\calM}]_3$ 
since  $\fhat$ is also the constant zero function, and it can be implemented by
a match-circuit with three terminals and a disjoint triangle.

If $f$ is not the constant zero function, then $\fhat(0,0,0)>0$.
There are values $a>0$ and 
$b,c,d\geq 0$ such that $\fhat(0,0,0)=a$,
$\fhat(0,1,1)=b$,
$\fhat(1,0,1)=c$,
$\fhat(1,1,0)=d$ and (by Lemma~\ref{lem:oddSD}) $\fhat(\bx)=0$ when
$\wt{\bx}$~is odd.  It is easily verified that, for $b,c,d>0$, $\fhat$~is implemented
by the match-circuit shown in
Figure~\ref{fig:gadget}.
Definition~\ref{def:MC} does not allow zero-weight edges but we can
implement~$\fhat$ in cases where some of $b$, $c$ and~$d$ are zero by
deleting the corresponding edge or edges from the match-circuit in Figure~\ref{fig:gadget}.
Hence, $f\in [\calM]_3$, so it is in  $[\ofclone{\calM}]_3$.
\end{proof}

\begin{theorem}\label{thm:collapse2}
  $[\Hferro]_3 = [\Iferro]_3$.
\end{theorem}
\begin{proof}
By Theorem~\ref{thm:main}, $\Iferro\subseteq\Hferro$ and thus
$[\Iferro]_3\subseteq[\Hferro]_3$. It remains to show the other inclusion,
$[\Hferro]_3\subseteq[\Iferro]_3$. In fact, we prove something stronger, namely
that $[\SDP\cap\LSM]_3\subseteq[\Iferro]_3$. By the proof of
Theorem~\ref{thm:Hferro-SDP}, $\Hferro\subseteq\SDP\cap\LSM$ so
the required inclusion follows.

An arbitrary $3$-ary function $f\in\SDP$ is given by $f(0,0,0)=f(1,1,1)=\lambda$,
$f(0,0,1)=f(1,1,0)=a$, $f(0,1,0)=f(1,0,1)=b$, and $f(1,0,0)=f(0,1,1)=c$. 

For $f$ to be in $\LSM$, by Lemma~\ref{lem:topkis} (and in particular the fact
that the necessary condition of Lemma~\ref{lem:topkis} holds even for
non-permissive functions, as discussed in
Section~\ref{sec:self-dual}), the following functions must also be in $\LSM$:
$f_1,\dots,f_6\in\LSM$ where $f_1(x,y)=f(0,x,y)$, $f_2(x,y)=f(x,0,y)$, 
$f_3(x,y)=f(x,y,0)$, $f_4(x,y)=f(1,x,y)$, $f_5(x,y)=f(x,1,y)$ and
$f_6(x,y)=f(x,y,1)$. This gives \begin{equation}\label{ineq3} \lambda c\geq ab\,,
\quad \lambda b\geq ac \quad\mbox{and}\quad \lambda a \geq bc\,. \end{equation}

Assume that $f$ is permissive. Without loss of generality (by scaling since
$\calB_0\subseteq\Iferro$), let $\lambda=1$. Let $g(x,y,z)=
\hIsing{2}{\lambda_1}(x,y)\,\hIsing{2}{\lambda_2}(x,z)\,\hIsing{2}{\lambda_3}(y,z)$,
where $\lambda_1=\sqrt{ bc/a}$, $\lambda_2 = \sqrt{ ac/b}$ and $\lambda_3=\sqrt{
ab/c}$. By~(\ref{ineq3}), $\lambda_1,\lambda_2,\lambda_3\leq 1$, and hence
$g\in[\Iferro]_3$. We now verify that $f=g$. By the definition of $g$,
$g(0,0,0)=g(1,1,1)=1$, $g(0,0,1)=g(1,1,0)=1\cdot \lambda_2\cdot \lambda_3 = a$,
$g(0,1,0)=g(1,0,1)=\lambda_1\cdot 1\cdot \lambda_3 = b$, and
$g(1,0,0)=g(0,1,1)=\lambda_1\cdot\lambda_2\cdot 1=c$. 

It remains to deal with non-permissive $f$. If $\lambda=0$ then~(\ref{ineq3})
implies that at most one of $a$, $b$, and $c$ is non-zero. If all three are zero
then $f$ is the constant zero function and thus trivially in $[\Iferro]_3$ since
$\calB_0\subseteq\Iferro$. Otherwise, let $a>0$ and $b=c=0$; the other two cases
are symmetric. Since $\fhat(1,1,0)=-a/4<0$, $f\not\in\SDP$, a contradiction. If
$\lambda>0$ then, by scaling, let $\lambda=1$. The inequalities~\eqref{ineq3}~imply that at most one of $a$, $b$, and $c$ is non-zero. If all
three are zero then $f(x,y,z)=\EQ(x,y)\,\EQ(y,z)$. Otherwise, let $a>0$ and
$b=c=0$; the other two cases are symmetric. In this case $f(0,0,0)=f(1,1,1)=1$,
$f(0,0,1)=f(1,1,0)=a$ and $f(x,y,z)=0$ otherwise. Thus
$f(x,y,z)=\EQ(x,y)\,\hIsing{2}{a}(y,z)$.  Now, because $f\in\LSM$, we
must have $f(0,0,0)\,f(1,1,1) = 1\geq f(0,0,1)\,f(1,1,0)=a^2$, so
$a\leq 1$ and $f\in[\Iferro]_3$. 
\end{proof}

\bibliographystyle{plain}
\bibliography{\jobname}

\begin{thebibliography}{10}

\bibitem{AHS}
David~H. Ackley, Geoffrey~E. Hinton, and Terrence~J. Sejnowski.
\newblock A learning algorithm for {B}oltzmann machines.
\newblock {\em Cognitive Science}, 9(1):147--169, 1985.

\bibitem{LSM}
Andrei~A. Bulatov, Martin~E. Dyer, Leslie~Ann Goldberg, Mark Jerrum, and Colin
  McQuillan.
\newblock The expressibility of functions on the {B}oolean domain, with
  applications to counting {CSP}s.
\newblock {\em J. {ACM}}, 60(5):32, 2013.

\bibitem{CG2014:matchgates}
Jin-Yi Cai and Aaron Gorenstein.
\newblock Matchgates revisited.
\newblock {\em Theory of Computing}, 10(7):167--197, 2014.

\bibitem{CLX2008:Fib}
Jin-Yi Cai, Pinyan Lu, and Mingji Xia.
\newblock Holographic Algorithms by {F}ibonacci gates and Holographic
  Reductions for Hardness
\newblock In {\em Proceedings of 49th {A}nnual IEEE {S}ymposium on
  {F}oundations of {C}omputer {S}cience (FOCS 2008)}, pages
644--653. IEEE Computer Society, 2008.

\bibitem{planarCSP}
Jin-Yi Cai, Pinyan Lu, and Mingji Xia.
\newblock Holographic algorithms with matchgates capture precisely tractable
  planar {$\#{\rm CSP}$}.
\newblock In {\em Proceedings of 51st {A}nnual IEEE {S}ymposium on {F}oundations of
  {C}omputer {S}cience ({FOCS} 2010)}, pages 427--436. IEEE Computer Society, 2010.

\bibitem{ApproxCSP}
Xi~Chen, Martin Dyer, Leslie~Ann Goldberg, Mark Jerrum, Pinyan Lu, Colin
  McQuillan, and David Richerby.
\newblock The complexity of approximating conservative counting {CSP}s.
\newblock {\em J. Comput. System Sci.}, 81(1):311--329, 2015.

\bibitem{Cipra}
Barry~A. Cipra.
\newblock An introduction to the {I}sing model.
\newblock {\em Amer. Math. Monthly}, 94(10):937--959, 1987.

\bibitem{DlCC}
Gemma De~las Cuevas and Toby~S. Cubitt.
\newblock Simple universal models capture all classical spin physics.
\newblock {\em Science}, 351(6278):1180--1183, 2016.

\bibitem{deWolf08:brief}
Ronald de~Wolf.
\newblock A brief introduction to {F}ourier analysis on the {B}oolean cube.
\newblock {\em Theory of Computing, Graduate Surveys}, 1:1--20, 2008.

\bibitem{Geiger}
David Geiger.
\newblock Closed systems of functions and predicates.
\newblock {\em Pacific J. Math.}, 27:95--100, 1968.

\bibitem{FerroIsing}
Leslie~Ann Goldberg and Mark Jerrum.
\newblock The complexity of ferromagnetic {I}sing with local fields.
\newblock {\em Combinatorics, Probability {\&} Computing}, 16(1):43--61, 2007.

\bibitem{FerroPotts}
Leslie~Ann Goldberg and Mark Jerrum.
\newblock Approximating the partition function of the ferromagnetic {P}otts
  model.
\newblock {\em J. {ACM}}, 59(5):25, 2012.

\bibitem{planartutte}
Leslie~Ann Goldberg and Mark Jerrum.
\newblock Inapproximability of the {T}utte polynomial of a planar graph.
\newblock {\em Computational Complexity}, 21(4):605--642, 2012.

\bibitem{PNAS}
Leslie~Ann Goldberg and Mark Jerrum.
\newblock A complexity classification of spin systems with an external field.
\newblock {\em Proc. Natl. Acad. Sci. USA}, 112(43):13161--13166, 2015.

\bibitem{Grimmett}
Geoffrey Grimmett.
\newblock Potts models and random-cluster processes with many-body
  interactions.
\newblock {\em J. Statist. Phys.}, 75(1-2):67--121, 1994.

\bibitem{JS1993:Ising}
Mark Jerrum and Alistair Sinclair.
\newblock Polynomial-time approximation algorithms for the {I}sing model.
\newblock {\em SIAM Journal on Computing}, 22(5):1087--1116, 1993.

\bibitem{LeRouxBengio}
Nicolas Le~Roux and Yoshua Bengio.
\newblock Representational power of restricted {B}oltzmann machines and deep
  belief networks.
\newblock {\em Neural Comput.}, 20(6):1631--1649, 2008.

\bibitem{ODonnell14:book}
Ryan O'Donnell.
\newblock {\em Analysis of Boolean Functions}.
\newblock Cambridge University Press, 2014.

\bibitem{Post}
Emil~L. Post.
\newblock {\em The {T}wo-{V}alued {I}terative {S}ystems of {M}athematical
  {L}ogic}.
\newblock Annals of Mathematics Studies, no. 5. Princeton University Press,
  Princeton, N. J., 1941.

\bibitem{Top1978:submod}
Donald~M. Topkis.
\newblock Minimizing a submodular function on a lattice.
\newblock {\em Operations Research}, 26(2):305--321, 1978.

\bibitem{Val2008:Holant}
Leslie~G. Valiant.
\newblock Holographic Algorithms.
\newblock {\em SIAM Journal on Computing}, 37(5):1565--1594, 2008.

\bibitem{VdWaerden}
B.~L. {van der Waerden}.
\newblock {Die lange Reichweite der regelm{\"a}{\ss}igen Atomanordnung in
  Misch\-kristallen}.
\newblock {\em Zeitschrift fur Physik}, 118:473--488, July 1941.

\end{thebibliography}

\appendix

\section{Fourier transforms}
\label{app:Fourier}

In this appendix, we prove Lemmas~\ref{lem:fops}--\ref{lem:ft-parity}.

 {\renewcommand{\thetheorem}{\ref{lem:fops}}
\begin{lemma} \statelemfops
\end{lemma}
\addtocounter{theorem}{-1}
}

\begin{proof} 
\begin{enumerate}[(i)]
 
\item In the following, the third equality is because permuting a
    Boolean vector doesn't change its Hamming weight and the fourth
    equality is reordering the terms of the sum.
    \begin{align*}
        \widehat{f^{\pi}}(\bx)
            &= \frac1{2^k}\sum_{\bw\in \Bools^k}
                   (-1)^{\wt{\bx\wedge\bw}}f^{\pi}(\bw) \\
            &= \frac1{2^k}\sum_{\bw\in \Bools^k}
                   (-1)^{\wt{\bx\wedge\bw}}f(\pi(\bw)) \\
            &= \frac1{2^k}\sum_{\bw\in \Bools^k}
                   (-1)^{\wt{\pi(\bx)\wedge\pi(\bw)}}f(\pi(\bw)) \\
            &= \frac1{2^k}\sum_{\bw\in \Bools^k}
                   (-1)^{\wt{\pi(\bx)\wedge\bw}}f(\bw) \\
            &= \fhat(\pi(\bx))\,.
    \end{align*}

\item For any $z\in\Bools$,
    \begin{align*}
        \hhat(\bx z)
            &=\frac1{2^{k+1}} \sum_{\bw\in\Bools^k}
                  (-1)^{\wt{\bx z\wedge \bw0}} f(\bw) +
              \frac1{2^{k+1}} \sum_{\bw\in\Bools^k}
                  (-1)^{\wt{\bx z\wedge \bw1}} f(\bw) \\
            &=\frac1{2^{k+1}} \sum_{\bw\in\Bools^k}
                  (-1)^{\wt{\bx\wedge\bw}}\; f(\bw)\;
                  \big((-1)^{\wt{z\wedge 0}} + (-1)^{\wt{z\wedge 1}}\big) \\
            &= \tfrac12\fhat(\bx)\,\big(1 + (-1)^{\wt{z}}\big) \\
            &=\begin{cases}
                  \ \fhat(\bx) &\text{if }z=0 \\
                  \ 0            &\text{if }z=1\,.
              \end{cases}
    \end{align*}

\item For any $\bx$,
    \begin{align*}
        \hhat(\bx)
            &= \frac{1}{2^{k-1}}\sum_{\bw\in \Bools^{k-1}}
                   (-1)^{\wt{\bx\wedge\bw}}
                   \big(f(\bw0) + f(\bw1)\big) \\
            &= \frac{1}{2^{k-1}}\sum_{\bw\in \Bools^{k-1}}
                   (-1)^{\wt{\bx0\wedge\bw0}}f(\bw0) +
               \frac{1}{2^{k-1}}\sum_{\bw\in \Bools^{k-1}}
                   (-1)^{\wt{\bx0\wedge\bw1}}f(\bw1) \\
            &= \frac{1}{2^{k-1}}\sum_{\bw\in \Bools^k}
                   (-1)^{\wt{\bx0\wedge\bw}}f(\bw) \\
            &= 2\fhat(\bx0)\,.
    \end{align*}

\item Noting that $(-1)^{a-b} = (-1)^{a+b}$, we have
    \begin{align*}
        \hhat(\bx)
            &= \frac1{2^k}\sum_{\bw\in\Bools^k}
                   (-1)^{\wt{\bx\wedge\bw}} f(\overline{\bw}) \\
            &= \frac1{2^k}\sum_{\bw\in\Bools^k}
                   (-1)^{\wt{\bx\wedge\overline{\bw}}} f(\bw) \\
            &= \frac1{2^k}\sum_{\bw\in\Bools^k}
                   (-1)^{\wt{\bx}-\wt{\bx\wedge\bw}} f(\bw) \\
            &= \frac1{2^k}(-1)^{\wt{\bx}}\sum_{\bw\in\Bools^k}
                   (-1)^{\wt{\bx\wedge\bw}} f(\bw) \\
            &= (-1)^{\wt{\bx}}\fhat(\bx)\,.
    \end{align*}

\item
    \begin{equation*}
        \|\ghat - \fhat\|_\infty
            = \bigg\| \frac{1}{2^k}\sum_{\bw\in\Bools^k}
                (-1)^{\wt{\bx\wedge\bw}} \big(g(\bw) - f(\bw)\big)\bigg\|_\infty
            \!<\ \, \frac1{2^k}2^k\epsilon = \epsilon\,.\qedhere
    \end{equation*}
    
    \item Suppose $k=0$ so $f()=c$.
    Consider the unary function $h$ defined by $h(0)=h(1)=c/2$.
    Then $f() = h(0)+h(1)$ so by item~\ref{op-sum},
    $\fhat() = 2 \hhat(0) = c$. 
    
\end{enumerate}
\end{proof}

{\renewcommand{\thetheorem}{\ref{lem:ft-hIsing}}
\begin{lemma} \statelemfthIsing
\end{lemma}
\addtocounter{theorem}{-1}
}

\begin{proof} 
    We have
    \begin{align*}
        2^k\hIsinghat{k}{\lambda}(\bx)
            &= \sum_{\bw\in\Bools^k} (-1)^{\wt{\bx\land\bw}}
                   \hIsing{k}{\lambda}(\bw) \\
            &= \big( 1-\lambda\big)\big(1 + (-1)^{\wt{\bx}}\big)
               + \lambda\sum_{\bw\in\Bools^k} (-1)^{\wt{\bx\land\bw}} \\
            &= \big(1-\lambda\big)\big(1 + (-1)^{\wt{\bx}}\big)
               + \lambda \, 2^{k-\wt{\bx}}
                  \sum_{\bu\in\Bools^{\wt{\bx}}} (-1)^{\wt{\bu}}\,,
    \end{align*}
    where we adopt the convention that $\Bools^0$ contains exactly one
    tuple, which has Hamming weight~$0$.  This means the sum evaluates
    to~$1$ if $\bx = \bzero$ and to zero, otherwise.
\end{proof}

{\renewcommand{\thetheorem}{\ref{lem:ft-parity}}
\begin{lemma} \statelemftparity
\end{lemma}
\addtocounter{theorem}{-1}
}
\begin{proof}
    The first two equalities are straightforward from the definition.
    Let $x\in\Bools^n\setminus\{\bzero,\bone\}$.
    \begin{align*}
        \parityhatk(\bx)
            &= \frac{1}{2^n} \sum_{\bw\in\Bools^n} (-1)^{\wt{\bw\wedge\bx}}\,
                   \parityk(\bx)\\
            &= \frac{1}{2^n}\sum_{\wt{\bw}\text{ even}} (-1)^{\wt{\bw\wedge\bx}} +
               \frac{\lambda}{2^n}\sum_{\wt{\bw}\text{ odd}} (-1)^{\wt{\bw\wedge\bx}}.
    \end{align*}
    Suppose without loss of generality that $x_1=0$ and $x_2=1$. Then,
    for every $\bw\in\Bools^n$, the tuples $\bw=(w_1,\dots,w_n)$ and
    $\bw'=(\overline{w_1},\overline{w_2},w_3, \dots, w_n)$ have the
    same parity, but $\wt{\bw\wedge\bx}\ne
    \wt{\bw'\wedge\bx}$. Therefore
    \begin{equation*}
        \sum_{\wt{\bw}\text{ even}} (-1)^{\wt{\bw\wedge\bx}}
            \ \ = \sum_{\wt{\bw}\text{ odd}} (-1)^{\wt{\bw\wedge\bx}}=0\,.\qedhere
    \end{equation*}
\end{proof}

 \end{document}